\documentclass[journal]{IEEEtran} 

\usepackage[colorlinks=true, allcolors=blue]{hyperref}

\usepackage[utf8]{inputenc} 
\usepackage[T1]{fontenc}    
\usepackage{hyperref}       
\usepackage{url}            
\usepackage{booktabs}       
\usepackage{amsfonts}       
\usepackage{nicefrac}       
\usepackage{microtype}      
\usepackage{xcolor}         

\usepackage{hyperref}
\usepackage{amsmath, mathtools, amsfonts, bm, amssymb, amsthm,graphicx, caption, subcaption,multirow}
\usepackage{algorithm} 

\usepackage{algpseudocode}

\usepackage{bm, bbm}
\usepackage{color}

\begin{document}

\newcommand{\x}{\bm{x}}
\newcommand{\g}{\bm{g}}
\newcommand{\G}{\bm{G}}
\newcommand{\X}{{\bm{X}}}
\newcommand{\s}{\bm{s}}
\renewcommand{\S}{\mathcal{S}}
\newcommand{\vv}{\bm{v}}
\renewcommand{\c}{c}

\newcommand{\e}{\bm{e}}
\newcommand{\y}{\bm{y}}
\newcommand{\w}{\bm{w}}
\newcommand{\U}{{\bm{U}}}
\newcommand{\R}{{\bm{R}}}
\newcommand{\M}{\bm{M}}
\renewcommand{\H}{\bm{H}}

\newcommand{\I}{\bm{I}}
\newcommand{\Y}{\bm{Y}}
\newcommand{\uhat}{{\bm{\u}}}
\newcommand{\Zhat}{{\bm{\hat{Z}}}}
\newcommand{\Z}{{\bm{{Z}}}}
\newcommand{\D}{{\bm{D}}}
\newcommand{\Q}{{\bm{Q}}}
\newcommand{\F}{{\bm{F}}}
\newcommand{\iidsim}{\stackrel{\mathrm{iid}}{\thicksim }}
\newcommand{\n}{{\cal{N}}}
\newcommand{\indepsim}{\stackrel{\mathrm{indep.}}{\thicksim }}
\newcommand{\SE}{\mathrm{SD}}  
\newcommand{\dist}{\mathrm{dist}}
\newcommand{\tta}{\ta^{\text{trunc}}}
\newcommand{\A}{\bm{A}}
\newcommand{\full}{{\mathrm{full}}}

\newcommand{\Span}{\mathrm{span}}
\newcommand{\rank}{\mathrm{rank}}
\newcommand{\evdeq}{\overset{\mathrm{EVD}}=} 
\newcommand{\svdeq}{\overset{\mathrm{SVD}}=} 
\newcommand{\qreq}{\overset{\mathrm{QR}}=} 
\newcommand{\bi}{\begin{itemize}} \newcommand{\ei}{\end{itemize}}
\newcommand{\ben}{\begin{enumerate}} \newcommand{\een}{\end{enumerate}}
\newcommand{\vs}{\vspace{0.1in}}
\newcommand{\vsl}{\vspace{0.05in}}
\newcommand{\vsm}{\vspace{-0.1in}}

\renewcommand{\implies}{\Rightarrow}


\newcommand{\cblue}{\color{blue}}
\newcommand{\cbl}{\color{black}}
\newcommand{\cred}{}
\newcommand{\skipit}{ }

\newcommand{\trace}{\mathrm{trace}}

\newcommand{\qfull}{q_\full}
\newcommand{\sub}{{\mathrm{sub}}}

\newcommand{\bbf}{\mathbb{F}}
\newcommand{\dsW}{\mathds{W}}
\newcommand{\dsZ}{\mathds{Z}}

\newcommand{\mtx}[1]{\mathbf{#1}}
\newcommand{\vct}[1]{\mathbf{#1}}
\newcommand{\abs}[1]{\left|#1\right|}
\newcommand{\p}{\bm{p}}
\renewcommand{\j}{\bm{j}}
\newcommand{\uu}{{\bm{u}}}
\newcommand{\tot}{\mathrm{tot}}
\renewcommand{\a}{\bm{a}}
\newcommand{\ta}{\bm{\tilde{a}}}
\newcommand{\h}{\bm{h}}
\renewcommand{\b}{\bm{b}}
\newcommand{\B}{\bm{B}}
\newcommand{\V}{{\bm{V}}}
\newcommand{\W}{{\bm{W}}}

\renewcommand{\aa}{\bm{a}}

\newcommand{\bP}{{\bm{P}}}

\newcommand{\z}{\bm{z}}
\newcommand{\zstar}{\z^*}
\newcommand{\indic}{\mathbbm{1}}
\newcommand{\one}{\mathbm{1}}

\newcommand{\C}{\bm{C}}
\newcommand{\Chat}{\bm{\hat{C}}}
\newcommand{\cb}{\bm{c}}

\newcommand{\bz}{\boldsymbol{z}}

\newcommand{\tc}{\tilde{c}}
\newcommand{\tC}{\tilde{C}}

\setlength{\arraycolsep}{0.01cm}

\newcommand{\Bstar}{{\V^*}}   
\newcommand{\bstar}{\v^*}             
\newcommand{\tB}{\B^*}     
\newcommand{\tb}{\b^*}  

\newcommand{\Bcheck}{\check{\V}}
\newcommand{\bcheck}{\check{\v}}

\newcommand{\xhat}{\hat\x}
\newcommand{\Bhat}{\hat\B}
\newcommand{\Xhat}{\hat\X}

\newcommand{\bhat}{\hat\b}
\newcommand{\Uhat}{\hat\U}

\newcommand{\td}{\tilde{\bm{d}}^*}
\newcommand{\init}{{\mathrm{init}}}

\newcommand{\Ustar}{\U^*{}}
\newcommand{\Xstar}{\X^*}
\newcommand{\xstar}{\x^*}
\newcommand{\sstar}{\s^*}
\newcommand{\Sstar}{\S^*}
\newcommand{\Vstar}{\V^*{}}

\newcommand{\deltinit}{\delta_\init}
\newcommand{\deltapt}{\delta_{t}}
\newcommand{\deltaptplus}{\delta_{t+1}}

\newcommand{\bSigma}{{\bm\Sigma^*}}
\newcommand{\tSigma}{\bm{E}_{det}}
\newcommand{\sigmin}{{\sigma_{\min}^*}}
\newcommand{\sigmax}{{\sigma_{\max}^*}}

\newcommand{\HH}{{\bm{D}}}
\newcommand{\GG}{{\bm{M}}}
\newcommand{\SSS}{{\bm{S}}}
\newcommand{\ik}{{ki}}
\newcommand{\J}{\mathcal{J}}

\renewcommand{\P}{\bm{P}}
\newcommand{\proj}{\mathcal{P}}
\newcommand{\E}{\mathbb{E}}
\newcommand{\norm}[1]{\left\|#1\right\|}

\renewcommand{\P}{\bm{U}}
\newcommand{\Phat}{\hat\P} 
\newcommand{\Lam}{\bm\Lambda} 
\renewcommand{\L}{\bm{L}}
\renewcommand{\V}{\bm{V}}

\renewcommand{\l}{\bm{\ell}}
\renewcommand{\v}{\bm{v}}
\newcommand{\tty}{\tilde\y}

\newcommand{\lhat}{\hat\l}

\newcommand{\at}{\a_t}
\newcommand{\yt}{\y_t}
\newcommand{\lt}{\l_t}
\newcommand{\xt}{\x_t}
\newcommand{\vt}{\v_t}
\newcommand{\et}{\e_t}

\newcommand{\T}{\mathcal{T}}
\newcommand{\Tmiss}{{\T_{miss}}}
\newcommand{\Tspar}{{\T_{spar}}}
\newcommand{\outfrac}{\text{\scriptsize{max-outlier-frac}}}
\newcommand{\outfracrow}{\text{\scriptsize{max-outlier-frac-row}}}
\newcommand{\outfraccol}{\text{\scriptsize{max-outlier-frac-col}}}
\newcommand{\missfracrow}{\text{\scriptsize{max-miss-frac-row}}}
\newcommand{\missfraccol}{\text{\scriptsize{max-miss-frac-col}}}

\newcommand{\smin}{s_{min}}
\newcommand{\smax}{s_{max}}

\newcommand{\bea}{\begin{eqnarray}}
\newcommand{\eea}{\end{eqnarray}}

\newcommand{\nn}{\nonumber}
\newcommand{\ds}{\displaystyle}

\newtheorem{theorem}{Theorem}[section]
\newtheorem{prop}[theorem]{Proposition}
\newtheorem{lemma}[theorem]{Lemma}
\newtheorem{claim}[theorem]{Claim}
\newtheorem{corollary}[theorem]{Corollary}
\newtheorem{fact}[theorem]{Fact}
\newtheorem{definition}[theorem]{Definition}
\newtheorem{remark}[theorem]{Remark}
\newtheorem{example}[theorem]{Example}
\newtheorem{sigmodel}[theorem]{Model}
\newtheorem{assu}[theorem]{Assumption}
\renewcommand\thetheorem{\arabic{section}.\arabic{theorem}}

\newcommand{\snr}{\text{SNR}}

\newcommand{\tk}{\tilde{k}}
\newcommand{\tl}{{\ell}}
\newcommand{\totl}{L}

\newcommand{\bfpara}[1]{ {\bf #1. }} 
\newcommand{\Item}{\item} 

\newcommand{\SEF}{\SE}
\renewcommand{\P}{\bm{P}}
	\newcommand{\kron}{\otimes}
\newcommand{\Uvec}{{\U_{vec}}}
\newcommand{\Zvec}{{\Z_{vec}}}
\newcommand{\ym}{{\y_{(mag)}}}
\renewcommand{\forall}{\text{ for all }}

\newcommand{\eps}{\epsilon}

\newcommand{\Section}[1]{\section{#1}} 
\newcommand{\Subsection}[1]{\subsection{#1}} 

\renewcommand{\bhat}{\b}  \renewcommand{\Bhat}{\B}
\renewcommand{\xhat}{\x}  \renewcommand{\Xhat}{\X}








\title{Fast Low Rank column-wise Compressive Sensing}
\title{Fast and Sample-Efficient
\\ Federated Low Rank Matrix Recovery
\\ from column-wise Linear and Quadratic Projections}
\author{Seyedehsara (Sara) Nayer and Namrata Vaswani \\ Dept. of Electrical and Computer Engineering, Iowa State University, USA. \\ Email: namrata@iastate.edu}

\maketitle

\begin{abstract} 
We study the following lesser-known low rank (LR) recovery problem: recover an $n \times q$ rank-$r$ matrix, $\Xstar  =[\xstar_1 , \xstar_2, ..., \xstar_q]$, with $r \ll \min(n,q)$, from $m$ independent linear projections of each of its $q$ columns, i.e., from $\y_k := \A_k \xstar_k ,  k  \in [q]$, when $\y_k$ is an $m$-length vector with $m < n$. The matrices $\A_k$ are known and mutually independent for different $k$.  We introduce a novel gradient descent (GD) based solution called AltGD-Min. We show that, if the $\A_k$s are i.i.d. with i.i.d. Gaussian entries, and if the right singular vectors of $\Xstar$ satisfy the incoherence assumption, then $\epsilon$-accurate recovery of $\Xstar$ is possible with order $(n+q) r^2 \log(1/\epsilon)$ total samples and order $ mq nr \log (1/\epsilon)$ time. Compared with existing work, this is the fastest solution. For $\epsilon < r^{1/4}$, it also has the best sample complexity. A simple extension of AltGD-Min also provably solves LR Phase Retrieval, which is a magnitude-only generalization of the above problem.

AltGD-Min factorizes the unknown $\X$ as $\X = \U\B$ where $\U$ and $\B$ are matrices with $r$ columns and rows respectively. It alternates between a (projected) GD step for updating $\U$, and a minimization step for updating $\B$. Its each iteration is as fast as that of regular projected GD because the minimization over $\B$ decouples column-wise. At the same time, we can prove exponential error decay for it, which we are unable to for projected GD. Finally, it can also be efficiently federated  with a communication cost of only $nr$ per node, instead of $nq$ for projected GD. 
\end{abstract}

%



\section{Introduction}
This work develops a sample-efficient, fast, and communication-efficient gradient descent (GD) solution, called AltGD-Min, for provably recovering a low-rank (LR) matrix from a set of mutually independent linear projections of each of its columns. The communication-efficiency considers a federated setting. 
This problem, which we henceforth refer to as ``Low Rank column-wise Compressive Sensing (LRcCS)'', is precisely defined below.
%
Unlike the other well-studied LR problems -- multivariate regression (MVR) \cite{wainwright_linear_columnwise}, LR matrix sensing \cite{lowrank_altmin} and LR matrix completion (LRMC) \cite{matcomp_candes,lowrank_altmin} -- LRcCS has received little attention so far in terms of approaches with provable guarantees. There are only two existing provably correct solutions.
(1) Its generalization {\em LR phase retrieval (LRPR)}, was studied in our recent work \cite{lrpr_icml,lrpr_it,lrpr_best} where we developed a provably correct alternating minimization (AltMin) solution. Since LRPR is a generalization, the algorithm also solves LRcCS. (2) In parallel work,  \cite{lee2019neurips} developed and analyzed a convex relaxation (mixed-norm minimization) for LRcCS. Both solutions are much slower than GD-based methods, and, in most practical settings, also have worse sample complexity.

LRcCS occurs in accelerated LR dynamic MRI \cite{st_imaging,dyn_mri1,dyn_mri2}, and in distributed/federated sketching \cite{hughes_icml_2014,aarti_singh_subs_learn,lee2019neurips}. We explain these in Sec. \ref{apps}. 
We show the speed and performance advantage of AltGD-Min for dynamic MRI in \cite{lrpr_gdmin_mri}.

\subsection{Problem Setting, Notation, and Assumption}\label{probdef}

\bfpara{Problem definition}
The goal is to recover an $n \times q$ rank-$r$ matrix $\Xstar =[\xstar_1, \xstar_2, \dots, \xstar_q]$ from  $m$ linear projections (sketches) of each of its $q$ columns, i.e. from
\bea
\y_k := \A_k \xstar_k, \ k  \in [q]
\label{ykvec}
\eea
where each $\y_k$ is an $m$-length vector,  $[q]:=\{1,2,\dots, q\}$, and the measurement/sketching matrices $\A_k$ are mutually independent and known. The setting of interest is  low-rank (LR), $r \ll \min(n,q)$, and undersampled measurements, $m < n$. Our guarantees assume that each $\A_k$ is random-Gaussian: each entry of it is independent and identically distributed (i.i.d.) standard Gaussian.

%

We also study the magnitude-only measurements' setting, LRPR \cite{lrpr_icml,lrpr_it,lrpr_best}. This involves recovering $\Xstar$ from
$$\ym_k := |\A_k \xstar_k|, \ k  \in [q].$$
Here $|\z|$ takes the entry-wise absolute value of entries of the vector $\z$.
%


\bfpara{Notation}
Everywhere, $\|.\|_F$ denotes the Frobenius norm,  $\|.\|$ without a subscript denotes the (induced) $l_2$ norm (often called the operator norm or spectral norm),  $\|\M\|_{\max}$ is the maximum magnitude entry of the matrix $\M$, $^\top$ denotes matrix or vector transpose, and $|\z|$ for a vector $\z$ denotes element-wise absolute values. $\I_n$ (or sometimes just $\I$) denotes the $n \times n$ identity matrix. We use $\e_k$ to denote the $k$-th canonical basis vector, i.e., the $k$-th column of $\I$.
For any matrix $\Z$, $\z_k$ denotes its $k$-th column.

We say $\U$ is a {\em basis matrix} if it contains orthonormal columns. For basis matrices $\U_1, \U_2$, we use $$\SEF(\U_1, \U_2): = \|(\I - \U_1 \U_1{}^\top)\U_2\|_F$$ as the Subspace Distance (SD) measure. For two $r$-dimensional subspaces, this is the $l_2$ norm of the sines of the $r$ principal angles between $\Span(\U_1)$ and $\Span(\U_2)$. $\SEF(\U_1,\U_2)$  is symmetric when $\U_1,\U_2$ are both $n \times r$ basis matrices. Notice here we are using the Frobenius SD, unlike many recent works including our older work \cite{lrpr_it} that use the induced 2-norm one. This is done because it enables us to prove the desired guarantees easily.
We reuse the letters $c,C$ to denote different numerical constants in each use with the convention that $c < 1$ and $C \ge 1$.
The notation  $a \in \Omega(b)$ means $a \ge C b$ while $a \in O(b)$ means $a \le Cb$. We use $\indic_{\text{statement}}$ to denote an indicator function that takes the value 1 if $\text{statement}$ is true and zero otherwise.
%

For a vector $\w$, we sometimes use $\w(k)$ to denote the $k$-th entry of $\w$. For a vector $\w$ and a scalar $\alpha$, $\indic(\w \le \alpha)$ returns a vector of 1s and 0s of the same length as $\w$, with 1s where $(\w(k) \le \alpha)$ and zero everywhere else. We use $\circ$ to denote the Hadamard product. Thus $\z:=\w \circ \indic(\w \le \alpha)$ zeroes out entries of $\w$ larger than $\alpha$, while keeping the smaller ones as is.

For $\Xstar$ which is a rank-$r$ matrix, we let
\[
\Xstar \svdeq \Ustar \underbrace{\bSigma \Bstar{}}_{\tB} := \Ustar \tB
\]
denote its reduced (rank $r$) SVD, i.e., $\Ustar$ and $\Bstar^\top$ are matrices with orthonormal columns {\em (basis matrices)}, $\Ustar$ is  $n \times r$ and $\Bstar$ is $r \times q$, and $\bSigma$ is an $r \times r$ diagonal matrix with non-negative entries. We use $\kappa:= \sigmax/\sigmin$ to denote the condition number of $\bSigma$. This is not the condition number of $\Xstar$ (whose minimum singular value is zero).
%
We let  $\tB:= \bSigma \Vstar{}$ and we use $\tb_k$ to denote its $k$-th column.

We use the phrase {\em $\epsilon$-accurate recovery} to refer to $\SE(\U, \Ustar) \le \epsilon$ or $\|\Xhat - \Xstar\|_F  \le \epsilon \|\Xstar\|_F$ or both.

\bfpara{Assumption}
Another way to understand \eqref{ykvec} is as follows: each scalar measurement $\y_\ik$ ($i$-th entry of $\y_k$) satisfies
\[
\y_\ik : =  \langle \a_\ik, \xstar_k \rangle ,   \ i \in [m], \ k \in [q]
\]
with $\a_\ik{}^\top$ being the $i$-th row of $\A_k$. Observe that the measurements are not global, i.e., no $\y_\ik$ is a function of the entire matrix $\Xstar$. They are global for each column  ($\y_\ik$ is a function of column $\xstar_k$) but not across the different columns. We thus need an assumption that enables correct interpolation across the different columns. The following assumption, which is a slightly weaker version of incoherence (w.r.t. the canonical basis) of right singular vectors suffices for this purpose.
%

\begin{assu}[(Weakened) Right Singular Vectors' Incoherence]
Assume that
\[
\max_k \|\tb_k\| \le \sigmax \mu \sqrt{r/q}. 
\]
for a constant $\mu \ge 1 $ ($\mu$ does not grow with  $n,q,r$).
Since $\|\xstar_k\| =\|\tb_k\|$, this implies that $\max_k \|\xstar_k\| \le \sigmax \mu \sqrt{r/q}$. Also, since $\sigmin \sqrt{r} \le \|\Xstar\|_F$,  this also implies that $\max_k \|\xstar_k\| \le \kappa \mu {\|\Xstar\|_F}/{\sqrt{q}}$.
\label{right_incoh}
\end{assu}
Right singular vectors incoherence is the assumption $\max_k \|\bstar_k\| \le \mu \sqrt{r/q}$. Since $\tb_k = \bSigma \bstar_k$, this implies that the above holds. Incoherence of both left and right singular vectors was introduced for guaranteeing correct ``interpolation'' for the LRMC problem \cite{matcomp_candes,lowrank_altmin}.

\subsection{Existing Work}\label{relwork}

\bfpara{Existing solutions for LRcCS and LRPR}
Since it is always possible to obtain magnitude-only measurements $\ym_k$ from linear ones $\y_k$ as $\ym_k = |\y_k|$, a solution to LRPR also automatically solves LRcCS under the same assumptions. Hence the AltMin algorithm for LRPR from \cite{lrpr_icml,lrpr_it} is the first provably correct solution for LRcCS. Of course, since LRcCS is an easier problem than LRPR, we expect a direct solution to LRcCS to need weaker assumptions. As we show in this paper, this is indeed true. A more recent work \cite{lee2019neurips} studied the noisy version of LRcCS and developed a convex relaxation (mixed norm minimization) to provably solve it. Its time complexity is not discussed in the paper, however, it is well known that solvers for convex programs are much slower when compared to direct iterative algorithms: they either require number of iterations proportional to $1/\sqrt{\epsilon}$ or the per-iteration cost has cubic dependence on the problem size (here $((n+q)r)^3$) \cite{lowrank_altmin}. Thus, if $q \le n$, its time complexity $O(mqnr \cdot \min(1/\sqrt{\eps}, n^3 r^3))$.
In \cite{lrpr_best}, we provided the best possible guarantee for the AltMin algorithm for solving LRPR, and hence LRcCS.
%
%
We discuss these results in detail in Sec. \ref{detailed_compare} and summarize them in Table \ref{compare_lrccs}. 

\bfpara{Other well-studied LR recovery problems} The multivariate regression (MVR) problem, studied in \cite{wainwright_linear_columnwise}, is our problem with $\A_k = \A$. However this is a very different setting than ours because, with $\A_k = \A$, the different $\y_k$'s are no longer mutually independent. As a result,  one cannot exploit law of large numbers' arguments over all $mq$ scalar measurements $\y_\ik$. Consequently, the required value of $m$ can never be less than $n$. The result of  \cite{wainwright_linear_columnwise} shows that $m$ of order $(n+q)r$ is both necessary and sufficient.
%
%
LRMS involves recovering $\Xstar$ from $\y_i = \langle \A_i, \Xstar \rangle, \ i=1,2,\dots,mq$ with $\A_i$ being dense matrices, typically i.i.d. Gaussian \cite{lowrank_altmin}. Thus  all measurements are i.i.d. and {\em global}: each contains information about the entire quantity-of-interest, here $\Xstar$. Because of this, for LRMS, one can prove a LR Restricted Isometry Property (RIP) that simplifies the rest of the analysis. This is what makes it very different from, and easier than, our problem.
%

{\em LRMC}, which involves recovering $\Xstar$ from a subset of its observed entries, {\em is the most closely related problem to ours} since it also involves recovery from non-global measurements. The typical model assumed is that each matrix entry is observed with probability $p$ independent of others \cite{matcomp_candes,lowrank_altmin}. Setting unobserved entries to zero, this can be written as $\y_{jk} = \delta_{jk} \Xstar_{jk}$ with $\delta_{jk} \iidsim Bernoulli(p)$. LRMC measurements are both row-wise and column-wise local. To allow correct interpolation across both rows and columns, it needs the incoherence assumption on both its left and right singular vectors. 
For our problem, the measurements are global for each column, but not across the different columns. For this reason, only right singular vectors' incoherence is needed. In fact, because of the nature of our measurements, even if left incoherence were assumed, it would not help.
%
%
This {\em asymmetry in our measurement model and the fact that our measurements are unbounded (each $\y_\ik$ is a Gaussian r.v) are two key differences} between LRMC and LRcCS {\em that prevent us from borrowing LRMC proof techniques for our work}. Here {\em symmetric} means: if we replace $\Xstar$ by its transpose, the probability distribution of the set of measurements does not change. {\em Bounded} means that the measurements' magnitude has a {\em uniform} bound. This bound is $\|\Xstar\|_{\max}$ for LRMC measurements. 



Non-convex (iterative, not convex relaxation based) LRMC algorithms with the best sample complexity are GD-based. There are two common approaches for designing GD algorithms in the LR recovery literature, and in particular for LRMC. The first is to use standard projected GD on $\X$ ({\em projGD-X}), also referred to as Iterative Hard Thresholding: at each iteration, perform one step of GD for minimizing the squared loss cost function, $\tilde{f}(\X)$, w.r.t. $\X$, followed by projecting the resulting matrix onto the space of rank $r$ matrices (by SVD). This was studied in \cite{fastmc,rmc_gd} for solving LRMC.
 This is shown to converge geometrically with a constant GD step size, while needing only $ \Omega ( (n+q)r^2 \log^2 n \log^2 (1/\epsilon))$ samples on average.

The second is to let $\X = \U \B$ where $\U$ is $n \times r$ and $\B$ is $r \times q$ and perform alternating GD  for the cost function $f(\U,\B) := \tilde{f}(\U\B)$,  i.e.,  update $\B$ with one step of GD for minimizing $f(\U,\B)$ while keeping $\U$ fixed at its previous value, and then do the same for $\U$ with $\B$ fixed, and repeat. Since the $\X = \U \B$ factorization is not unique, i.e., $\X = \U \R^{-1} \R \B$ for any invertible $r \times r$ matrix $\R$, this approach can result in the norm of one of $\U$ or $\B$ growing in an unbounded fashion, while that of the other decreases at the same rate, causing numerical problems. A typical approach to resolve this issue, and one that was used for LRMC \cite{rpca_gd,lafferty_lrmc}, is to  change the cost function to minimize to $f(\U, \B) + \lambda f_2(\U,\B)$ where $f_2(\U,\B):=\|\U^\top \U - \B \B^\top\|_F$ is the ``norm-balancing term'' (helps ensure that norms of $\U$ and $\B$ remain similar). We henceforth refer to this approach as {\em altGDnormbal}.
The sample complexity bound for this approach is similar to that for projGD-X. But, it needs a GD step size of order $1/r$ or smaller \cite{rpca_gd,lafferty_lrmc}; making it $r$-times slower than projGD-X.


\begin{table*}
\begin{center}
\resizebox{0.95\linewidth}{!}
{
\begin{tabular}{llllll} \toprule
				 &  Sample Comp. & Time Comp.  & Communic. Comp.  & Holds for & Column-wise \\ 
&   $mq   \gtrsim$ &   & per node (predicted)   &   all $\Xstar$? &  error bound?  \\ 
\hline \midrule
Convex \cite{lee2019neurips} &   $n r  \frac{1}{\epsilon^4}$   &  $\text{\scriptsize{linear-time}} \cdot  \min\left( \frac{1}{\sqrt\epsilon}, n^3 r^3\right)$ & not clear &  yes  & no \\ 
 &&&&& \\
AltMin \cite{lrpr_icml,lrpr_it} &   $n r^4  \log(\frac{1}{\epsilon}) $   & $\text{\scriptsize{linear-time}}   \cdot r\log^2(\frac{1}{\epsilon})$ & $nr\log(\frac{1}{\epsilon})  \cdot r\log^2(\frac{1}{\epsilon})$ & no  \\ 
 &&&&& \\
AltMin \cite{lrpr_best} &   $n r^2 (r+\log(\frac{1}{\epsilon}) )$   & $\text{\scriptsize{linear-time}}   \cdot  r\log^2(\frac{1}{\epsilon})$  & $nr\log(\frac{1}{\epsilon}) \cdot  r\log^2(\frac{1}{\epsilon})$ & no  & yes  \\ 
 &&&&& \\
{\cblue altGD-Min}  & \cblue $\mathbf{n r^2 \log(\frac{1}{\epsilon})}$  & \cblue $\mathbf{\text{\scriptsize{linear-time}} \cdot r\log(\frac{1}{\epsilon})}$  & \cblue  $\mathbf{nr \cdot r\log(\frac{1}{\epsilon})}$ & \cblue  no & \cblue  yes  \\ 
{\cblue  (proposed)} && &&& \\ 
 &&&&& \\
\bottomrule
\multicolumn{6}{l}{Best sample LRMC algorithms among those that do not solve a convex relaxation} \\
\toprule
ProjGD-X & $\max(n,q) r^2 \log^2 n \log^2(\frac{1}{\epsilon})$   & $\bf{\text{\scriptsize{linear-time}} \cdot r \log(\frac{1}{\epsilon})}$  & $nq$ **  && \\ 
\cite{rmc_gd} &&&&& \\  
 &&&&& \\
AltGDnormbal & $\mathbf{\max(n,q) r^2 \log n} $ &  $\text{\scriptsize{linear-time}}  \cdot  r^2 \log(\frac{1}{\epsilon})$  & $\mathbf{\max(n,q)r}$ & &  \\ 
\cite{rpca_gd}  &&&&& \\ 
 &&&&& \\
 \bottomrule
\end{tabular}
}
\end{center}
**The communication complexity of ProjGD-X would be $nq$ because the gradient w.r.t. $\X$ computed at each node will need to be transmitted by the nodes to the center. The gradient w.r.t. $\X$ is not low rank (LR), and hence one cannot transmit just its rank $r$ SVD.
%
%
\vspace{-0.05in}
\caption{\small{Existing work versus our work.
For brevity, this table assumes $q \le n$ and treats $\kappa,\mu$ as numerical constants.
All approaches also need $m \ge \max(r,\log q, \log n)$.
Column-wise error bound exists means $\max_k \|\xstar_k - \xhat_k\|/\|\xstar_k\| \le \eps$ holds in addition to a similar bound on matrix Frobenius norm error.
Linear-time is the time needed to read all algorithm inputs. For LRcCS, this is $\y_k,\A_k$ for all $k \in [q]$ and thus linear-time is order $mnq$. For LRMC, this is the set of observed entries and their locations and thus linear-time is order $mq$.
\cred
None of the other algorithms have been studied in the federated context and hence the communication complexity (Comm. Comp.) listed in the fourth column is based on our understanding of how one would federate the algorithm.
\cbl
Notice that AltGD-min has the best time and communication complexities; and for  $\epsilon^4 < r$, it also has the best sample complexity.
}}
\label{compare_lrccs}
\vspace{-0.1in}
\end{table*}

\subsection{Contributions and Novelty}\label{contrib}
\noindent \bfpara{Contribution to solving LRcCS and LRPR}
(1) This work develops a novel GD-based solution to LRcCS, called AltGD-Min, that is fast and communication-efficient. We show that, with high probability (w.h.p.), AltGD-Min obtains an $\epsilon$-accurate estimate in order $\kappa^2 \log(1/\epsilon)$ iterations, as long as Assumption \ref{right_incoh} holds, the matrices $\A_k$ are i.i.d., with each containing i.i.d. standard Gaussian entries,  $mq \in \Omega ( \kappa^6 \mu^2 (n + q)  r^2 \log(1/\epsilon) )$, and $m \in \Omega(\max(\log q, \log n)\log(1/\epsilon))$.  Its time complexity is $O( mqnr \cdot \kappa^2 \log(1/\epsilon))$ and its communication complexity per node is $O(nr\cdot \kappa^2 \log(1/\epsilon))$.
We provide a comparison of our guarantee with those of other works in Table \ref{compare_lrccs}. 
This table also summarizes the guarantees for the two most sample-efficient LRMC solutions: projGD-X and altGDnormbal. The former is also the fastest LRMC solution, while the latter is the most communication-efficient.
As mentioned earlier, LRMC is the most similar problem to ours that has been extensively studied. Notice that, our sample complexity matches that of the best results for LRMC algorithms that do solve a convex relaxation. 
(2) We show that a simple extension of AltGD-Min also provides the fastest provable solution to LRPR, as long as the above assumptions hold and $mq \in \Omega(\kappa^6 \mu^2 nr^2(r + \log(1/\epsilon))$. Its time complexity is the same too.

\noindent \bfpara{Contributions / Novelty of algorithm design and proof techniques}\label{novelty}
As explained earlier, there are three commonly used provably correct iterative algorithms for LR recovery problems -- altMin, projGD-X, and altGD (altGDnormbal to be precise). AltMin is slower than GD-based methods because, for updating both $\U$ and $\B$, it requires solving a minimization problem keeping the other variable fixed. For our specific asymmetric problem, the min step for $\U$ is the slow one. ProjGD-X and altGDnormbal are faster, but it is not clear how to analyze them for LRcCS under the desired sample complexity\footnote{\cred
In order to show that a GD-based algorithm converges, one needs to be able to bound the norm of the gradient and show that it goes to zero with iterations.  When studying both projGD-X and altGDnormbal, for different reasons, the estimates of the different columns are coupled. Consequently, it is not possible to get a tight enough bound on $\max_k\|\xstar_k - \xhat_k\|$. But, due to the form of the LRcCS measurement model, such a bound is needed to get a tight enough bound on the 2-norm of the gradient of the cost function, and show that it decreases sufficiently at each iteration, under the desired sample complexity. Moreover, in case of projGD-X, even if one could somehow get the desired bound, it would not suffice because the summands will still be too heavy tailed. This point is explained in detail in Appendix \ref{algo_understand}.}. Our novel altGD-min approach however resolves both issues: it is fast as projGD-X and it can be analyzed.
Moreover, its communication complexity for a federated implementation (and its memory complexity) is only $nr$ per node per iteration, instead of $nq$ for projGD-X.
As can be seen from Table \ref{compare_lrccs}, treating $\kappa,\mu$ as numerical constants, it has the best sample-, time-, and communication/memory- complexity among all approaches for LRcCS and all fast (iterative) approaches for LRMC as well. Because of this, an AltGD-Min type algorithm may also be of interest for solving LRMC in a fast, sample-efficient and communication-efficient fashion.  In fact, it can be also be useful for other bilinear inverse problems such as blind deconvolution.



\bfpara{AltGDmin algorithm} The main idea is as follows. Express $\X$ as  $\X = \U \B$ and alternatively update $\U$ and $\B$ as follows: (a) keeping $\B$ fixed at its previous value, update $\U$ by a GD step for it for the cost function $f(\U,\B)$ followed by projecting the output onto the space of matrices with orthonormal columns; and (b) keeping $\U$ fixed at its previous value, update $\B$ by minimizing $f(\U, \B)$ over it. Because of the column-wise decoupled form of our measurement model, step (b) is as fast as the GD step and thus the per-iteration time complexity of AltGD-Min is equal to that of any other GD method such as projGD-X or altGDnormbal. This decoupling (which means that, given $\U$, $\b_k$ only depends on $\xstar_k$, and not on the other columns of $\Xstar$) also allows us to get the desired tight-enough bound on $\max_k\|\bhat_k- \U^\top \xstar_k\|$ and hence on $\max_k\|\xhat_k - \xstar_k\|$. This, and the fact that we use the gradient w.r.t. $\U$ in our algorithm, means that the summands in the gradient, and in other error bound terms, are {\em nice-enough  sub-exponential random variables (r.v.s)}: sub-exponential r.v.s whose maximum sub-exponential norm is small enough (is proportional to $(r/q)$), so that the summation can be bounded w.h.p. under the desired sample complexity.

\bfpara{AltGDmin analysis}  When we analyzed the AltMin approach for LRPR \cite{lrpr_it,lrpr_best}, we could directly modify proof techniques from AltMin for LRMC \cite{lowrank_altmin} for getting a bound on $\SE(\U,\Ustar)$ in terms of the bound on this distance from the previous iteration.
We cannot do this for AltGD-Min because the algorithm itself is different from the two GD approaches studied for solving LRMC. We instead analyze AltGD-Min by a novel use of the fundamental theorem of calculus \cite{lan93} that, along with other linear algebra tricks, helps us get a bound on $\SE(\U,\Ustar)$ which has the desired property: the terms in it are sums of {\em nice-enough sub-exponentials}. See Lemma \ref{algebra} and its proof.
\cred
The use of this result is motivated by its use in \cite{pr_mc_reuse_meas}, and many earlier works, where it is used in a standard way: to bound the Euclidean distance, $\|\x - \xstar\|$, for standard GD to solve the PR problem for recovering a single vector $\xstar$. Thus, at the true solution $\x=\xstar$, the gradient of the cost function was zero. In our case, there are two differences: (i) we need to bound the subspace distance error, and (ii) our algorithm is not standard GD, and this means that $\nabla_U f(\Ustar \Ustar^\top \U, \B) \neq 0$. We explain our approach in Sec. \ref{outline_iters}.
\cbl


\bfpara{AltGDmin initialization}  The standard LR spectral initialization approach cannot be used because its summands are sub-exponential r.v.s that are not {\em nice-enough}. 
We give a detailed explanation in Appendix \ref{algo_understand}. 
We address this issue by borrowing the truncation idea from the PR literature \cite{twf,rwf,lrpr_it}. But, in our case, truncation is applied to a non-symmetric matrix. Thus the sandwiching arguments developed for symmetric matrices in \cite{twf}, and modified in \cite{rwf,lrpr_it}, cannot be borrowed. We need a different argument which is used for proving Lemma \ref{lem:init_denom_term} and is briefly explained in Sec. \ref{outline_init}. 
\subsection{Applications}\label{apps}
The LRcCS and LRPR problems occur in projection imaging applications involving sets of images, e.g., dynamic MRI \cite{st_imaging,dyn_mri1,dyn_mri2}, federated LR sketching \cite{hughes_icml_2014,lee2019neurips}, and  dynamic Fourier ptychography (LRPR) \cite{TCIgauri}.
In MRI, Fourier projections of the region of interest, e.g., a cross-section of the brain or the heart, are acquired one coefficient at a time, making the scanning (data acquisition) quite slow. Hence, reduced sample complexity enables accelerated scanning. 
Since medical image sequences are usually slow changing, the LR model is a valid assumption for a time sequence \cite{st_imaging,dyn_mri1,dyn_mri2}.  
In our notation, $\xstar_k$ is the vectorized version of the $k$-th image of the sequence and there are a total of $q$ images. The matrices $\A_k$ are random Fourier, i.e., $\A_k = \H_k \F$ where $\F$ is the $n \times n$ matrix that models computation of the 2D discrete Fourier transform as a matrix-vector operation, and $\H_k$ is an $m \times n$ random sampling ``mask'' matrix that models the frequency selection.
In \cite{lrpr_gdmin_mri}, we have shown the power of AltGD-Min for fast undersampled dynamic MRI of medical image sequences. It is both much faster, and in most cases, also provides better reconstructions,  than many existing solutions from the MRI literature.

Large scale usage of smartphones results in large amounts of geographically distributed data, e.g., images. There is a need to compress/sketch this data before storing it. Sketch refers to a compression approach where the compression end is low complexity, usually simple linear projections \cite{hughes_icml_2014,lee2019neurips}. Consider the setting where different subsets of columns  of $\Xstar$ (each column corresponds to one vectorized image) are available at each of the $\rho \le q$ nodes. The goal is to sketch them so that they can be correctly recovered using a federated algorithm.
We can store the sketches $\y_k := \A_k \xstar_k$ with $\A_k$'s being i.i.d. Gaussian. This way we store a total of only $mq$ scalars, with $mq$ of order roughly just $(n+q)r^2$. 
Traditional LR sketching approaches, e.g., \cite{cov_sketch}, are designed for centralized settings and will not be efficient in a distributed setting. 

\subsection{Organization}
In Sec. \ref{algo_thm}, we develop AltGD-Min, give its guarantee for solving LRcCS, and compare it with existing results. We state and prove the two theorems that help prove our main result in Sec. \ref{proving_mainres}. This section also contains brief proof outlines before the actual proofs. The lemmas used in these proofs are proved in Sec. \ref{proof_lemmas}. The extension for solving LRPR is developed, and its guarantee is stated and proved, in Sec. \ref{algo_thm_proof_lrpr}.  We discuss the limitations of our results in Sec. \ref{limitations}. Simulation experiments are provided in Sec. \ref{sims}. We conclude in Sec. \ref{conclude}.%

\section{The Proposed AltGD-Min Algorithm and Guarantee}
\label{algo_thm}

\subsection{The AltGD-Min algorithm} \label{algo_explain}
We would like to design a fast GD algorithm to find the matrix $\X$ that minimizes the squared-loss cost function
$
\tilde{f}(\X): = \sum_{k=1}^q \|\y_k - \A_k \xhat_k\|^2.
$
For reasons described earlier, we decompose $\X = \U \B$ and develop an alternating GD-min (AltGD-Min) approach for the squared loss function,
\[
f(\U,\B) := \tilde{f}(\U\B) = \sum_k \|\y_k - \A_k \U \b_k\|^2.
\]
Starting with a careful initialization for $\U$ explained below,  AltGD-Min proceeds as follows. At each new iteration,
\bi
\item {\em Min-B: } update $\B$ by solving $\B \leftarrow \arg\min_{\tilde\B} f(\U,\tilde\B)$. Since $\b_k$ only occurs in the $k$-th summand of $f(\U,\B)$, this decouples to a much simpler column-wise least squares (LS) problem: $\b_k \leftarrow \arg\min_{\tilde\b_k} \|\y_k - \A_k \U \tilde\b_k\|^2$. This is solved in closed form as $\bhat_k = (\A_k \U)^\dag \y_k$ for each $k$; here $\M^\dag:=(\M^\top \M)^{-1}\M^\top$.

\item {\em ProjGD-U: } update $\U$ by one GD step for it, $\Uhat^+ \leftarrow \U - \eta \nabla_U f(\U,\B)$, followed by  projecting $\Uhat^+$ onto the space of matrices with orthonormal columns to get the updated $\U^+$.
 We get $\U^+$ by QR decomposition: $\Uhat^+ \qreq \U^+ \R^+$.
\ei
Notice that, because of the decoupling for $\B$, the min step only involves solving $q$ $r$-dimensional Least Squares (LS) problems, in addition to also first computing the matrices, $\A_k \U$. 
Computing the matrices needs time of order $mnr$, and solving one LS problem needs time of order $mr^2$. Thus, the LS step needs time $O(q\max(mnr, m r^2)) = O(mqnr)$ since $r \le n$. This is equal to the time needed to compute the gradient w.r.t. $\U$; and thus, the per-iteration cost of AltGD-Min is only $O(mqnr)$. The QR decomposition of an $n \times r$ matrix takes time only $nr^2$.


Since $f(\U,\B)$ is not a convex function of the unknowns $\{\U, \B\}$, a careful initialization is needed.
Borrowing the spectral initialization idea from LRMC and LRMS solutions, we should initialize $\U_0$ by computing the top $r$ singular vectors of
\[
\X_{0,full} = \frac{1}{m} [ (\A_1^\top \y_1), (\A_2^\top \y_2), \dots, (\A_k^\top \y_k), \dots (\A_q^\top \y_q)]
\]
Clearly the expected value of the $k$-th column of this matrix equals $\xstar_k$ and thus $\E[\X_{0,full}] = \Xstar$. But, as we explain next, it is not clear how to prove that this matrix concentrates around $\Xstar$.
Observe that it can also be written as
\begin{align*}
\X_{0,full} :=    \frac{1}{m}\sum_{k=1}^q \sum_{i=1}^m  \a_\ik \y_\ik \e_k{}^\top
\end{align*}
Its summands are independent sub-exponential r.v.s with maximum sub-exponential norm $\max_k \|\xstar_k\| \le \mu \sqrt{r/q} \sigmax$. This is too large and does not allow us to bound  $\|\X_{0,full} - \Xstar\|$ under the desired sample complexity; see Appendix \ref{algo_understand}. 
To resolve this issue, we borrow the truncation idea from earlier work on PR \cite{twf,lrpr_it} and initialize $\U_0$ as the top $r$ left singular vectors of
\bea
\Xhat_{0} & := & \frac{1}{m}\sum_{k=1}^q \sum_{i=1}^m  \a_\ik \y_\ik \e_k{}^\top \indic_{\left\{ \y_\ik^2 \le  \alpha \right\} } \nonumber \\
 & = & \frac{1}{m}\sum_{k=1}^q \A_k^\top \y_{k,trunc}(\alpha) \e_k^\top
\label{newinit}
\eea
where $\alpha := \tC \frac{\sum_\ik  (\y_\ik)^2}{mq}$ and $\y_{k,trunc}(\alpha) := \y_k \circ \indic(|\y_k| \le  \sqrt\alpha )$. We set $\tC$ in our main result. 
Observe that we are summing over only those $i,k$ for which $\y_\ik^2$ is not too large (is not much larger than its empirically computed average value). This {truncation} filters out the too large (outlier-like) measurements and sums over the rest. Theoretically, this converts the summands into sub-Gaussian r.v.s which have lighter tails than the un-truncated ones. This allows us to prove the desired concentration bound.
Different from the above setting, in  \cite{twf,lrpr_it}, truncation was applied to symmetric positive definite matrices and was used to convert summands that were heavier-tailed than sub-exponential to sub-exponential.%

We summarize the complete algorithm in Algorithm \ref{gdmin}.  This uses sample-splitting which is a commonly used approach in the LR recovery literature \cite{lowrank_altmin,fastmc,rmc_gd} as well as in other compressive sensing settings. It helps ensure that the measurement matrices in each iteration for updating $\U$ and $\B$ are independent of all previous iterates. This allows one to use concentration bounds for sums of independent r.v.s. We provide a detailed discussion in Sec. \ref{samplesplit}.


\subsubsection{Practical algorithm and setting algorithm parameters}
First, when we implement the algorithm, we use Algorithm \ref{gdmin} with using the full set of measurements for all the steps (no sample-splitting).
The algorithm has 4 parameters: $\eta$, $T$, $\tC$ and the rank $r$. According to the theorem below, we should set $\eta = c / \sigmax^2$ with $c<0.5$. But $\sigmax$ is not known. The initialization matrix $\Xhat_0$ provides an approximation to $\Xstar$ and hence we can set $\eta = c/\|\Xhat_0\|^2$.
Consider $\tC$. The theorem requires setting $\tC = 9 \kappa^2 \mu^2$, however $\kappa,\mu$ are functions of $\Xstar$ which is unknown. Using the definition of $\mu$ from Assumption \ref{right_incoh}, we can replace $\kappa^2 \mu^2$ by an estimate of its lower bound: $q \cdot \max_k \widehat{\|\xstar_k\|^2} / \widehat{\|\Xstar\|_F^2}$ with $ \widehat{\|\xstar_k\|^2} = (1/m) \sum_i \y_\ik^2$ and $\widehat{\|\Xstar\|_F^2} = (1/m)  \sum_k \sum_i \y_\ik^2$.
To set the total number of algorithm iterations $T$, we can use a large maximum value along with breaking the loop if a stopping criterion is satisfied. A common stopping criterion for GD is to stop when the iterates do not change much. One way to do this is to stop when $\SE(\U_t, \U_{t-1}) \le 0.01 \sqrt{r}$ for last few iterations.

As explained in \cite{lrpr_gdmin_mri}, we can use the following constraints to set the rank. We need our choice of rank, $\hat{r}$, to be sufficiently small compared to $\min(n,q)$ for the algorithm to take advantage of the LR assumption. Moreover, for the LS step for updating $\bhat_k$'s (which are $r$-length vectors) to work well (for its error to be small), we also need it to also be small compared with $m$.
One approach that is used often is to use the ``$b\%$ energy threshold'' on singular values.
Thus, one good heuristic that respects the above constraints is to compute the  ``$b\%$ energy threshold'' of the first $\min(n,q,m)/10$ singular values, i.e. compute $\hat{r}$ as the smallest value of $r$ for which
\[
\sum_{j=1}^{{r}} \sigma_j(\X_0)^2 \ge (b/100) \cdot \sum_{j=1}^{\min(n,q,m)/10} \sigma_j(\X_0)^2
\]
for a $b \le 100$. In our MRI experiments in \cite{lrpr_gdmin_mri}, we used $b=85$. We also realized from the experiments that the algorithm is not very sensitive to this value as long as $\hat{r} \ll \min(n,q,m)$.

\subsubsection{Federating the algorithm}
Suppose that our sketches $\y_k$ are geographically distributed across a set of $\totl$ nodes. Each node $\tl$ stores a subset, denoted $\mathcal{S}_\tl$, of the $\y_k$s with $|\mathcal{S}_\tl|=q_\tl$. These subsets are mutually disjoint so that $\sum_\tl q_\tl = q$. 
Typically $\totl \ll q$.
%
Privacy constraints dictate that we cannot share the $\y_k$s with the central server; although summaries computed using the $\y_k$s can be shared at each algorithm iteration. 
This will be done as follows. Consider the GDmin steps of Algorithm \ref{gdmin} first. Line 13 (Update $\b_k$s, $\x_k$s) is done locally at the node that stores the corresponding $\y_k$. For line 14 (Gradient w.r.t $\U$), the partial sums over $k \in \mathcal{S}_\tl$ are computed at node $\tl$ and transmitted to the center which adds all the partial sums to obtain $\nabla_\U f(\U,\B)$. Line 15 (GD step) and line 16 (projection via QR) are done at the center. The updated $\U$ is then broadcast to all the nodes for use in the next iteration.
The per node time complexity of this algorithm is thus $mnr q_\tl$ at each iteration. The center only performs additions and a QR decomposition (an order $nr^2$ operation) in each iteration. Thus, the time complexity of the federated solution is only $mnr(\max_\tl q_\tl)T$ per node. 
%

The initialization step can be federated by using the Power Method (PM) \cite{golub89,npm_hardt} to compute the top $r$ eigenvectors of $\Xhat_0 \Xhat_0{}^\top$. Any PM guarantee helps ensure that its output is close in subspace distance to the span of the top $r$ eigenvectors of $\Xhat_0 \Xhat_0{}^\top$ after a sufficient number of iterations.
%
The communication complexity of the federated implementation is thus just $nr$ per node per iteration (need to share the partial gradient sums).
Observe also that the information shared with the center is not sufficient to recover $\Xstar$ centrally. It is only sufficient to recover $\Span(\Ustar)$.
The recovery of the columns of $\B$, $\tb_k$, is entirely done locally at the node where the corresponding $\y_k$ is stored, thus ensuring privacy. 

\begin{algorithm}[t]
\caption{\small{The AltGD-Min algorithm.   Let $\M^\dagger:= (\M^\top\M)^{-1} \M^\top$.}} 
\label{gdmin}
\begin{algorithmic}[1]
   \State {\bfseries Input:} $\y_k, \A_k, k \in [q]$
 \State {\bfseries Parameters:} Multiplier in specifying $\alpha$ for init step, $\tC$; GD step size, $\eta$; Number of iterations, $T$

\State {\bfseries Sample-split:} Partition the measurements and measurement matrices into $2T+1$ equal-sized disjoint sets: one set for initialization and $2T$ sets for the iterations. Denote these by $\y_k^{(\tau)}, \A_k^{(\tau)}, \tau=0,1,\dots 2T$.

  \State {\bfseries Initialization:}
\State Using $\y_k \equiv \y_k^{(0)}, \A_k \equiv \A_k^{(0)}$,
set
\\ $\alpha = \tC \frac{1}{mq}\sum_\ik\big|\y_\ik\big|^2$,
\\
$\y_{k,trunc}(\alpha):= \y_k \circ \indic\{|\y_k| \le \sqrt{\alpha} \}$ 
\\
$\ds  \Xhat_{0}:= (1/m) \sum_{k \in [q]} \A_k^\top \y_{k,trunc}(\alpha) \e_k^\top$

%
\State   Set $\U_0 \leftarrow $ top-$r$-singular-vectors of $\Xhat_0$
\State {\bfseries GDmin iterations:}

   \For{$t=1$ {\bfseries to} $T$}

   \State  Let $\U \leftarrow \U_{t-1}$.
\State {\bfseries Update $\bhat_k, \xhat_k$: } For each $k \in [q]$, set $(\bhat_k)_{t}  \leftarrow  (\A_k^{(t)} \U)^\dagger \y_k^{(t)}$ and set $(\xhat_k)_{t}    \leftarrow \U (\bhat_k)_{t}$

\State {\bfseries Gradient w.r.t. $\U$: } With $\y_k \equiv \y_k^{(T+t)}, \A_k \equiv  \A_k^{(T+t)}$, compute $\nabla_\U f(\U, \B_t) =   \sum_k \A_k^\top (\A_k \U (\bhat_k)_t - \y_k) (\bhat_k)_t^\top$
%
\State  {\bfseries GD step: } Set $\ds \Uhat^+   \leftarrow \U - (\eta/m) \nabla_\U f(\U, \B_t)$.

 \State {\bfseries Projection step: }  Compute $\Uhat^+ \qreq \U^+ \R^+$.
 \State Set $\U_t \leftarrow \U^+$.
     \EndFor
\end{algorithmic}
\end{algorithm}


\subsection{Main Result}\label{main_res}
We can prove the following result.

\begin{theorem}
Consider Algorithm \ref{gdmin}. Let $m_t$ denote the number of samples used in iteration $t$.
 Set $\tC = 9\kappa^2 \mu^2$, $\eta = c / \sigmax^2$ with a $c \le 0.5$, and $T = C \kappa^2 \log(1/\epsilon)$. Assume that Assumption \ref{right_incoh} holds and that the $\A_k$s are i.i.d. and each contains i.i.d. standard Gaussian entries.
If 
\[
m_0 q \ge C \kappa^6 \mu^2  (n + q) r^2,
\]
and $m_t$ for $t \ge 1$ satisfies
\[
m_t q \ge C \kappa^4 \mu^2  (n + q) r^2 \log \kappa \text{ and } m_t \ge C \max(r,\log q, \log n)
\]
then, with probability (w.p.) at least $1- t n^{-10}$, for all $t \ge 0$,
\[
\SEF(\U_t, \Ustar) \le \left(1 - \frac{(\eta \sigmax^2) 0.4}{\kappa^2} \right)^t \delta_0 
\]
with $\delta_0 = 0.09/\kappa^2.$ 
Thus, with $T= C \kappa^2 \log(1/\epsilon)$ and $\eta = 0.5/\sigmax^2$, w.p. at least $1 - (T+1) n^{-10}$,
\begin{align*}
&	\SEF(\U_T,\Ustar) \le \epsilon,   \  {\|(\xhat_{k})_T-\xstar_{k}\|} \le \epsilon {\|\xstar_k\|}, \text{  for all $k \in [q]$, } \\
& \|\Xhat_T - \Xstar\|_F \le 1.4  \epsilon \|\Xstar\|
\end{align*}
\label{gdmin_thm}
\end{theorem}

{\em Sample complexity}
The sample complexity (total number of samples needed to achieve $\epsilon$-accurate recovery) is  $m_\tot = \sum_{\tau=0}^T m_\tau \ge m_0 + T \min_{t\ge 1} m_t$. From the above result, this needs to satisfy $m_{tot} q \ge C \kappa^6 \mu^2 (n+q) r^2 \log(1/\epsilon) \log (\kappa)$ and $m_{tot} > C \kappa^2 \max(r,\log q, \log n ) \log(1/\epsilon)$.

{\em Time complexity}
Let $m \equiv m_t$. The initialization step needs time $mq n $ for computing $\Xhat_0$; and  time of order $nqr$ times the number of iterations used in the $r$-SVD step. Since we only need a $\delta_0$-accurate initial estimate of $\Span(\Ustar)$, with $\delta_0 = c/\kappa^2$, order $\log(\kappa)$ number of iterations suffice for this SVD step. Thus the complexity is $O(nq(m+r) \cdot \log \kappa) = O(mqn \cdot \log \kappa)$ since $m \ge r$. One gradient computation needs time $O(mq nr)$. The QR decomposition needs time of order $nr^2$. The update of columns of $\B$ by LS also needs time $O(mqnr)$ (explained earlier).
 As we prove above, we need to repeat these steps $T = O(\kappa^2 \log (1/\epsilon))$ times.
 Thus the total time complexity is $O(mqn\log\kappa + \max(mqnr, nr^2, mqnr) \cdot T ) = O( \kappa^2 mqnr \log(1/\epsilon) \log\kappa)$.

{\em Communication complexity} The communication complexity per node per iteration for a federated implementation is just order $nr$. Thus, the total is  $O( nr \cdot \kappa^2 \log(1/\epsilon))$.

Thus, we have the following corollary.
\begin{corollary}[AltGD-Min]
In the setting of Theorem \ref{gdmin_thm}, if Assumption \ref{right_incoh} holds, and if
\[
m_{tot} q \ge C \kappa^6 \mu^2 (n+q) r^2 \log(1/\epsilon) \log (\kappa)
\]
and $m_{tot} > C \kappa^2 \max(r,\log q, \log n ) \log(1/\epsilon)$,
then, w.p. at least $1 - (C\kappa^2\log(1/\epsilon)) n^{-10}$,
$ \|\Xhat - \Xstar\|_F \le 1.4 \eps \|\Xstar\|$ and ${\|\xhat_{k}-\xstar_{k}\|} \le \eps { \|\xstar_k\|}$ for all $k \in [q]$.
The time complexity is $C \kappa^2 mqnr \log (1/\epsilon) \log\kappa$ and the communication complexity is $O( nr \cdot \kappa^2 \log(1/\epsilon))$.
\end{corollary}

Observe that the above results show that after $T=C \kappa^2 \log(1/\epsilon)$ iterations,  $\SEF(\U_T,\Ustar) \le \epsilon$, ${\|\xhat_{k}-\xstar_{k}\|} \le \eps { \|\xstar_k\|}$, and $\|\Xhat_T - \Xstar\|_F \le 1.4 \epsilon \|\Xstar\|$. The RHS in the third bound does indeed contain $\|\Xstar\|$ (the induced 2-norm). This is correct because, $\SE(.,.)$ is a Frobenius norm subspace distance. We explain this in Sec. \ref{outline_iters}.

\subsection{Discussion and comparison with the best LRMC results}
An algorithm is called linear time if its time complexity is the same order as the time needed to load all input data. In our case, this is $O(mqn)$. Treating $\kappa$ as a constant, the AltGD-Min complexity is worse than linear-time by a factor of only $r \log(1/\epsilon)$. As can be seen from Table \ref{compare_lrccs}, the same is also true for the fastest LRMC solution, projGD-X \cite{rmc_gd}. For LRMC, linear time is $O(mq)$.
To our best knowledge, this is the case for the fastest algorithms for all LR problems.%
Consider the sample complexity.
The degrees of freedom (number of unknowns) of a rank-$r$ $n \times q$ matrix are $(n+q)r$. A sample complexity of $\Omega ( (n+q) r)$ samples (or, sometimes this times log factors) is called ``optimal''.
Thus, ignoring the log factors, our sample complexity of $m_{tot} q \gtrsim (n + q) r^2$ is sub-optimal only by a factor of $r$.
As can also be seen from Table \ref{compare_lrccs}, this suboptimality matches that of the best results for LRMC solutions that are not convex relaxation based \cite{rmc_gd,rpca_gd,lafferty_lrmc}.
The need for exploiting incoherence while obtaining the high probability bounds on the recovery error terms is what introduces the extra factor of $r$ for both LRMC and LRcCS. LRMC has been extensively studied for over a decade and there does not seem to be a way to obtain an (order-) optimal sample complexity guarantee for it except when studying convex relaxation solutions (which are much slower).  

%

In addition, we also need $m \gtrsim \max(r,\log q, \log n)$. This is redundant except for very large $q,n$. This is needed because, the recovery of each column of $\tB$ is a decoupled $r$-dimensional LS problem. We analyze this step in Lemma \ref{B_lemma}; notice that the bound on the recovery error of column $k$ holds w.p. at least $1 - \exp(r - cm)$. By union bound, it holds for all $q$ columns w.p. at least at least $1 - q \exp(r - cm) = 1 - \exp(\log q  + r - c m)$. This probability is at least $1 - n^{-10} = 1 - \exp(-10 \log n)$ if $m \gtrsim \max(r,\log q, \log n)$.

\subsection{Detailed comparison with existing LRcCS results}\label{detailed_compare}
There are two existing solutions for LRcCS -- AltMin  \cite{lrpr_icml,lrpr_it,lrpr_best} and the convex relaxation (mixed norm minimization) \cite{lee2019neurips}.
Mixed norm is defined as $\|\X\|_{mixed}:= \inf_{ \{\U, \V: \U \V = \X\} } \|\U\|_F \max_{k \in [q]} \|\v_k\|$, where $\U$ is $n \times r$ and $\V:=[\v_1, \v_2, \dots \v_q]$ is an $r \times q$ matrix.
In our notation, for the noise-free case ($\sigma=0$), their main result states the following.%
\begin{prop}[Convex relaxation (mixed norm min) in the $\sigma=0$ (noise-free) setting \cite{lee2019neurips}]\label{convex_sol}
Consider a matrix $\Xstar \in \{ \Xstar: \max_k \|\xstar_k\|^2 \le \alpha^2, \|\Xstar\|_{mixed} \le R \le \alpha \sqrt{r} \}$. Then, w.p. $ 1- \exp( - c_2 n R^2 / \alpha^2)$,
$
\frac{\|\Xhat - \Xstar\|_F^2}{\|\Xstar\|_F^2} \le c_1 \frac{\alpha^2}{\|\Xstar\|_F^2/ q }  \sqrt{ \frac{(n+q) r \log^6 n}{m_{tot}  q} }
$
Under our Assumption \ref{right_incoh}, $\max_k \|\xstar_k\|^2 \le \mu^2 (r/q) \sigmax^2 = (\mu^2 \kappa^2)(r/q) \sigmin^2 \le (\kappa^2 \mu^2 ) \|\Xstar\|_F^2/q$, i.e. $\frac{\alpha^2}{\|\Xstar\|_F^2/ q } =  (\kappa^2 \mu^2) $.
Thus, the above result can also be stated as:

For all matrices $\Xstar$ that satisfy Assumption \ref{right_incoh} and for which  $\|\Xstar\|_{mixed} \le  \sqrt{r} \cdot \kappa \mu \|\Xstar\|_F/\sqrt{q} $, if   
$$m_{tot} q \ge C_1 \kappa^4  \mu^4 (n+q) r \log^6 n \cdot \frac{1}{\eps^4},$$ then,  w.p. at least $1 - \exp(-c_2 n)$, ${\|\Xhat - \Xstar\|_F} \le \eps {\|\Xstar\|_F}$.
The time complexity is $C mq nr  \min( \frac{1}{\sqrt\epsilon}, n^3 r^3)$ (explained earlier in Sec. \ref{relwork}).
\end{prop}

Notice that both the sample and the time complexity of the convex solution depend on powers of $1/\sqrt{\eps}$: the sample complexity grows as $1/\epsilon^4$ while the time complexity grows as $1/\sqrt{\eps}$. However, its sample complexity has an order-optimal dependence on $r$.  For AltGD-Min, both sample and time complexities depend only logarithmically on $\eps$ only as $\log(1/\eps)$. But its sample complexity depends sub-optimally on $r$, it grows as $r^2$. In summary, the time complexity of the convex solution is always much worse, its sample complexity is worse when a solution with accuracy level $\eps < 1/{r}^{1/4}$ is needed.
A second point to mention is that our result for AltGD-Min  provides a column-wise error bound (bounds $\|\xstar_k - \xhat_k\|/\xstar_k\|$). The convex result only provides a bound on the Frobenius norm of the entire matrix. Thus it is possible that some columns have much larger recovery error than others. This can be problematic in applications such as dynamic MRI where each column corresponds to one signal/image of a time sequence and where the goal is to ensure accurate-enough recovery of all columns. 
On the other hand, the advantage of the convex guarantee is that it holds w.h.p. for all matrices $\Xstar$ in the specified set, where as our result only holds w.h.p. for {\em a} matrix $\Xstar$ satisfying Assumption \ref{right_incoh}.
The reason for these last two points and the reason that we cannot avoid using sample-splitting is the same: the update of $\B$ is a column-wise LS problem. We explain the reasoning carefully in Sec. \ref{samplesplit} where we discuss the limitations of our approach. 
A second advantage of the convex result is that it directly studies the noisy version of the LRcCS problem. This should be possible for AltGD-Min too, we postpone it to future work.

The best result for AltMin is from \cite{lrpr_best}, it states the following.%
\begin{prop}[AltMin \cite{lrpr_best}]
Under Assumption \ref{right_incoh}, if $$m_{tot} q \ge C \kappa^8 \mu^2 n r^2 (r + \log(1/\epsilon)) \text{ and }  m_{tot} > \max(r,\log q, \log n ),$$ then,  w.p. at least $1 - (\log(1/\eps)) n^{-10}$, $ \|\Xhat - \Xstar\| \le \eps  \|\Xstar\|$ and ${\|\xhat_{k}-\xstar_{k}\|} \le \eps { \|\xstar_k\|}$ for all $k \in [q]$.
The time complexity is $C mqnr \log^2 (1/\epsilon)$.
\label{altmin_best}
\end{prop}

Treating $\kappa$ as a numerical constant, compared with the above result for AltMin, the sample complexity of AltGD-Min is either better by a factor of $ r$ or is as good. It is better when $r > \log(1/\epsilon)$. Also, the time complexity  is always better by a factor $\log(1/\epsilon)$. As a function of $\kappa$, the AltGD-Min sample complexity is better by a factor of $\kappa^2$, but its time is worse by a factor of $\kappa^2$ compared to that of AltMin. The reason is that its error decays as $(1-c/\kappa^2)^t$. For AltMin the error decays as $c^t$.
Experimentally, GD is usually much faster than AltMin because the constants in its time complexity are also lower.

\section{Proving Theorem \ref{gdmin_thm}} \label{proving_mainres}

\subsection{Two key results for proving Theorem \ref{gdmin_thm} and its proof}\label{gdmin_thm_proof}
Theorem \ref{gdmin_thm} is an almost immediate consequence of the following two results. 

\begin{theorem}[Initialization]
\label{init_thm}
Pick a $\delta_0 < 0.1$.
If $mq \ge C  \kappa^4 \mu^2 (n + q)r^2 / \delta_0^2$ , then w.p. at least $1 - \exp(-c (n+q))$,
\[
\SEF(\Ustar, \U_0) \le \delta_0.
\]
\end{theorem}

\begin{proof} See Sec. \ref{init_proof} (simpler proof with sample-splitting for $\alpha$) or Appendix \ref{init_reuse_proof} (proof without sample-splitting). Proof outline is given in Sec. \ref{outline_init}. \end{proof}

\begin{theorem}[GD Descent]
\label{iters_thm}
If, at each iteration $t$, $m q \ge C \kappa^4 \mu^2 (n+q)r^2  \log \kappa$ and $m > C \max(\log q, \log n)$;
if $\SEF(\Ustar, \U_0) \le \delta_0 = c/\kappa^2$ for a $c \le 0.1/1.1$;
and if $\eta \le 0.5 /  \sigmax^2$,
then w.p. at least $1 - (t+1) n^{-10}$,
$$
\SEF(\Ustar, \U_{t+1}) \le \delta_{t+1}:= \left(1 - (\eta \sigmax^2) \tfrac{0.4}{\kappa^2}\right)^{t+1} \delta_0. 
$$
If $\eta = 0.5 \sigmax^2$, this simplifies to $\SEF(\Ustar, \U_{t+1}) \le (1 -0.2/\kappa^2)^{t+1} \delta_0$.

Also, with the above probability,
\[
\|(1/m)\nabla_U f(\U_t,\B_{t+1})\| \le 1.6 \delta_t \sigmax^2.
\]
with $\delta_t$ defined in the $\SEF(\Ustar, \U_{t+1})$ bound above.
\end{theorem}

Since $\delta_t$ decays exponentially with $t$, the same is also true for the gradient norm at iteration $t$, $\|(1/m)\nabla_U f(\U_t,\B_{t+1})\|$.


\begin{proof} See Sec. \ref{iters_proof}. Proof outline is given in Sec. \ref{outline_iters}.
\end{proof}



\begin{proof}[Proof of Theorem \ref{gdmin_thm}]
The $\SE(.)$ bound is an immediate consequence of Theorems \ref{init_thm} and \ref{iters_thm}. To apply Theorem \ref{iters_thm}, we need $\delta_0 = c / \kappa^2$. By Theorem \ref{init_thm}, if $mq \ge C \kappa^6 \mu^2 (n+q)r^2$, then, w.p. at least $1-n^{-10}$, $\SEF(\Ustar, \U_0) \le \delta_0 = c / \kappa^2$.
With this,  if, at each iteration, $mq \ge C \kappa^4 \mu^2 (n+q)r^2  \log \kappa$ and $m \ge C \max(\log q, \log n)$, then by Theorem \ref{iters_thm}, w.p. at least $1-(t+1) n^{-10}$, the stated bound on  $\SEF(\Ustar, \U_{t+1})$ holds.
By setting $T =  C \kappa^2 \log(1/\epsilon)$ in this, we can guarantee $\left(1 - \tfrac{c_1}{\kappa^2}\right)^T \le \epsilon$. This proves the $\SE(\U_T, \Ustar)$ bound. 
The bounds on $\|\xhat_{k}-\xstar_{k}\|$ and $\|\Xhat - \Xstar\|_F$ follow by Lemma \ref{B_lemma} given in Sec. \ref{iters_proof}.%
\end{proof}

\subsection{Proof outline (and novelty) for Theorem \ref{iters_thm}}\label{outline_iters}
For proving exponential error decay, we need to show this: at iteration $t$, if $\SE(\U,\Ustar) \le \delta_t$ with $\delta_t < \delta_0 = c/\kappa^2$. Then, $\SE(\U^+,\Ustar) \le c \delta_t$ for a $c<1$. We explain how to do this next.
Suppose that, at iteration $t$, $\SE(\U,\Ustar) \le \delta_t  < \delta_0 = 0.1/\kappa^2$.

\bfpara{Analyzing the minimization step for updating $\B$ (Lemma \ref{B_lemma})}
Recall from Algorithm \ref{gdmin} that $\b_k = (\A_k \U)^\dag \y_k$, $\x_k = \U \b_k$, and $\xstar_k = \Ustar \tb_k$. Using standard results from \cite{versh_book}, we can show that the estimates $\bhat_k$ satisfy $\|\bhat_k - \U^\top \xstar_k\| \le 0.4 \|(\I - \U \U^\top) \Ustar \tb_k\|$. This then implies that (i) $\bhat_k$'s are incoherent, i.e., $\|\bhat_k\|  \le 1.1\mu \sigmax \sqrt{r/q}$; and (ii) $\|\xhat_k - \xstar_k\| \le  1.4 \|(\I - \U \U^\top) \Ustar \tb_k\| \le 1.4\delta_t \max_k \|\xstar_k\|$, i.e., we can get the desired column-wise error bound.
 Also (iii) $\|\Xhat - \Xstar\|_F \le 1.4 \delta_t \sigmax$ (notice this bound does not contain $r$). We get this as follows:
\begin{align*}
\|\Xhat - \Xstar\|_F 
& = \sqrt{\sum_k \|\xhat_k - \xstar_k\|^2} \\
& \le \sqrt{1.4^2\sum_k \|(\I - \U \U^\top) \Ustar \tb_k\|^2} \\
& = 1.4 \|(\I - \U \U^\top) \Ustar \tB \|_F \\
& \le  1.4\|(\I - \U \U^\top) \Ustar \|_F \sigmax 
\end{align*}
 Similarly, $\|\Bhat - \U^\top \Xstar\|_F \le 0.4 \delta_t \sigmax.$
(iv) Using Weyl's inequality and $\delta_t <0.1/\kappa^2$, this then implies that $\sigma_{\max}(\B) \le 1.1 \sigmax$ and $\sigma_{\min}(\B) \ge 0.9 \sigmin$.

\bfpara{Bounding $\SE(\U^+,\Ustar)$ by a novel use of fundamental theorem of calculus (Lemma \ref{algebra})}
%
Recall from Algorithm \ref{gdmin} that $\Uhat^+ = \Uhat - (\eta/m) \nabla_U f(\U,\B)$ and $\Uhat^+ \qreq \U^+ \R^+$.
We bound $\SE(\U^+, \Ustar)$  using the fundamental theorem of calculus \cite[Chapter XIII, Theorem 4.2]{lan93},\cite{pr_mc_reuse_meas}, summarized in Theorem \ref{funda_calc}. The use of this result is motivated by its use in \cite{pr_mc_reuse_meas}, and many earlier works, where it is used in a standard way: to bound the Euclidean norm error $\|\x - \xstar\|$ for standard GD to solve the PR problem for recovering a single vector $\xstar$. Thus, at the true solution $\x=\xstar$, the gradient of the cost function was zero. In our case, there are two differences: (i) we need to bound the subspace distance error, and (ii) our algorithm is not standard GD; in particular, this means that $\nabla_U f(\Ustar \Ustar^\top \U, \B) \neq 0$.

To deal with (i) and (ii), we proceed as follows.
 We first bound $\|(\I - \Ustar \Ustar^\top) \Uhat^+\|_F$. To do this, we apply Theorem \ref{funda_calc} on vectorized $\nabla_U f(\U,\B)$ with the pivot being vectorized $\nabla_U f(\Ustar \Ustar^\top \U, \B)$, and use this in the equation for $\Uhat^+$. Next, we project both sides of this expression orthogonal to $\Ustar$ followed by some careful linear algebra. Notice here that  $\nabla_U f(\Ustar \Ustar^\top \U, \B) \neq 0$, because $\B \neq \tB$. Because of this, we get an extra term, $\mathrm{Term2}:=(\I - \Ustar \Ustar^\top) \nabla_U f(\Ustar \Ustar^\top \U, \B)$, in our bound other than the usual term containing the Hessian. We are able to bound it by $\eps \delta_t \sigmax^2$ for any constant small enough $\eps$, by realizing that $\E[\mathrm{Term2}]=0$ (conditioned on past measurements), and that its summands are {\em nice-enough} subexponentials. 
%
Next, we bound $\SE(\Ustar, \U^+)$ by using
\begin{align*}
\SE(\Ustar, \U^+) 
& \le \|(\I - \Ustar \Ustar^\top) \Uhat^+\|_F \|(\R^+)^{-1}\| \\
& = \frac{\|(\I - \Ustar \Ustar^\top) \Uhat^+\|_F}{ \sigma_{\min}(\Uhat^+) } 
\end{align*}
and
$\sigma_{\min}(\Uhat^+)= \sigma_{\min}(\U - (\eta/m) \nabla_\U f(\U,\B)) \ge 1 - (\eta/m) \|\nabla_\U f(\U,\B)\|$.


\bfpara{Bounding the terms in the $\SE(\Ustar,\U^+)$ bound (Lemma \ref{terms_bnds})}
Consider $\|\nabla_\U f(\U,\B)\|$.
Using Lemma \ref{B_lemma}, it can be shown that, for unit vectors $\w,\z$, the maximum sub-exponential norm of any summand of $\w^\top \nabla_\U f(\U,\B) \z$
is bounded by  $\|\xhat_k - \xstar_k\| \cdot \|\b_k\| \le 1.1 \mu^2 \sigmax^2 \delta_t  (r/q)$. Observe that we get this (sufficiently small) bound because of the extra $\b_k^\top$ term in the summands of $\nabla_\U f(\U,\B)$ compared to those in $\nabla_\X \tilde{f}(\X)$. 
This, along with using the sub-exponential Bernstein inequality \cite{versh_book} followed by a standard epsilon-net argument, and  bounding  $\|\E[\nabla_U f]\|$ using $\|\E[\nabla_U f]\|= \|m(\Xhat - \Xstar) \B^\top\| \le m \delta_t \sigmax^2$ (by Lemma \ref{B_lemma}), helps guarantee that $\|\nabla_U f\| \lesssim 2 m \delta_t \sigmax^2$ w.h.p. as long as $mq \gtrsim (n+q)r^2$.
We bound $\|\mathrm{Term2}\|_F$ using similar ideas and the key fact that $\E[\mathrm{Term2}]=0$. This is true because of sample-splitting. We upper and lower bound the eigenvalues of the Hessian, $\mathrm{Hess}$, using similar ideas and the following: for a unit vector $\w$ of length $nr$ and its rearranged unit Frobenius norm matrix $\W$ of size $n \times r$, $\E[\w^\top \mathrm{Hess} \ \w ] = \E[ \sum_\ik ( \a_\ik{}^\top \W \b_k)^2] = m \|\W \B\|_F^2$. Using the bounds on $\sigma_i(\B)$ from Lemma \ref{B_lemma}, this can be upper and lower bounded. 

\subsection{Lemmas for proving GD descent Theorem \ref{iters_thm} and its proof} \label{iters_proof}

Let $\U \equiv \U_{t}$, $\Bhat \equiv \Bhat_{t+1}$. The proof follows using the following 3 lemmas.


\begin{lemma}[Error bound on $\B$ and its implications]
Let $U \equiv \U_t$, $\B \equiv \B_{t+1}$, and
$$\g_k: = \U^\top \xstar_k.$$
Assume that $\SEF(\Ustar, \U_{t}) \le \delta_t$ with $\delta_t < \delta_0 = c/\kappa^2$ (this bound on $\delta_t$ is needed for the second part of this lemma).
Then, w.p. $\ge 1 - q \exp(r - c m)$,
\ben
\item
\begin{align}
\|\g_k - \bhat_k \|   & \leq 0.4 \|\left(\I_n-\U\U^\top \right)\Ustar\tb_k\|
\label{bhatk_bnd}
\end{align}

\item This in turn implies all of the following.
\ben
\item $ \|\xhat_{k}-\xstar_{k}\| \leq 1.4 \|\left(\I-\U\U^\top \right)\Ustar\tb_k\|$

\item $\|\G-\Bhat\|_F   \leq   0.4 \delta_t \sigmax$ and  $\|\Xstar-\Xhat\|_F	\le \sqrt{1.16} \delta_t \sigmax$,

\item  $\|\g_k - \bhat_k \| \leq 0.4 \delta_t \|\tb_k\|$ and $ \|\xhat_{k}-\xstar_{k}\| \le  1.4 \delta_t \|\xstar_k\|$,

\item  $\| \Ustar{}^\top\U \bhat_{k} - \tb_k \|  \le  2.4 \delta_t \|\tb_k\| $, 

\item $\|\bhat_k\| \le   1.1 \mu \sigmax \sqrt{r/q}$.

\item  $\sigma_{\min}(\Bhat) \ge 0.9 \sigmin$ and  $\sigma_{\max}(\Bhat) \le 1.1 \sigmax$,

\een

\een
\label{B_lemma}
\end{lemma}

\begin{proof} See Sec. \ref{B_lemma_proof}. \end{proof}

\begin{lemma}\label{algebra}  
Let $\U \equiv \U_{t}$, $\Bhat \equiv \Bhat_{t+1}$.  Let $\kron$ denote the Kronecker product. We have
\begin{align*}
&\SEF(\U_{t+1},\Ustar) \\
&\qquad \le   \dfrac{ \|\I_{nr} - (\eta/m) \mathrm{Hess}\| \cdot \SEF(\Ustar,\U) + (\eta/m) \|\mathrm{Term2}\|_F }{ 1  - (\eta/m)  \|\mathrm{GradU}\|},  
\end{align*}
where,
\begin{align*}
\mathrm{GradU} & :=  \nabla_\U f(\U,\Bhat)  = \sum_\ik (\y_\ik - \a_\ik{}^\top \U \bhat_k) \a_\ik \bhat_k{}^\top
\\
\mathrm{Term2} & :=  (\I - \Ustar \Ustar^\top) \nabla_\U f((\Ustar\Ustar^\top \U),\Bhat)  \\
& = (\I - \Ustar \Ustar^\top) \sum_\ik (\y_\ik - \a_\ik{}^\top \Ustar\Ustar^\top\U \bhat_k) \a_\ik \bhat_k{}^\top
\\
\mathrm{Hess} & :=  \sum_\ik (\a_\ik \kron \bhat_k) (\a_\ik \kron \bhat_k)^\top  
\end{align*}
\end{lemma}

\begin{proof} See Sec. \ref{algebra_proof} \end{proof}


\begin{lemma}
\label{terms_bnds}
Assume $\SEF(\Ustar, \U) \le \delta_t < \delta_0 =c/\kappa^2$.  Then,%
\ben
\item w.p. at least $1 - \exp( (n+r) - c mq \epsilon_1^2 / r \mu^2 ) - \exp(\log q + r - c m)$,
\[
\|\mathrm{GradU}\| \le   1.5(1.1+ \epsilon_1)  m  \delta_t  \sigmax^2 ;
\]
\item w.p. at least $1 - \exp( nr  - c mq \epsilon_2^2 / r  \mu^2 ) - \exp(\log q + r - c m)$,
\[
\|\mathrm{Term2}\|_F \le    1.1 m \epsilon_2 \delta_t \sigmax^2 ;
\]
\item w.p. at least $1 - \exp( nr \log \kappa  - c mq \epsilon_3^2 / r \kappa^4  \mu^2 ) - \exp(\log q + r - c m)$,%
{\small
\begin{align*}
m (0.65- 1.2 \epsilon_3) \sigmin^2 &\le \lambda_{\min}( \mathrm{Hess} ) \\
&\le \lambda_{\max}( \mathrm{Hess} ) \le   m  (1.1 + \epsilon_3) \sigmax^2.
\end{align*}
}
\een
\end{lemma}

\begin{proof} See Sec. \ref{terms_bnds_proof}. 
\end{proof}


\begin{proof}[Proof of Theorem \ref{iters_thm}]
The proof follows by induction. Base case for $t=0$ is true by assumption. Induction assumption: Assume that, w.p. at least $1-t n^{-10}$,  $\SEF(\Ustar, \U_{t}) \le \delta_t$ with $\delta_t \le \delta_0 =c_0/\kappa^2$.

Set  $ \epsilon_1 = 0.1$, $ \epsilon_3 = 0.01$,  $\epsilon_2 = 0.01 / 1.1 \kappa^2$ and, $c_0 =0.1 / 1.5(1.1+0.1) $.

The upper bound on $\lambda_{\max}(\mathrm{Hess})$ and using $\eta \le 0.5 /\sigmax^2$ implies that
$
\lambda_{\min}(\I_{nr} - (\eta/m) \mathrm{Hess}) = 1 - (\eta/m) \lambda_{\max}(\mathrm{Hess}) \ge 1 - \frac{0.5 (1.1 + 0.01) m\sigmax^2}{m \sigmax^2} > 1-0.555 > 0
$
i.e. $\I_{nr} - (\eta/m) \mathrm{Hess}$ is positive definite.
Thus,
$
\| \I_{nr} - (\eta/m) \mathrm{Hess} \| = \lambda_{\max}(\I_{nr} - (\eta/m) \mathrm{Hess}) =  1 - (\eta/m) \lambda_{\min}(\mathrm{Hess}) \le 1 - (\eta/m) m (0.65- 1.2 \epsilon_3) \sigmin^2  \le 1 - (\eta \sigmax^2) 0.63 / \kappa^2  . 
$

By Lemma \ref{algebra}, Lemma \ref{terms_bnds}, and the above,   
w.p. at least $1 -t n^{-10} - \exp( (n+q) - c mq  / r \mu^2 ) - \exp( nr  - c mq / r \kappa^4 \mu^2 )   - \exp( nr \log \kappa  - c mq  / r  \kappa^4 \mu^2) - \exp(\log q + r - c m)$,
\begin{align*}
& \SEF(\Ustar, \U_{t+1}) \\
& \le \dfrac{ (1 - (\eta \sigmax^2) 0.63 / \kappa^2)\cdot \delta_t  +  (\eta/m)  1.1 m \epsilon_2  \sigmax^2 \delta_t  }{1 - (\eta/m)  1.5(1.1+\epsilon_1)  m    \delta_t  \sigmax^2}  \\
& \le \left( \frac{ 1 - (\eta \sigmax^2)  0.63/\kappa^2  + (\eta \sigmax^2) 0.01/\kappa^2   }{1 - (\eta \sigmax^2) 0.1/\kappa^2} \right)  \delta_t \\
& \le \left( 1 - (\eta \sigmax^2) \frac{0.42}{ \kappa^2} \right) \delta_t
\end{align*}
The second inequality substituted the values of $\epsilon_j$'s and used $\delta_t < \delta_0= 0.1 / ( 1.5(1.1+0.1) \kappa^2)$ for its denominator term. The third inequality used $(1 - (\eta \sigmax^2) 0.1/\kappa^2)^{-1} \le (1 + (\eta \sigmax^2) 0.2/\kappa^2)$ (for $0 < x < 1$, $1/(1-x) \le 1+2x$).

By plugging in the epsilon values in the probability, the above holds w.p.
$\ge 1 - tn^{-10} - 0.2\exp( (n+q) - c mq  / r \mu^2 ) - 0.2\exp( nr  - c mq / r  \mu^2 \kappa^4 )   - 0.2\exp( nr \log \kappa  - c mq / r \mu^2 \kappa^4 ) - \exp(\log q + r - c m)$ .
If $mq \ge C \kappa^4 (n+q)r^2 \log \kappa$ and $m \ge C\max(r, \log q, \log n)$ for a $C$ large enough, then, this probability is $\ge 1- tn^{-10} - 0.2\exp( - c (n+q)) - 0.4 \exp( - c n r)  - n^{-10} >  1 -  (t+1) n^{-10}$.
\end{proof}


\subsection{Proof outline (and novelty) for Initialization Theorem \ref{init_thm}}\label{outline_init}
Recall that we compute $\U_0$ as the top $r$ left singular vectors of $\Xhat_0$ defined in \eqref{newinit} and that this is a truncated version of $\Xhat_{0,full}$.  As noted there, we cannot use $\X_{0,full}$ because its summands are not {\em nice-enough sub-exponentials}. Truncation converts the summands into sub-Gaussian r.v.s. 
For these, we can use the sub-Gaussian Hoeffding inequality \cite[Chap 2]{versh_book} which needs a small enough bound on only the squared sum of the sub-Gaussian norms of the $mq$ summands, and not on their maximum value (as needed by the sub-exponential Bernstein inequality). This is an easier requirement that gets satisfied for our problem.
Of course, truncation also means that the summands of $\Xhat_0$ are not mutually independent (each summand depends on the truncation threshold $\alpha$ which is computed using all measurements $\y_\ik$) and that $\E[\Xhat_0] \neq \Xstar$.
There are two ways to resolve this issue. The first and simpler approach, but one that assumes more sample-splitting is given below in Sec \ref{init_proof}. This assumes that $\alpha$ is a computed using a different independent set of measurements than those used to define the rest of $\Xhat_0$. With this, $\E[\Xhat_0|\alpha] = \Xstar \D(\alpha)$, where $\D$ is a diagonal matrix defined below in Lemma \ref{Wedinlemma} and the summands are independent conditioned on $\alpha$. Thus, we can apply Wedin's $\sin \Theta$ theorem \cite{wedin,spectral_init_review} (given in Proposition \ref{Wedin_sintheta}) on $\Xhat_0$ and $\E[\Xhat_0|\alpha]$ to bound $\SE(\U_0, \Ustar)$, followed by  subGaussian Hoeffding and a standard epsilon-net argument, to bound the terms in this bound.


To avoid sample-splitting for $\alpha$, we need to significantly modify the sandwiching arguments from \cite{twf,lrpr_it} for our setting. This is done in Appendix \ref{init_reuse_proof}. In the previous works, sandwiching was used for a symmetric positive definite (p.d.) matrix. Here we need such an argument for a non-symmetric matrix. Briefly, we do this as follows. We define a matrix $\X_+$ that is such that the span of top $r$ left singular vectors of its expected value equals that of $\Ustar$ and that can be shown to be close to $\Xhat_0$. $\X_+$ is $\Xhat_0$ with $\alpha$ replaced by $\tC(1+\eps) \|\Xstar\|_F^2/q$.
We bound $\|\Xhat_0 - \E[\X_+]\|$ by bounding $\|\X_+ - \Xhat_0\|$ and $\|\X_+ - \E[\X_+]\|$. Bounding the latter is simple.  
Bounding $\|\X_+ - \Xhat_0\|$ requires bounding $\w^\top (\X_+ - \Xhat_0) \z$ for unit vectors $\w, \z$ and this is not straightforward because its summands are not mutually independent. To deal with this, we first bound each summand by its absolute value, and then bound the indicator function term to get a new one that is non-random so that the summands of this new term are mutually independent.
But, its summands are no longer zero mean (because of taking the absolute values), and hence more work is needed to get the desired small enough bound on the expected value of this term.

\subsection{Simpler proof of Theorem \ref{init_thm} that assumes independent measurements used for computing $\alpha$}\label{init_proof}

For the simpler proof given here, assume that we use a different independent set of measurements for computing $\alpha$ than those used for the rest of $\Xhat_{0}$, i.e., let
\[
\alpha = \tC \frac{\sum_\ik  (\y_\ik^{nrmX})^2}{mq}
\]
with $\y_\ik^{nrmX}$ independent of $\{\A_k^{(0)}, \y_k^{(0)} \}$.
With this change, it is possible to compute $\E[\Xhat_0|\alpha]$ easily. But, it does not affect the sample complexity order and so it does not change our theorem statement.
The proof follows by combining the two lemmas and facts given next.

\begin{lemma} \label{Wedinlemma}
Conditioned on $\alpha$, we have the following conclusions.
\ben
\item Let $\zeta$  be a scalar standard Gaussian r.v.. Define
\[
\beta_{k}(\alpha) := \E[\zeta^2 \indic_{\{\|\xstar_{k}\|^2\zeta^2 \leq  \alpha\}}].
\]
Then,
\begin{align}
&\E[\Xhat_0|\alpha] = \Xstar \D(\alpha), \nonumber\\
&\text{ where }  \D(\alpha):=diagonal(\beta_k(\alpha),k \in [q])
\label{X0}
\end{align}
i.e. $\D(\alpha)$ is a diagonal matrix of size $q\times q$ with diagonal entries $\beta_{k}$ defined above.

\item Let $\E[\Xhat_{0}|\alpha] = \Xstar \D(\alpha) \svdeq \Ustar \check\bSigma \Bcheck$ be its $r$-SVD. Then,
\begin{align}\label{Wedin_main}
& \SEF(\U_0,\Ustar) \le \nonumber \\
&   \dfrac{\sqrt{2} \max\left( \| (\Xhat_0 - \E[\Xhat_{0}|\alpha] )^\top \Ustar \|_F , \| (\Xhat_0 - \E[\Xhat_{0}|\alpha] )  \Bcheck{}^\top \|_F \right)}{\sigmin \min_k \beta_k(\alpha) - \|\Xhat_0 - \E[\Xhat_{0}|\alpha] \|}
\end{align}
as long as the denominator is non-negative.
\een
\end{lemma}

\begin{proof} See Sec. \ref{Wedinlemma_proof} \end{proof}

\newcommand{\ev}{\mathcal{E}}

\cred

Define the set $\ev$ as follows
\bea
\ev:= \left\{ \tC(1 - \epsilon_1) \frac{\|\Xstar\|_F^2}{q} \le \alpha \le \tC (1 + \epsilon_1) \frac{\|\Xstar\|_F^2}{q}  \right\}.
\label{def_ev}
\eea

The following fact is an immediate consequence of sub-exponential Bernstein inequality for bounding $|\alpha - \|\Xstar\|_F^2/q|$.

\begin{fact}\label{sumyik_bnd}
$\Pr(\alpha \in \ev) \ge 1 - \exp(- \tc mq \epsilon_1^2):= 1 - p_\alpha$.  Here  $\tc = c/\tC = c / \kappa^2 \mu^2.$
\end{fact}


\cbl

The next lemma bounds the terms of Lemma \ref{Wedinlemma}.

\begin{lemma} \label{init_terms_bnd}
Fix $0 < \epsilon_1 < 1$. Then, 
\ben
\item \label{Xhat0_1}
w.p. at least $1-\exp\left[(n+q)-c\epsilon_1^2mq/\mu^2\kappa^2\right]$, conditioned on $\alpha$, for an $\alpha \in \ev$,
		\[
		\|\Xhat_{0} -\E[\Xhat_{0}|\alpha]\| \leq 1.1 \epsilon_1 \|\Xstar\|_F
		\]
%

\item
\label{Xhat0_Ustar_1}
	w.p. at least $1-\exp\left[ qr - c \epsilon_1^2 mq / \mu^2\kappa^2\right]$, conditioned on $\alpha$, for an $\alpha \in \ev$,
		\[
		\|\left(\Xhat_{0} - \E[\Xhat_{0}|\alpha]\right){}^\top\Ustar\|_F \leq 1.1 \epsilon_1\|\Xstar\|_F
		\]

\item
		\label{lem:init_nom_B_term2}
\label{Xhat0_Bstar_1}
w.p. at least $1-\exp\left[nr - c \epsilon_1^2mq/\mu^2\kappa^2\right]$, conditioned on $\alpha$, for an $\alpha \in \ev$,
		\[
		\|\left(\Xhat_{0} - \E[\Xhat_{0}|\alpha]\right)\Bcheck{}^\top\|_F \leq 1.1  \epsilon_1\|\Xstar\|_F.
		\]


\een
	\end{lemma}

\begin{proof} See Sec. \ref{init_terms_bnd_proof} \end{proof}

We also need to the following fact. 

\begin{fact}
\label{betak_bnd}
		For any $\epsilon_1 \leq 0.1$, 
		$
 \min_k  \E\left[\zeta^2 \indic_{ \left\{ |\zeta| \leq \tC \frac{ \sqrt{1- \epsilon_1} \|\Xstar\|_F }{ \sqrt{q}\|\xstar_{k}\| } \right\} } \right] \geq 0.92.
	$
	\end{fact}

\begin{proof}[Proof of Theorem \ref{init_thm}]
%
Set $\epsilon_1 = 0.4 \delta_0 / \sqrt{r} \kappa $. Define
$
p_0 = 2\exp( (n+q)- c mq \delta_0^2 / r \kappa^2 ) + 2\exp( n r - c mq \delta_0^2 / r \kappa^2 ) + 2\exp( q r - c mq \delta_0^2 / r \kappa^2 ).
$
Recall that  $\Pr(\alpha \in \ev) \ge 1 - p_\alpha$ with
$
p_\alpha = \exp(- \tc mq \epsilon_1^2) = \exp(- c mq \delta_0^2 /r \mu^2 \kappa^2 ).
$

Using Lemma \ref{init_terms_bnd}, conditioned on $\alpha$, for an $\alpha \in \ev$,
\bi
\item  w.p. at least $ 1-p_0$,
$
\|\Xhat_0 - \E[\Xhat_{0}|\alpha]\| \le 1.1 \epsilon_1  \|\Xstar\|_F  =  0.44 \delta_0 \sigmin,  
$ and %
{\small
$
\max\left( \| (\Xhat_0 - \E[\Xhat_{0}|\alpha])^\top \Ustar \|_F , \| (\Xhat_0 - \E[\Xhat_{0}|\alpha])  \Bcheck^\top \|_F \right) \le 0.44 \delta_0 \sigmin
$
}
\item 
$\min_k \beta_k(\alpha) \ge \min_k  \E\left[\zeta^2 \indic_{\{ |\zeta| \leq \tC \frac{ \sqrt{1- \epsilon_1} \|\Xstar\|_F }{ \sqrt{q}\|\xstar_{k}\| } \} } \right]  \ge 0.9
$
The first inequality is an immediate consequence of $\alpha \in \ev$ and the second follows by Fact \ref{betak_bnd}.
\ei


Plugging the above bounds into \eqref{Wedin_main} of Lemma \ref{Wedinlemma}, conditioned on $\alpha$, for any $\alpha \in \ev$, w.p. at least $ 1 - p_0$,
$\SEF(\U_0,\Ustar)  \le  \frac{0.44 \delta_0}{0.9 - 0.44 \delta_0} < \delta_0$ since $\delta_0 < 0.1$. In other words,
\begin{align}
& \Pr\left( \SEF(\U_0,\Ustar)  \ge  \delta_0 | \alpha \right) \le p_0 \ \text{for any $\alpha \in \ev$}.
\label{SD_given_alpha_bnd}
\end{align}
Since (i)
$
\Pr( \SEF(\U_0,\Ustar)  \ge \delta_0 )
\le \Pr( \SEF(\U_0,\Ustar)  \ge  \delta_0 \text{ and }  \alpha \in \ev) +  \Pr (\alpha \notin \ev),
$ and
(ii)
$
\Pr( \SEF(\U_0,\Ustar)  \ge \delta_0 \text{ and } \alpha \in \ev ) \le \Pr(\alpha \in \ev) \max_{\alpha \in \ev}\Pr( \SEF(\U_0,\Ustar)  \le  \delta_0 |\alpha ),
$
thus, using Fact \ref{sumyik_bnd} and \eqref{SD_given_alpha_bnd}, we can conclude that
\[
\Pr \left( \SEF(\U_0,\Ustar)  \ge \delta_0 \right) \le p_0 (1 - p_\alpha)  + p_\alpha \le p_0 + p_\alpha
\]

Thus, for a $\delta_0 < 0.1$, $\SEF(\U_0,\Ustar) < \delta_0$ w.p. at least $ 1- p_0 - p_\alpha = 1 - 2\exp( (n+q)- c mq \delta_0^2 / r \kappa^2 ) - 2\exp( n r - c mq \delta_0^2 / r \kappa^2 ) - 2\exp( q r - c mq \delta_0^2 / r \kappa^2 ) -  \exp(- c mq \delta_0^2 /r \mu^2 \kappa^4 )$.
This is $\ge 1 - 5 \exp(-c(n+q))$ if $mq > C \kappa^2 \mu^2 (n+q)r^2 / \delta_0^2$. This finishes our proof.
\end{proof}

\section{Proofs of all the lemmas}
\label{proof_lemmas}

\subsection{Basic tools used}
Our proofs use the following results and definitions:

\begin{theorem}[Wedin $\sin \Theta$ theorem for Frobenius norm subspace distance \cite{wedin,spectral_init_review}[Theorem 2.3.1]]
For two $n_1 \times n_2$ matrices $\M^*$, $\M$, let $\Ustar, \U$ denote the matrices containing their top $r$ singular vectors and let $\Vstar^\top, \V^\top$ be the matrices of their right singular vectors (recall from problem definition that we defined SVD with the right matrix transposed). Let $\sigma^*_r, \sigma^*_{r+1}$ denote the $r$-th and $(r+1)$-th singular values of $\M^*$. 
If $\|\M - \M^*\| \le \sigma^*_r - \sigma^*_{r+1}$, then
\begin{align*}
&\SE(\U, \Ustar) \\
&\qquad \le \frac{\sqrt{2} \max(\|(\M - \M^*)^\top \Ustar\|_F, \|(\M - \M^*)^\top \Vstar^\top\|_F    )}{\sigma^*_r - \sigma^*_{r+1} - \|\M - \M^*\|}
\end{align*}
\label{Wedin_sintheta}
\end{theorem}

\begin{theorem}[Fundamental theorem of calculus \cite{lan93}[Chapter XIII, Theorem 4.2], \cite{pr_mc_reuse_meas}]
For two vectors $\z_0, \zstar \in \Re^d$, and a differentiable vector function $g(\z) \in \Re^{d_2}$,
\[
g(\z_0) - g(\zstar) = \left( \int_{\tau=0}^1 \nabla g (\z(\tau) ) d\tau  \right) (\z_0 - \zstar),  
\]
\text{ where }
\[
 \z(\tau) = \zstar + \tau (\z_0 - \zstar).
\]
Observe that $\nabla_{\z} g(\z)$ is a $d_2 \times d$ matrix.
\label{funda_calc}
\end{theorem}

\begin{definition}
For any $n\times r$ matrix $\Z$, let $\Zvec$ denote the $nr$ length vector formed by arranging all $r$ columns of $\Z$ one below the other.
Thus, for  $n$-length and $r$-length vectors $\a$ and $\b$,
\bi
\item $(\a \b^\top)_{vec} = \a \kron \b$ with $\kron$ being the Kronecker product;
\item $\a^\top \U \b = \trace(\a^\top \U \b) =  \trace(\b \a^\top \U) = \langle (\a \b^\top), \U \rangle =\langle \a \kron \b, \U_{vec} \rangle $;
\ei
 $f(\Uvec, \B) = \sum_\ik ((\a_\ik \kron \b_k)^\top \Uvec - \y_\ik)^2$ and
\bea
(\nabla_{\U}f(\U, \B) )_{vec} = \nabla_{\U_{vec}}f(\U_{vec}, \B)
\label{gradU_vecmat}
\eea
\end{definition}

\begin{definition}
At various places, $\nabla f(\U, \Bhat)$ is short for $\nabla_\U f(\U, \Bhat) =  \sum_\ik \a_\ik \bhat_k{}^\top(\a_\ik{}^\top \U \bhat_k - \y_\ik)$ and similarly $\nabla f(\Uvec,\Bhat)$ is short for $\nabla_\Uvec f(\Uvec,\Bhat) =  \sum_\ik  (\a_\ik \kron \b_k)  ((\a_\ik \kron \b_k)^\top \Uvec - \y_\ik)$.
\end{definition}

\begin{definition}
For any vector $\w$, we use $\w(k)$ to denote its $k$-th entry.
\end{definition}

\begin{definition}
Everywhere we use $\S_{nr}$ to denote both the set of matrices  $\{\W \in \Re^{n \times r}: \|\W\|_F = 1 \}$ and the set of these matrices vectorized $\{\w \in \Re^{nr}: \|\w\|=1 \}$. We also switch between the two sometimes. In the entire writing below, $\w = \W_{vec}$.
\end{definition}

All the high probability bounds for initialization use subGaussian Hoeffding inequality, while those for GD lemmas use the sub-exponential Bernstein inequality, both are from \cite{versh_book}.
In addition, these lemmas also use the following results to ``epsilon-net'' extend a bound holding for a fixed unit norm $\W$ (or $\w$) to all unit norm $\W$s (or $\w$s)

\begin{prop}[Epsilon-netting for bounding  $\max_{\w \in \S_n, \z \in \S_r} |\w^\top \M \z |$]
For an $n \times r$ matrix $\M$ and fixed vectors $\w, \z$ with, $\w \in \S_n$ and $\z \in \S_r$, suppose that $|\w^\top \M \z | \le b_0$ w.p. at least $1-p_0$.
Consider an $\eps_{net}$ net covering $\S_n $ and $\S_r$, $\bar\S_n$, $\bar\S_r$
Then w.p. at least $1 - (1+2/\eps_{net})^{n+r} p_0$,
\bi
\item $\max_{\w \in \bar\S_n, \z \in \bar\S_r} |\w^\top \M \z | \le b_0$ and
\item $\max_{\w \in \S_n, \z \in \S_r} |\w^\top \M \z | \le \frac{1}{1 - 2\eps_{net} - \eps_{net}^2} b_0$.
\ei
Using $\eps_{net}=1/8$, this implies the following simpler conclusion: \\
W.p. at least $1 - 17^{n+r} p_0= 1- \exp( (\log 17) (n+r) ) \cdot p_0$, $\max_{\w \in \S_n, \z \in \S_r} |\w^\top \M \z | \le 1.4 b_0$.
\label{epsnet_Mwz}
\end{prop}

\begin{proof} The proof follows that of Lemma 4.4.1 of \cite{versh_book} \end{proof}

\begin{prop}[Epsilon-netting for bounding  $\max_{\W \in \S_{nr}}  \langle \M, \W$]
For an $n \times r$ matrix $\M$ and a fixed $n \times r$ matrix  $\W \in \S_{nr}$ (unit Frobenius norm matrix), suppose that $\langle \M, \W \rangle \le b_0$ w.p. at least $1-p_0$.
Consider an $\eps_{net}$ net  covering $\S_{nr}$, $\bar\S_{nr}$.
Then w.p. at least $1 - (1+2/\eps_{net})^{nr} p_0$,
\bi
\item $\max_{\W \in \bar\S_{nr}} \langle \M, \W \rangle \le b_0$ and
\item $\max_{\W \in \S_{nr}}  \langle \M, \W \rangle \le \frac{1}{1 - \eps_{net}} b_0$.
\ei
Using $\eps_{net}=1/8$, this implies the following simpler conclusion: \\
w.p. at least $1 - 17^{nr} p_0 = 1- \exp( (\log 17) (nr) ) \cdot p_0$, $\max_{\W \in \S_{nr}}  \langle \M, \W \rangle \le 1.2 b_0$.
\label{epsnet_MW}
\end{prop}

\begin{proof} The proof follows exactly as that of Exercise 4.4.3 of \cite{versh_book} \end{proof}

\begin{prop}[Epsilon-netting for upper and lower bounding $\sum_\ik \langle \M_\ik, \W \rangle^2$ over all $\W \in \S_{nr}$]
For an $n \times r$ matrices $\M_\ik$ and a fixed $\W \in \S_{nr}$, suppose that, w.p. at least $1-p_0$,
\[
b_1 \le \sum_\ik \langle \M_\ik, \W \rangle^2 \le b_2
\]
Consider an $\eps_{net}$ net  covering $\S_{nr}$, $\bar\S_{nr}$.
Then, w.p. at least $1 - (1+2/\eps_{net})^{nr} p_0$, 
\[
 \max_{\W \in \S_{nr}}  \sum_\ik \langle \M_\ik, \W \rangle^2 \le \frac{1}{1 - \eps_{net}^2 - 2 \eps_{net}} b_2  
\]
and
\[
 \min_{\W \in \S_{nr}}  \sum_\ik \langle \M_\ik, \W \rangle^2  \ge b_1 - 2 \eps_{net} \cdot \frac{1}{1 - \eps_{net}^2 - 2 \eps_{net}} b_2
\]
Picking $\eps_{net} = b_1/ (8 b_2)$ guarantees that the above lower bound is non-negative.
In particular, it implies the following:
\\
w.p. at least $1 - (24b_2/b_1)^{nr} p_0 = 1 - \exp(C nr \log(b_2/b_1) ) \cdot p_0$,
$
0.8 b_1 \le \min_{\W \in \S_{nr}}  \sum_\ik \langle \M_\ik, \W \rangle^2  \le \max_{\W \in \S_{nr}}  \sum_\ik \langle \M_\ik, \W \rangle^2 \le 1.4 b_2  
$
\label{epsnet_MWsquared}
\end{prop}

\begin{proof}
By union bound, for all $\bar\W \in \bar\S_{nr}$,
$b_1 \le \sum_\ik \langle \M_\ik, \bar\W \rangle^2 \le b_2$ holds w.p. at least $1 - (1+2/\eps_{net})^{nr} p_0$.

Proof for the upper bound:
Let $\gamma^* =  \max_{\W \in \S_{nr}}  \sum_\ik \langle \M_\ik, \W \rangle^2$.
Writing $\W = \bar\W + (\W - \bar\W)$ where $\bar\W$ is the closest point to $\W$ on $\bar\S_{nr}$, we have
$
\sum_\ik \langle \M_\ik, \W \rangle^2 = \sum_\ik \langle \M_\ik, \bar\W \rangle^2 + \sum_\ik \langle \M_\ik, (\W - \bar\W) \rangle^2 + 2 \sum_\ik \langle \M_\ik, \bar\W \rangle \cdot \sum_\ik \langle \M_\ik, (\W - \bar\W) \rangle$ and
$\|(\W - \bar\W)\|_F \le \eps_{net}$.

Rewriting $(\W - \bar\W) = (\W - \bar\W) \cdot (\W - \bar\W) / \|(\W - \bar\W)\|_F$ and using the fact that  $(\W - \bar\W)/\|(\W - \bar\W)\|_F \in \S_{nr}$ and $\|(\W - \bar\W)\|_F \le \eps_{net}$ and using Cauchy-Schwarz for the third term in the above expression, we have
\[
\gamma^* \le b_2 + \eps_{net}^2 \gamma^* + 2 \sqrt{\gamma^* \cdot \eps_{net}^2 \gamma^*} = b_2 + \eps_{net}^2 \gamma^* + 2\eps_{net} \gamma^*
\]
Thus, $\gamma^* \le 1/(1 - \eps_{net}^2 - 2 \eps_{net}) \cdot b_2$.

Proof for the lower bound:
Let $\beta^* =  \min_{\W \in \S_{nr}}  \sum_\ik \langle \M_\ik, \W \rangle^2$. Proceeding as above, we have
\[
\beta^* \ge b_1 - 2 \sqrt{\gamma^* \cdot \eps_{net}^2 \gamma^*} = b_1 - 2\eps_{net} \gamma^*
\]
\end{proof}

\subsection{Proving GD iterations' lemmas:  Proof of Lemma \ref{algebra} (algebra lemma)}
	\label{algebra_proof}


Recall that $\Uvec$ denotes the vectorized $\U$.
 We use this so that we can apply the simple vector version of the fundamental theorem of calculus \cite[Chapter XIII, Theorem 4.2]{lan93},\cite[Lemma 2 proof]{pr_mc_reuse_meas} (given in Theorem \ref{funda_calc}) on the $nr$ length vector $\nabla f(\Uvec,\Bhat)$, and so that the Hessian can be expressed as an $nr \times nr$ matrix.


We apply Theorem \ref{funda_calc} with $\z_0 \equiv \Uvec$, $\zstar \equiv (\Ustar \Ustar^\top \U)_{vec}$, and $g(\z) = \nabla f(\z,\Bhat)$. Thus $d = d_2 = nr$ and $\nabla g(\z)$ is the Hessian of $f(\z,\Bhat)$ computed at $\z$.
Let $\U(\tau) := \Ustar \Ustar{}^\top \U + \tau (\U - \Ustar \Ustar{}^\top \U)$.
Applying the theorem,
\begin{align}
&\nabla f(\U_{vec},\Bhat) - \nabla f((\Ustar\Ustar{}^\top\U)_{vec},\Bhat) \nonumber\\
&= ( \int_{\tau=0}^1 \nabla_{\Uvec}^2 f(\U(\tau)_{vec}, \Bhat) d\tau ) ( \U_{vec} - (\Ustar\Ustar{}^\top \U)_{vec})
\label{funda_thm_app_0}
\end{align}
where
\begin{align}\label{Hess_compute}
\nabla_{\Uvec}^2 f(\U(\tau)_{vec}, \Bhat)  = \sum_\ik (\a_\ik \kron \bhat_k) (\a_\ik \kron \bhat_k)^\top : =  \ \mathrm{Hess} \
\end{align}
This is an $nr \times nr$ matrix. Because the cost function is quadratic, the Hessian is  constant w.r.t. $\tau$. Henceforth, we refer to it as $ \ \mathrm{Hess} \ $. With this, the above simplifies to
\begin{align}
&\nabla f(\U_{vec},\Bhat) - \nabla f((\Ustar\Ustar{}^\top\U)_{vec},\Bhat) \nonumber\\
&=  \ \mathrm{Hess} \  ( \U_{vec} - (\Ustar\Ustar{}^\top \U)_{vec}) =  \mathrm{Hess} \ (\P \U)_{vec}
\label{funda_thm_app}
\end{align}
with
\[
\P : = \I  - \Ustar \Ustar^\top
\]
denoting the $n \times n$ projection matrix to project orthogonal to $\Ustar$.
%
This proof is motivated by a similar approach used in \cite[Lemma 2 proof]{pr_mc_reuse_meas} to analyze GD for standard PR. However, there the application was much simpler because $f(.)$ was a function of one variable and at the true solution the gradient was zero, i.e., $\nabla f(\xstar) =  \bm{0}$. In our case $\nabla f(\Ustar \Ustar^\top \U, \Bhat) \neq  \bm{0}$ because $\Bhat \neq \tB$. But we can show that $\E[ (\I - \Ustar \Ustar^\top) \nabla f(\Ustar \Ustar^\top \U, \Bhat)] = \bm{0}$ and this helps us get the final desired result.

From Algorithm \ref{gdmin}, recall that $\Uhat^+ = \U - (\eta/m) \nabla f(\U,\Bhat)$. Vectorizing this equation, 
and using \eqref{funda_thm_app}, we get
		\begin{align}
		(\Uhat^+)_{vec}
		& = \U_{vec}  - (\eta/m) \nabla f(\U_{vec}, \Bhat) \nonumber  \\
		& = \U_{vec} - (\eta/m)  \ \mathrm{Hess} \     ( \P \U)_{vec}  \nonumber\\
		&\qquad- (\eta/m) \nabla f((\Ustar\Ustar{}^\top\U)_{vec},\Bhat) ) 
\label{Uplusvec_eq}
		\end{align}
We can prove our final result by using \eqref{gradU_vecmat} and the following simple facts:
\ben
\item  For an $n \times n$ matrix $\M$, let 
$
\mathrm{big}(\M) :=  \I_r \kron \M.
$
be an $nr \times nr$ block diagonal matrix with $\M$ in the diagonal blocks.
%
For any $n \times r$ matrix $\Z$, 
		\begin{align}
\label{MZ}
& \mathrm{big}(\M) \Zvec =		(\M \Z)_{vec} 
\end{align}

\item Since $\P$ is idempotent, $\P = \P^2$. Also, because of its block diagonal structure, $\mathrm{big}(\M^2)= (\mathrm{big}(\M))^2$.  Thus,
		\begin{align}
		\mathrm{big}(\P)& = \mathrm{big}(\P^2) = (\mathrm{big}(\P))^2  = \mathrm{big}(\P)) \I_{nr}  (\mathrm{big}(\P)
\label{P_props}
		\end{align}
\een

Left multiplying both sides of \eqref{Uplusvec_eq} by $\mathrm{big}(\P)$, and using \eqref{MZ}, \eqref{P_props}, and \eqref{gradU_vecmat}, 
		\begin{align*}
		& \mathrm{big}(\P) (\Uhat^+)_{vec}
		 = \mathrm{big}(\P) \U_{vec} -  (\eta/m) \mathrm{big}(\P)   \ \mathrm{Hess} \   (\P \U)_{vec}  \\
		&\qquad- (\eta/m)  \mathrm{big}(\P) \nabla f((\Ustar\Ustar{}^\top\U)_{vec},\Bhat) \\
		& = \mathrm{big}(\P)\I_{nr} \mathrm{big}(\P) \U_{vec} -  (\eta/m) \mathrm{big}(\P)   \ \mathrm{Hess} \     \mathrm{big}(\P) \U_{vec} \\
		&\qquad  - (\eta/m) \mathrm{big}(\P)  \nabla f((\Ustar\Ustar{}^\top\U)_{vec},\Bhat)  \\
		& = \mathrm{big}(\P) (\I_{nr}  - (\eta/m)   \ \mathrm{Hess})  \mathrm{big}(\P) \U_{vec}  \\
		&\qquad -  (\eta/m) \mathrm{big}(\P) \nabla f((\Ustar\Ustar{}^\top\U)_{vec},\Bhat).
		\end{align*}
		Thus, using $\|\mathrm{big}(\P)\| = \|\P\| = 1$, \eqref{MZ}, and \eqref{gradU_vecmat},
		\begin{align}
 \|  ( \P \Uhat^+)_{vec}\|
		& \le \|\I_{nr}  - (\eta/m)   \ \mathrm{Hess} \   \|  \ \|  (\P \U)_{vec}  \|  \nonumber\\
		&+ (\eta/m) \|  (  \nabla f((\Ustar\Ustar{}^\top\U),\Bhat))_{vec} \|
		\end{align}
Converting the vectors to matrices, using $||\M_{vec}|| = ||\M||_F$, and substituting for $\P$,
		\begin{align*}
		&\|(\I - \Ustar \Ustar{}^\top)\Uhat^+\|_F \\
		&\le  \|\I_{nr}  - (\eta/m)   \ \mathrm{Hess} \  \| \ \|(\I - \Ustar \Ustar{}^\top)\U\|_F\\
		&\qquad + (\eta/m) \|(\I - \Ustar \Ustar{}^\top) \nabla f((\Ustar\Ustar{}^\top\U),\Bhat) \|_F
		\end{align*}
		Since $\Uhat^+ \qreq \U^+ \R^+$ and since $\|\M_1 \M_2\|_F \le \|\M_1\|_F \|\M_2\|$, this means that
$$\SEF(\Ustar, \U^+) \le  \|(\I - \Ustar \Ustar{}^\top)\Uhat^+\|_F  \| (\R^+)^{-1}\|.$$
		Since $\| (\R^+)^{-1}\| = 1/\sigma_{\min}(\R^+) = 1/\sigma_{\min}(\Uhat^+)$, using $\Uhat^+ = \U   - (\eta/m) \nabla f(\U, \Bhat)$,
		\begin{align*}
		\| (\R^+)^{-1}\| &= \frac{1}{\sigma_{\min}(\U   - (\eta/m) \nabla f(\U, \Bhat))} \\
		&\le \frac{1}{1 - (\eta/m) \|\nabla f(\U, \Bhat)\|}
		\end{align*}
where we used $\sigma_{\min}(\U   - (\eta/m) \nabla f(\U, \Bhat)) \ge \sigma_{\min}(\U)   - (\eta/m) \|\nabla f(\U, \Bhat)\| = 1- (\eta/m) \|\nabla f(\U, \Bhat)\|$ for the last inequality.
Combining the last three equations above proves our lemma.

\subsection{Proof of GD iterations' lemmas: Proof of Lemma \ref{terms_bnds}} \label{terms_bnds_proof}  

\subsubsection{Upper and Lower bounding the Hessian eigenvalues and hence HessTerm}
	\label{subsec:hessian}
First assume the event that implies that the conclusions of Lemma \ref{B_lemma} hold.

	Recall from \eqref{Hess_compute} that 
$
	 \ \mathrm{Hess} \ :=	\nabla_{\tilde\U_{vec}}^2 f(\tilde\U_{vec}; \Bhat) =   \sum_\ik    (\a_\ik \otimes \bhat_{k}) (\a_\ik \otimes \bhat_k){}^\top.
$
Since $  \ \mathrm{Hess} \ $ is a positive semi-definite matrix,
		$
		 \lambda_{\min}\left(  \ \mathrm{Hess} \  \right) = \min_{\w\in\S_{nr}} \w{}^\top  \ \mathrm{Hess} \  \ \w $ and $\lambda_{\max}\left(  \ \mathrm{Hess} \  \right) = \max_{\w\in\S_{nr}} \w{}^\top   \ \mathrm{Hess} \  \ \w.
		$
For a fixed $\w\in\S_{nr}$,
		\[
		\w{}^\top   \ \mathrm{Hess} \  \ \w =  \sum_\ik   (\a_\ik{}^\top \W \bhat_k)^2
		\]	
where $\W$ is an $n \times r$ matrix with $\|\W\|_F = 1$.
		Clearly $(\a_\ik{}^\top \W \bhat_k)^2 $ are mutually independent sub-exponential random variables (r.v.) with sub-exponential norm $K_\ik\leq \|\W\bhat_{k}\|^2$. Also, $\E[ (\a_\ik{}^\top \W \bhat_k)^2] = \|\W \bhat_k\|^2$ and thus $\E[ \sum_\ik   (\a_\ik{}^\top \W \bhat_k)^2]= m\|\W\Bhat\|_F^2$.
Applying the sub-exponential Bernstein inequality, Theorem 2.8.1 of \cite{versh_book}, for a fixed $\W\in\S_{nr}$ yields
		\begin{align*}
		&\Pr\left\{ \Big|\sum_\ik \big|\a_\ik{}^\top\W\bhat_{k}\big|^2 - m\|\W\Bhat\|_F^2 \Big| \geq t \right\}\\
		&\qquad \leq \exp\left[-c\min\left( \frac{t^2}{\sum_\ik K_\ik^2},~\frac{t}{\max_\ik K_\ik} \right)\right].
		\end{align*}
		We set $t = \epsilon_3 m \sigmin^2$. By Lemma \ref{B_lemma}, $\|\bhat_{k}\|^2 \leq 1.1 \mu^2 \sigmax^2 (r/q)  = 1.1 \kappa^2 \mu^2 \sigmin^2 (r/q)$. Thus,
		\begin{align*}
		\frac{t^2}{\sum_\ik K_\ik^2} &\geq \frac{\epsilon_3^2m^2 \sigmin^4}{\sum_\ik \|\W\bhat_{k}\|^4  } \geq \frac{\epsilon_3^2 m \sigmin^4}{ \max_k \|\bhat_{k}\|^2 \sum_{k}\|\W\bhat_{k}\|^2  } \\
		&\geq \frac{\epsilon_3^2m \sigmin^4}{ \mu^2  \sigmax^2 (r/q) 1.1. \sigmax^2  } =  c\epsilon_3^2mq/ r \mu^2\kappa^4
		\end{align*}
Here we used $\sum_{k}\|\W\bhat_{k}\|^2 = \|\W \Bhat\|_F^2  \le \|\W\|_F \|\Bhat\|_2 \le 1.1. \sigmax$ using the bound on $\|\Bhat\|_2$ from Lemma \ref{B_lemma}.
Also,
		\begin{align*}
		\frac{t}{\max_\ik K_\ik} &\geq \frac{\epsilon_3 m \sigmin^2}{\max_\ik \|\W\bhat_{k}\|^2} \geq \frac{\epsilon_3 m \sigmin^2}{1.1 \mu^2 \sigmax^2 (r/q)} \\
		&=  c\epsilon_3 mq/r\mu^2\kappa^2.
		\end{align*}
Therefore, for a fixed $\W \in \S_{nr}$, w.p. $1-\exp\left[-c\epsilon_3^2 mq/r \mu^2\kappa^4 \right]$ we have
		\begin{align}\label{bnd_term}
		\Big|\sum_\ik \big|\a_\ik{}^\top\W\bhat_{k}\big|^2 - m\|\W\Bhat\|_F^2 \Big| \leq \epsilon_3 m \sigmin^2.
		\end{align}
and hence, by  Lemma \ref{B_lemma}, w.p. $1-\exp\left[-c\epsilon_3^2 mq/r \mu^2\kappa^4 \right]$,
		\begin{align}
\label{fixedW_upperbnd}
		&\sum_\ik \big|\a_\ik{}^\top\W\bhat_{k}\big|^2 \leq m\|\W\Bhat\|_F^2 + \epsilon_3m\sigmin^2 \nonumber \\
		&\qquad\leq m\|\Bhat\|^2 + \epsilon_3 m \sigmin^2
		\leq m(1.1 + \epsilon_3/\kappa^2) \sigmax^2.
		\end{align}
and
\begin{align}
\label{fixedW_lowerbnd}
		&\sum_\ik \big|\a_\ik{}^\top\W\bhat_{k}\big|^2  \ge  m\|\W\Bhat\|_F^2 - \epsilon_3m\sigmin^2 \nonumber \\
		&\qquad \ge  0.9 m \sigmin^2 + \epsilon_3 m\sigmin^2
		\geq m(0.9 - \epsilon_3) \sigmin^2.
		\end{align}

To extend these bounds to all $\W \in \S_{nr}$ we apply Proposition \ref{epsnet_MWsquared} with $b_1 \equiv m(0.9 - \epsilon_3) \sigmin^2$ and $b_2 \equiv  m(1.1 + \epsilon_3/\kappa^2) \sigmax^2$. Applying it we can conclude that, given the event that the claims of Lemma \ref{B_lemma} holds,
 w.p. at least $1 - \exp( nr \log \kappa  - c mq \epsilon_3^2 / r \mu^2 \kappa^4 )$,  
\begin{align*}
m (0.7 - 1.2 \epsilon_3) \sigmin^2 &\le \lambda_{\min}(  \ \mathrm{Hess} \  ) \\&\le \lambda_{\max}(  \ \mathrm{Hess} \  ) \le  m  (1.1 + \epsilon_3) \sigmax^2
\end{align*}
Using the probability from Lemma \ref{B_lemma},
the above bound holds  w.p. at least $1 - \exp( nr \log \kappa  - c mq \epsilon_3^2 / r \mu^2 \kappa^4 )  - \exp(\log q + r - c m)$.

\subsubsection{Bounding the GradU Term}
	\label{subsec:sigmin_R_U}
 We have
		$
		\|\nabla f(\U,\Bhat)\| = \max_{\z\in\S_n, \w\in\S_r} \z{}^\top \nabla f(\U,\Bhat)\w.
		$
		For a fixed $\z\in\S_n, \w\in\S_r$ we have
		\begin{align*}
		&\z{}^\top \left( \nabla f(\U,\Bhat) - \E[\nabla f(\U,\Bhat)] \right) \w \\
		&\qquad= \sum_\ik \left[ \left(\a_\ik{}^\top\U\bhat_{k}-\y_\ik \right)\left(\a_\ik{}^\top\z\right)\left(\w{}^\top\bhat_{k}\right) - \E[.] \right]
		\end{align*}
where $\E[.]$ is the expected value of the first term.
Clearly, the summands are independent sub-exponential r.v.s with norm $K_\ik \leq C \|\xhat_{k}-\xstar_{k}\|\|\bhat_{k}\|$. We apply the sub-exponential Bernstein inequality, Theorem 2.8.1 of \cite{versh_book}, with $t = \epsilon_1 \delta_t m \sigmax^2$. 
To apply this, we use bounds on $\|\bhat_{k}\|$,  $\|\Xstar - \Xhat\|_F$ and $\|\xhat_{k}-\xstar_{k}\|$ from  Lemma \ref{B_lemma} to show that
		\begin{align*}
		\frac{t^2}{\sum_\ik K_\ik^2} &\geq c\frac{\epsilon_1^2\delta_t^2 m^2 \sigmax^4}{ m \max_k \|\bhat_{k}\|^2  \sum_k\|\xhat_{k}-\xstar_{k}\|^2 } \\&\geq c\frac{\epsilon_1^2 \delta_t^2 m \sigmax^4}{C \mu^2 \sigmax^2 (r/q) \|\Xhat - \Xstar\|_F^2 } \\&\geq c \frac{\epsilon_1^2  \delta_t^2 mq \sigmax^4}{C \mu^2 \sigmax^2 r   \delta_t^2 \sigmax^2 } = c \epsilon_1^2 \frac{mq }{ r \mu^2}.
		\end{align*}
and
		\[
		\frac{t}{\max_\ik K_\ik} \geq c \frac{\epsilon_1\delta_tm\sigmax^2}{C\delta_t \sigmax^2 \mu^2 (r/q) } \ge  c \epsilon_1   \frac{mq }{ r \mu^2 }.
		\]
		Therefore, for a fixed $\z\in\S_n, \w\in\S_{r}$ w.p.  $1-\exp (-c\epsilon_1^2mq/r \mu^2 )$,
		\begin{align*}
		\z{}^\top\left( \nabla f(\U,\Bhat) - \E[\nabla f(\U,\Bhat)] \right)\w &\leq  \epsilon_1 \delta_t m \sigmax^2  
		\end{align*}
Since $\nabla f(\U,\Bhat) = \sum_\ik \a_\ik\a_\ik{}^\top(\xhat_{k}-\xstar_{k})\bhat_{k}{}^\top$,
\[
\E [\nabla f(\U,\Bhat) ] = m\sum_k (\xhat_{k}-\xstar_{k})\bhat_{k}{}^\top = m\left(\Xhat -\Xstar \right)\Bhat{}^\top.
\]
Using the bounds on  $\|\Xstar - \Xhat\|_F$ and  $\|\Bhat\|$ from Lemma \ref{B_lemma},
		\begin{align*}
		\|\E [\nabla f(\U,\Bhat)  ] \| &= m \| (\Xhat -\Xstar )\Bhat{}^\top\| \\&\leq m \|\Xhat -\Xstar\|~\|\Bhat\| \\& \le m \|\Xhat -\Xstar\|_F ~\|\Bhat\| \\& \leq 1.1 m\delta_t\sigmax^2
		\end{align*}
Hence, for a fixed $\z\in\S_n, \w\in\S_{r}$ w.p. $1-\exp\left[-c \epsilon_1^2  mq/r \mu^2 \right]$ we have
		\[
	| \z^\top \nabla f(\U,\Bhat) \w | \leq  (1.1 + \epsilon_1)  m\delta_t \sigmax^2.
		\]
Applying Proposition \ref{epsnet_Mwz}, this implies that, w.p. $1-\exp ((n+r) (\log 17) -c\epsilon_1^2mq/r \mu^2 )$, $ \max_{\z\in\S_n, \w\in\S_r} \z{}^\top \nabla f(\U,\Bhat)\w \le 1.4 (1.1 + \epsilon_1)  m\delta_t \sigmax^2.$

\subsubsection{Bounding Term2}
First, since $\mathrm{Term2} = (\I - \Ustar \Ustar{}^\top) \sum_\ik  \a_\ik (\a_\ik{}^\top\Ustar (\Ustar{}^\top\U\bhat_{k} - \tb_k) ) \bhat_{k}{}{}^\top$, and $\E[\a_\ik \a_\ik{}^\top] = \I$,
\[
\E[\mathrm{Term2}] = 0
\]
We have
		\begin{align*}
		&\|(\I - \Ustar \Ustar{}^\top) \nabla f((\Ustar\Ustar{}^\top\U),\Bhat) \|_F \\&\qquad=\max_{\W\in\S_{nr}} \langle (\I - \Ustar \Ustar{}^\top) \nabla f((\Ustar\Ustar{}^\top\U),\Bhat) ,~\W\rangle
		\end{align*}
For a fixed $n \times r$ matrix $\W$ with unit Frobenius norm,
		\begin{align*}
		&\langle (\I - \Ustar \Ustar{}^\top) \nabla f((\Ustar\Ustar{}^\top\U),\Bhat) ,~\W\rangle \\&\qquad= \sum_\ik  \left(\a_\ik{}^\top\Ustar (\Ustar{}^\top\U\bhat_{k} - \tb_k) \right)\left(\a_\ik{}^\top(\I - \Ustar \Ustar{}^\top)\W\bhat_{k}\right)
		\end{align*}
Observe that the summands are independent, zero mean, sub-exponential r.v.s with sub-exponential norm  $K_\ik\leq C \| \Ustar{}^\top\U\bhat_{k} - \tb_k \| \|(\I - \Ustar \Ustar{}^\top)\W\bhat_{k}\| \le
\| \Ustar{}^\top\U\bhat_{k} - \tb_k \| \| \W\bhat_{k}\|$.
We can now apply the sub-exponential Bernstein inequality Theorem 2.8.1 of \cite{versh_book}. 
Let $t=\epsilon_2\delta_t m \sigmax^2$.
Using the bound on $\| \Ustar{}^\top\U\bhat_{k} - \tb_k \|$ from Lemma \ref{B_lemma} followed by Assumption \ref{right_incoh} (right incoherence), and also the bound on $\|\Bhat\|$ from Lemma \ref{B_lemma}, 
		\begin{align*}
		\frac{t^2}{\sum_\ik K^2_\ik} &\geq \frac{\epsilon_2^2\delta_t^2m^2\sigmax^4}{ \delta_t^2 \sigmax^2 \mu^2 (r/q)   \sum_\ik  \|\W\bhat_{k}\|^2 } \\&\geq \frac{\epsilon_2^2 m^2 \sigmax^2}{C    \mu^2 (r/q)   m \|\W\Bhat\|_F^2  }\geq \frac{\epsilon_2^2m^2 \sigmax^2}{ \mu^2 (r/q)   m  \sigmax^2 } \\&\geq c \epsilon_2^2mq/ r\mu^2 ,
		\end{align*}
and
		\[
		\frac{t}{\max_\ik K_\ik}\geq  \frac{\epsilon_2\delta_tm\sigmax^2}{C\delta_t \kappa^2 \mu^2 \sigmax^2 (r/q)}\geq c \epsilon_2 mq/ (r  \kappa^2 \mu^2) .
		\]
		Thus, by the sub-exponential Bernstein inequality,
for a fixed $\W\in\S_{nr}$, w.p. $1-\exp (-c\epsilon_2^2mq/r \kappa^2 \mu^2 )$,
		\[
		\langle (\I - \Ustar \Ustar{}^\top) \nabla f((\Ustar\Ustar{}^\top\U),\Bhat) ,~\W\rangle \leq \epsilon_2 \delta_t m \sigmax^2.
		\]
Applying Proposition \ref{epsnet_MW}, w.p. at least $1-\exp ( nr-c\epsilon_2^2mq/r \kappa^2 \mu^2  )$, $\max_{\W\in\S_{nr}} \langle (\I - \Ustar \Ustar{}^\top) \nabla f((\Ustar\Ustar{}^\top\U),\Bhat), \W \rangle \le 1.2 \epsilon_2 \delta_t m \sigmax^2.$

\newcommand{\tempAk}{\A_k}
\newcommand{\temptopAk}{\A_k{}^\top}

\subsection{Proof of GD iterations' lemmas: Proof of Lemma \ref{B_lemma}, all parts other than the first part} \label{B_lemma_proof}
%
Recall that $\g_k = \U^\top \xstar_k = \U^\top \Ustar \tb_k$, and $\G = \U^\top \Ustar \tB$.

Using the $\SE$ bound and the first part, $\|\g_k - \bhat_k\| \le 0.4  \delta_t \|\tb_k\|$.

Since $\xstar_k  - \xhat_k = \U \g_k + (\I - \U \U^\top) \xstar_k  - \U \bhat_k =  \U (\g_k  - \bhat_k) + (\I - \U \U^\top) \xstar_k $, using \eqref{bhatk_bnd},
\begin{align*}
 \|\xstar_k - \xhat_k\| &  \le 	\|\g_k - \bhat_k\|   + \|(\I - \U \U^\top) \Ustar \tb_k\|
\le 1.4 \delta_t \|\tb_k\|.
\end{align*}

$\| \Ustar{}^\top\U\bhat_{k} - \tb_k \| = \| \Ustar\Ustar{}^\top\U\bhat_{k} - \Ustar\tb_k \|  = \|\U \bhat_k - (\I - \Ustar\Ustar{}^\top)\U\bhat_{k} - \Ustar\tb_k \| =  \|\xhat_k - (\I - \Ustar\Ustar{}^\top)\U\bhat_{k} - \xstar_k \|
\le \|\xhat_{k}-\xstar_{k}\| + \|(\I - \Ustar \Ustar{}^\top) \U \bhat_k\|  \le  2.4 \delta_t \|\tb_k\| $

Bounding $\|\G - \Bhat\|_F$ and $\|\Xstar - \Xhat\|_F$:
Since $\sum_k \|\M \tb_k\|^2 = \|\M \tB\|_F^2 \le \|\M\|_F^2 \|\tB\|^2 =\|\M\|_F^2  \sigmax^2$, we can use the first bound from \eqref{bhatk_bnd} to conclude that
	\begin{align*}
\|\G - \Bhat\|_F^2 & = \sum_k \|\g_k - \bhat_k\|^2  \\
&  \le 0.4^2 \sum_k \|  (\I - \U \U^\top) \Ustar \tb_k\|^2
\\&=  0.4^2  \|  (\I - \U \U^\top) \Ustar \tB \|_F^2
 \le   0.4^2 \delta_t^2 \sigmax^2
	\end{align*}
and, similarly, 
\begin{align*}
\|\Xstar - \Xhat\|_F^2 &\le  \sum_k \|\g_k - \bhat_k\|^2  +  \sum_k \|(\I - \U \U^\top) \Ustar \tb_k\|^2  \\&\le (0.4^2 +1^2)  \delta_t^2 \sigmax^2
\end{align*}

Incoherence of $\bhat_k$'s: Using the bound on $\|\bhat_k - \g_k\|$, and using $\|\g_k\| \le \|\tb_k\|$ and the right incoherence assumption,
\begin{align*}
\|\bhat_k\|  & = \| (\bhat_k - \g_k + \g_k) \| \le (1 + 0.4 \delta_t) \|\tb_k\| \le 1.04 \sigmax \sqrt{r/q}.
	\end{align*}

Lower and Upper Bounds on $\sigma_i(\Bhat)$): Using the bound on $\|\G - \Bhat\|_F$ and using $\SE(\U,\Ustar) \le \delta_t < c/\kappa$,
\begin{align*}
\sigma_{\min}(\Bhat) &\geq \sigma_{\min}(\G) - \|\G - \Bhat\| \\
		&	\geq \sigma_{\min}(\U^\top\Ustar) \sigma_{\min}( \tB ) - \|\G - \Bhat\|_F \\
		& \ge  \sqrt{1- \|\Ustar_\perp{}^\top \U\|^2} \sigmin  - 0.4 \delta_t \sigmax \\
		& \ge  \sqrt{1- \delta_t^2} \sigmin  -0.4 \delta_t \sigmax \ge 0.9 \sigmin
	\end{align*}
since we assumed $\delta_t \le \delta_0 < 0.1/ \kappa$.
Similarly,
\begin{align*}
\|\Bhat\|= \sigma_{\max}(\Bhat) &	\leq \sigma_{\max}(\U^\top\Ustar) \sigma_{\max}( \tB ) + \|\G - \Bhat\|_F \\
		& \le  \sigmax   + 0.4 \delta_t \sigmax \le 1.1 \sigmax
	\end{align*}

\subsection{Proof of GD iterations' lemmas: Proof of Lemma \ref{B_lemma}, first part} \label{proof_B_lemma_part1}
We bound $\|\g_k - \bhat_{k}\|$ here.
Recall that $\g_k = \U^\top \xstar_k$.
Since $\y_k = \tempAk\xstar_{k} = \tempAk\U\U{}^\top\xstar_{k} + \tempAk(\I-\U\U{}^\top)\xstar_{k}$, therefore
		\begin{align*}
		\bhat_{k}  &= \left(\U{}^\top\temptopAk\tempAk\U\right)^{-1}(\U{}^\top\temptopAk)\tempAk\U\U{}^\top\xstar_{k} \\&\qquad+ \left(\U{}^\top\temptopAk\tempAk\U\right)^{-1}(\U{}^\top\temptopAk) \tempAk(\I-\U\U{}^\top)\xstar_{k},\\
		&=\left(\U{}^\top\temptopAk\tempAk\U\right)^{-1}\left(\U{}^\top\temptopAk\tempAk\U\right)\U{}^\top\xstar_{k}  \\&\qquad+ \left(\U{}^\top\temptopAk\tempAk\U\right)^{-1}(\U{}^\top\temptopAk) \tempAk(\I-\U\U{}^\top)\xstar_{k},\\
		&=\g_k + \left(\U{}^\top\temptopAk\tempAk\U\right)^{-1}(\U{}^\top\temptopAk) \tempAk(\I-\U\U{}^\top)\xstar_{k}.
		\end{align*}
		Thus,
		\begin{align}
		\|\bhat_{k} - \g_k\| &\leq \|\left(\U{}^\top\temptopAk\tempAk\U\right)^{-1}\| \nonumber\\& \qquad\times ~\|\U{}^\top\temptopAk \tempAk(\I-\U\U{}^\top)\xstar_{k}\|.
\label{gk_bhatk_bnd}
		\end{align}
Using standard results from \cite{versh_book}, one can show the following:
\ben
\item
W.p. $\ge 1-q\exp\left(r-cm\right)$, for all $k\in[q]$, $\min_{\w \in \S_r} \sum_i \big|\a_\ik{}^\top\U\w\big|^2 \ge 0.7 m$ and so
\begin{align*}
\|\left(\U{}^\top\temptopAk\tempAk\U\right)^{-1}\| &= \frac{1}{\sigma_{\min}\left(\U{}^\top\temptopAk\tempAk\U \right)} \\&= \frac{1}{\min_{\w\in\S_{r}} \sum_i \langle \U^\top \a_\ik , \w \rangle^2 }\\& \le \frac{1}{0.7 m}
\end{align*}		
\item W.p. at least $1-q\exp(r-cm)$, $ \forall k\in[q]$,
\[
\|\U{}^\top\temptopAk \tempAk(\I-\U\U{}^\top)\xstar_{k}\| \leq    0.15 m \| (\I-\U\U{}^\top)\xstar_{k}\|
\]
\een
Combining the above two bounds and \eqref{gk_bhatk_bnd}, w.p. at least $1 - 2\exp(\log q + r - c m)$, $\forall k\in[q]$,
\[
\|\g_k - \bhat_{k}\|    \leq 0.4 \|\left(\I_n-\U\U^\top\right)\Ustar\tb_k\|. 
\]
This completes the proof. We explain next how to get the above two bounds.

The first bound above follows by a restatement of Theorem 4.6.1 of \cite{versh_book}. Or, it follows more directly by using $\E[\sum_i \big|\a_\ik{}^\top\U\w\big|^2] = m$, applying the sub-exponential Bernstein inequality \cite[Theorem 2.8.1]{vershynin} to bound the deviation from this mean,  and then applying Proposition \ref{epsnet_MWsquared} with $n \equiv 1, r \equiv r$ (epsilon net argument).

The second bound is obtained as follows. Notice that
\begin{align*}
	&\|\U{}^\top\temptopAk \tempAk(\I-\U\U{}^\top)\xstar_{k}\| \\
	&\qquad= \max_{\w\in\S_{r}} \w{}^\top\U{}^\top\temptopAk \tempAk(\I-\U\U{}^\top)\xstar_{k} \\&\qquad=  \max_{\w\in\S_{r}} \sum_i (\a_\ik{}^\top\U\w)(\a_\ik{}^\top(\I-\U\U{}^\top)\xstar_{k} )
\end{align*}
Clearly $\E\left[ \U{}^\top\temptopAk \tempAk(\I-\U\U{}^\top)\xstar_{k}\right] = \U{}^\top (\I-\U\U{}^\top)\xstar_{k} = 0$.
Moreover, the summands are products of sub-Gaussian r.v.s and are thus sub-exponential. Also, the different summands are mutually independent and zero mean. 
Applying  sub-exponential Bernstein with $t = \epsilon_0 m \|(\I-\U\U{}^\top)\xstar_{k}\|$ for a fixed $\w\in\S_{r}$,
$$ | \sum_i (\a_\ik{}^\top\U\w)(\a_\ik{}^\top(\I-\U\U{}^\top)\xstar_{k} ) | \leq \epsilon_0 m \| (\I-\U\U{}^\top)\xstar_{k}\|$$
w.p. at least $1-\exp (-c\epsilon_0^2m )$.
Setting $\epsilon_0 = 0.1$, this implies that the above is bounded by $0.1 m  \| (\I-\U\U{}^\top)\xstar_{k}\|$ w.p. at least $1-\exp (-c m)$.  By Proposition \ref{epsnet_MW} with $n \equiv 1, r \equiv r$, the above is bounded by $0.12 m \| (\I-\U\U{}^\top)\xstar_{k}\|$ for all $\w \in \S_r$ w.p. at least $1-\exp(r -c m)$. Using a union bound over all $q$ columns, the bound holds for all $q$ columns w.p.  at least $1- q \exp(r -c m)$.

\subsection{Proof of Initialization lemmas/facts: Proof of Lemma \ref{Wedinlemma}} \label{Wedinlemma_proof}

%

To see why \eqref{X0} holds, it suffices to show that $\E[(\Xhat_0)_k |\alpha] = \xstar_k \beta_k(\alpha)$ for each $k$. The easiest way to see this is to express $\xstar_k = \|\xstar_k\| \Q_k \e_1$ where $\Q_k$ is an $n \times n$ unitary matrix with first column $\xstar_k/\|\xstar_k\|$; and to use the fact that $\tilde\a_\ik:=\Q_k^\top \a_\ik$ has the same distribution as $\a_\ik$, both are $\n(0,\I_n)$.
Using $\Q_k \Q_k^\top = \I$, $(\Xhat_0)_k = (1/m) \sum_i \Q_k \Q_k^\top  \a_\ik  \a_\ik^\top \|\xstar_k\| \Q_k \e_1 \indic_{ \|\xstar_k\| |\a_\ik^\top  \Q_k \e_1 | \le \sqrt{\alpha} } = (1/m) \sum_i \Q_k \|\xstar_k\| \tilde\a_\ik \tilde\a_\ik (1) \indic_{  |\tilde\a_\ik(1)| \le \sqrt{\alpha}/\|\xstar_k\\|}$. Thus $\E[((\Xhat_0)_k] =(1/m) m \Q_k \|\xstar_k\| \e_1 \E[ \zeta^2 \indic_{|\zeta| < \sqrt{\alpha}/\|\xstar_k\\|} ]$.
This follows because $\E[\a \a(1) \indic_{ |\a(1)| < \beta } = \e_1 \E[\a(1)^2 \indic_{ |\a(1)| < \beta }]$.

Recall that $\tC = 9 \kappa^2 \mu^2$ and $\tc = c/\tC$ for a $c<1$.
Recall also that $\Xstar \svdeq \Ustar \bSigma \Bstar$ and  $\E[\Xhat_{0}|\alpha] \svdeq \Ustar \check\bSigma \Bcheck$. Thus, using \eqref{X0}, $\check\bSigma = \bSigma \Bstar \D \Bcheck{}^\top$. Hence,
\begin{align*}
\sigma_r(\E[\Xhat_{0}|\alpha]) &=  \sigma_{\min}(\check\bSigma) \\&=\sigma_{\min}(\bSigma \Bstar \D \Bcheck{}^\top) \\&\ge \sigma_{\min}(\bSigma)\sigma_{\min}(\Bstar)\sigma_{\min}(\D)\sigma_{\min}(\Bcheck{}^\top) \\&= \sigmin \cdot 1 \cdot (\min_k \beta_k (\alpha)) \cdot 1
\end{align*}
Also, $\sigma_{r+1}(\E[\Xhat_{0}]) = 0$ since it is a rank $r$ matrix. Thus, using Wedin's $\sin \Theta$ theorem for the Frobenius norm subspace distance $\SE$ \cite{wedin,spectral_init_review}[Theorem 2.3.1, second row] (specified in Theorem \ref{Wedin_sintheta} above) applied with $\M \equiv \Xhat_0$, $\M^* \equiv \E[\Xhat_{0}]$ we get \eqref{Wedin_main}.

\subsection{Proof of Initialization lemmas and facts: Proof of Lemma \ref{init_terms_bnd}} \label{init_terms_bnd_proof}


	\begin{proof}[Proof of first part of Lemma \ref{init_terms_bnd}]
The proof involves an application of the sub-Gaussian Hoeffding inequality, Theorem 2.6.2 of \cite{versh_book}, followed by an epsilon-net argument. The application of sub-Gaussian Hoeffding uses conditioning on $\alpha$ for $\alpha \in \ev$.
For $\alpha \in \ev$, $\alpha \leq  \sqrt{\tC(1+\epsilon_1)}\|\Xstar\|_F/\sqrt{q}$ and this helps get a simple probability bound. Since $\alpha$ is independent of all $\a_\ik, \y_\ik$'s used in defining $\Xhat_0$, the conditioning does not change anything else in our proof. For example, the different summands are mutually independent even conditioned on it.

We have,
		\[
		\|\Xhat_{0} -\E[\Xhat_{0}|\alpha]\| = \max_{\z\in\S_n, \w\in\S_q} \langle  \Xhat_{0} -\E[\Xhat_{0}|\alpha], ~\z\w{}^\top\rangle.
		\]
For a fixed $\z\in\S_n, \w\in\S_q$, we have
		\begin{align*}
		&\langle  \Xhat_{0} -\E[\Xhat_{0}|\alpha], ~\z\w{}^\top\rangle \\&\qquad= \frac{1}{m} \sum_\ik \w(k)\y_\ik(\a_\ik{}^\top\z)\indic_{ \{|\y_\ik|^2 \leq \alpha \} } \\&\qquad  - \E\left[\w(k)\y_\ik(\a_\ik{}^\top\z)\indic_{ \{|\y_\ik|^2 \leq \alpha \} } \right] .
		\end{align*}
The summands are mutually independent, zero mean sub-Gaussian r.v.s with sub-Gaussian norm $K_\ik \le C |\w(k)| \sqrt{\alpha} / m$. For $\alpha \in \ev$, $\alpha \leq  \sqrt{\tC(1+\epsilon_1)}\|\Xstar\|_F/m \sqrt{q}$.
Let $t=\epsilon_1 \|\Xstar\|_F$. Then, for any $\alpha \in \ev$,
		\[
		\frac{t^2}{\sum_\ik K_\ik^2} \geq \frac{\epsilon_1^2\|\Xstar\|_F^2}{\sum_\ik \tC(1+\epsilon_1) \w(k)^2 \|\Xstar\|_F^2/m^2 q} \geq \frac{\epsilon_1^2mq}{C\mu^2\kappa^2}
		\]
since $\sum_k \w(k)^2 = \|\w\|^2 = 1$.
		Thus, for a fixed $\z\in \S_n, \w\in \S_q$, by sub-Gaussian Hoeffding, we conclude that, conditioned on $\alpha$, for any $\alpha \in \ev$, w.p. at least $1-\exp\left[-c\epsilon_1^2mq/\mu^2\kappa^2\right]$,
		\[
		\langle  \Xhat_{0} -\E[\Xhat_{0}|\alpha], ~\z\w{}^\top\rangle \leq C \epsilon_1 \|\Xstar\|_F.
		\]
		The rest of the  proof follows by a standard epsilon net argument summarized in Proposition \ref{epsnet_Mwz}. Applying it, conditioned on $\alpha$, for any $\alpha \in \ev$,
 w.p. at least $1-\exp\left[ (n+q) -c\epsilon_1^2mq/\mu^2\kappa^2\right]$, $\max_{\z\in\S_n, \w\in\S_q} \langle  \Xhat_{0} -\E[\Xhat_{0}|\alpha], ~\z\w{}^\top\rangle \leq 1.4  C \epsilon_1 \|\Xstar\|_F.$
	\end{proof}

\begin{proof}[Proof of second part of Lemma \ref{init_terms_bnd}]
We have
\[
		\|\left(\Xhat_{0} - \E[\Xhat_{0}|\alpha]\right){}^\top\Ustar\|_F  =
\max_{\W\in\S_{qr}} \langle \W, \left(\Xhat_{0} - \E[\Xhat_{0}|\alpha]\right){}^\top\Ustar  \rangle
\]
For a fixed $\W \in \S_{qr}$,
		\begin{align*}
		&\langle \W, \left(\Xhat_{0} - \E[\Xhat_{0}|\alpha]\right){}^\top\Ustar  \rangle \\&\qquad= \trace\left(\W{}^\top \left(\Xhat_{0} - \E[\Xhat_{0}|\alpha]\right){}^\top\Ustar\right)
		 \\&\qquad= \frac{1}{m}\sum_\ik \left( \y_\ik(\a_\ik{}^\top\Ustar\w_k)\indic_{\left\{|\y_\ik|^2 \leq \alpha \right\} } - \E[.] \right)
		\end{align*}
Conditioned on $\alpha$, for an $\alpha \in \ev$, the summands are independent zero mean sub-Gaussian r.v.s with subGaussian norm $K_\ik \leq \sqrt{\alpha} \|\w_k\|/m \le \sqrt{\tC(1+\epsilon_1)}\|\Xstar\|_F\|\w_k\|/m\sqrt{q}$. Thus,
\[
\sum_\ik K_\ik^2 \le  m \tC(1+ \epsilon_1) \|\W\|_F^2 \|\Xstar\|_F^2 / m^2q = \tC   \|\Xstar\|_F^2 / mq
\]
Applying the sub-Gaussian Hoeffding inequality Theorem 2.6.2 of \cite{versh_book}, for a fixed $\W \in \S_{qr}$, conditioned on $\alpha$, for an $\alpha \in \ev$, w.p. $1-\exp\left[-\epsilon_1^2mq/C\mu^2\kappa^2\right]$,
		\[
		\trace\left(\W{}^\top\left(\Xhat_{0} - \E[\Xhat_{0}|\alpha]\right){}^\top\Ustar \right) \leq \epsilon_1 \|\Xstar\|_F.
		\]
The rest of the  proof follows by a standard epsilon net argument summarized in Proposition \ref{epsnet_MW}.
Applying Proposition \ref{epsnet_MW}, conditioned on $\alpha$, for an $\alpha \in \ev$, w.p. at least $1-\exp\left[qr -c \epsilon_1^2mq/ \mu^2\kappa^2\right]$, $\max_{\W \in \S_{qr}} \trace\left(\W{}^\top\left(\Xhat_{0} - \E[\Xhat_{0}|\alpha]\right){}^\top\Ustar \right) < 1.2 \epsilon_1 \|\Xstar\|_F$.
	\end{proof}

	\begin{proof}[Proof of third part of Lemma \ref{init_terms_bnd}]
We have
		\[
		\|\left(\Xhat_{0} - \E[\Xhat_{0}|\alpha]\right)\Bcheck{}^\top\|_F = \max_{\W\in\S_{nr}} \langle \left(\Xhat_{0} - \E[\Xhat_{0}|\alpha]\right)\Bcheck{}^\top,~\W\rangle.
		\]
For a fixed $\W\in\S_{nr}$ we have,
		\begin{align*}
		&\langle \left(\Xhat_{0} - \E[\Xhat_{0}|\alpha]\right)\Bcheck{}^\top,~\W\rangle \\&\qquad= \frac{1}{m}\sum_\ik \left( \y_\ik (\a_\ik{}^\top\W\bcheck_{k})\indic_{\left\{  |\y_\ik|^2 \leq \alpha \right\} } - \E[.]\right)
		\end{align*}
where $\E[.]$ is the expected value of the first term. 
Conditioned on $\alpha$, for an $\alpha \in \ev$, the summands are independent, zero mean, sub-Gaussian r.v.s with subGaussian norm   $K_\ik \leq C \sqrt\alpha \|\W\bcheck_k\|  \leq C \sqrt{\tC(1+\epsilon_1)}\|\Xstar\|_F\|\W\bcheck_k\|/m\sqrt{q}$. Thus, by applying the  sub-Gaussian Hoeffding inequality Theorem 2.6.2 of \cite{versh_book},
with $t=\epsilon_1  \|\Xstar\|_F$, and using $\|\W\Bcheck\|_F = 1$ (holds since $\Bcheck$ contains orthormal rows which are right singular vectors of $\E[\X_0|\alpha]$), conditioned on  $\alpha$, for an $\alpha \in \ev$,
we will get that,
		\[
		\frac{t^2}{\sum_\ik K_\ik^2} \geq \frac{m^2\epsilon_1^2 \|\Xstar\|_F^2}{\sum_\ik  \tC(1+\epsilon_1)\|\Xstar\|_F^2\|\W\bcheck_{k}\|^2/q } = \frac{mq\epsilon_1^2}{C\mu^2\kappa^2},
		\]
w.p. $1-\exp\left[- c\epsilon_1^2mq/(\mu^2\kappa^2)\right]$. Here we used the fact that $\Bcheck\Bcheck{}{}^\top=\I$ and thus $\|\W\Bcheck\|_F^2 = 1$.
		$$
		\langle \left(\Xhat_{0} - \E[\Xhat_{0}|\alpha]\right)\Bcheck{}^\top,~\W\rangle \leq C \epsilon_1 \|\Xstar\|_F.
		$$
Applying Proposition \ref{epsnet_MW}, conditioned on $\alpha$, for an $\alpha \in \ev$, w.p. at least $1-\exp\left[nr -c\epsilon_1^2mq/(\mu^2\kappa^2)\right]$, $\max_{\W \in \S_{nr}} \langle \left(\Xhat_{0} - \E[\Xhat_{0}|\alpha]\right)\Bcheck{}^\top,~\W\rangle \le 1.2 C \epsilon_1 \|\Xstar\|_F.$
	\end{proof}

\subsection{Proof of Initialization lemmas and facts: Proof of Facts}\label{facts_proof}

\begin{proof}[Proof of Fact \ref{sumyik_bnd}] Apply sub-exponential Bernstein.  \end{proof}

\begin{proof}[Proof of Fact \ref{betak_bnd}]
		Let $\gamma_k = \frac{\sqrt{\tC(1-\epsilon_1)}\|\Xstar\|_F}{\sqrt{q}\|\xstar_{k}\|}$. 
Since $\tC = 9\mu^2\kappa^2$ and $\|\xstar_{k}\|^2\leq \mu^2\kappa^2\|\Xstar\|_F^2/q$ (Assumption \ref{right_incoh}) thus
\[
\gamma_k \geq 3.
\]
Now,
		\begin{align*}
		 \E\left[ \zeta^2\indic_{\left\{ |\zeta| \leq \gamma_k \right\}} \right]
		=& 1 - \E\left[ \zeta^2\indic_{\left\{ |\zeta| \geq \gamma_k \right\}} \right]\\
		\ge &  1 - \frac{2}{\sqrt{2\pi}}\int_{3}^{\infty}z^2\exp(-z^2/2)dz \\
		\geq & 1 - \frac{2e^{-1/2}}{\sqrt{\pi}}\int_{3}^{\infty}z\exp(-z^2/4)dz
		\\&=  1 - \frac{2e^{-11/4}}{\sqrt{\pi}}  \geq 0.92.
		\end{align*}
		The first inequality used $\gamma_k \ge 3$. The second used the fact that $z\exp(-z^2/4) \leq \sqrt{2e}$ for all $z\in \Re$.
	\end{proof}

In all the proofs above, notice that the only thing we used about $\Bcheck$ is the fact that its rows contain singular vectors and thus $\Bcheck \Bcheck^\top = \I$ and so $\sigma_r(\Bcheck) = \sigma_1(\Bcheck) = 1$. We never required incoherence for it


\begin{algorithm}[t]
\caption{\small{The AltGD-Min-LRPR algorithm.  
}} 
\label{AltGD-Min_lrpr}
\begin{algorithmic}[1]
   \State {\bfseries Input:} $\ym_k, \A_k, k \in [q]$
 \State {\bfseries Parameters:}  GD step size, $\eta$; Number of iterations, $T$ 

\State {\bfseries Sample-split:} Partition the measurements and measurement matrices into $2T+1$ equal-sized disjoint sets: one set for initialization and $2T$ sets for the iterations. Denote these by ${\ym_k}^{(\tau)}, \A_k^{(\tau)}, \tau=0,1,\dots 2T$.

   \State {\bfseries Initialization:}
\State Compute $\U_0$ as the top $r$ singular vectors of $\Y_U : = \frac{1}{mq}\sum_\ik ( \ym_\ik)^2 \a_\ik \a_\ik^\top \indic_{ \left\{ (\ym_\ik)^2 \leq \tC \frac{1}{mq}\sum_\ik (\ym_\ik)^2 \right\}}$.
 with $\ym_\ik \equiv {\ym_\ik}^{(0)}, \a_\ik \equiv \a_\ik^{(0)}$.

\State {\bfseries GDmin Iterations:}

   \For{$t=1$ {\bfseries to} $T$}

   \State  Let $\U \leftarrow \U_{t-1}$.
\State {\bfseries Update $\bhat_k, \xhat_k$: } For each $k \in [q]$, set $(\bhat_k)_{t}  \leftarrow  \mathrm{RWF}( {\ym_k}^{(t)}, (\U^\top \A_k^{(t)}), T_{RWF,t})$. Set $(\xhat_k)_{t}    \leftarrow \U (\bhat_k)_{t}$

\State {\bfseries Estimate gradient w.r.t. $\U$: }  With $\ym_\ik \equiv {\ym_\ik}^{(T+t)},\a_\ik \equiv  \a_\ik^{(T+t)}$,
\bi
\item compute  $\hat\y_\ik:=\ym_\ik \hat \cb_\ik$ with $\hat\cb_\ik = phase(\a_\ik{}^\top \xhat_k)$ and
\item compute $\widehat{\mathrm{GradU}} =  \sum_\ik (\hat\y_\ik - \a_\ik{}^\top \xhat_k)_{t} ) \a_\ik (\bhat_k)_t{}^\top$
\ei

 \State Set $\ds \Uhat^+   \leftarrow \U - (\eta/m) \widehat{\mathrm{GradU}} $

 \State {\bfseries Orthornormalize to get new $\U$: }  Compute $\Uhat^+ \qreq \U^+ \R^+$.  Set $\U_t \leftarrow \U^+$.
     \EndFor
\end{algorithmic}
\end{algorithm}

\section{Extension to Low Rank Phase Retrieval (LRPR)} \label{algo_thm_proof_lrpr}
In LRPR, recall that, we measure $\ym_k = |\A_k \xstar_k|$. 
This problem commonly occurs in dynamic phaseless imaging applications such as Fourier ptychography.
Because of the magnitude-only measurements, we can recover each column only up to a global phase uncertainty.  We use $\dist(\xstar,\xhat): = \min_{\theta \in [-\pi, \pi]} \|\xstar - e^{-j \theta} \xhat\|$ to quantify this phase invariant distance \cite{pr_altmin,rwf}. Also, for a complex number, $z$, we use $\bar{z}$ to denote its conjugate and we use $phase(z):= z /|z|$.%

\subsection{AltGD-Min-LRPR algorithm} 
With three simple changes that we explain next, the AltGD-Min approach also solves LRPR and provides the fastest existing solution for it.
First, observe that because of the magnitude-only measurements, we cannot use $\Xhat_0$ with $\y_\ik$ replaced by $\ym_\ik$ for initialization. The reason is $\E[\a_\ik \ym_\ik] = 0$  and so $\E[ \a_\ik \ym_\ik \indic_{\ym_\ik \le \sqrt{\alpha}}] = 0$ too. In fact, because of this, it is not even possible to define a different matrix $\Xhat$ whose expected value can be shown to be close to $\Xstar$. Instead, we have to use the initialization approach of \cite{lrpr_it}. This is given in line 5 of Algorithm \ref{AltGD-Min_lrpr}. The matrix $\Y_U$ is such that its expected value is close to $\Xstar \Xstar{}^\top + c \I$. This fact is used to argue that its top $r$ singular vectors span a subspace that is close to that spanned by columns of $\Ustar$.

Next, consider the GDmin iterations. We use the following idea to deal with the magnitude-only measurements: $\ym_\ik:=|\y_\ik|$.  Let $\cb_\ik:= \mathrm{phase}(\a_\ik{}^\top \xstar_k)$. Then, clearly,
\[
\y_\ik = \cb_\ik \ym_\ik
\]
and $\ym_\ik = \bar{\cb}_\ik \y_\ik$.
We do not observe $\cb_\ik$, but we can estimate it using $\xhat_k$ which is an estimate of $\xstar_k$. Using the estimated phase, we can get an estimate $\hat\y_\ik$ of $\y_\ik$.
We replace $\nabla_\U f(\U,\B)$ by its estimate which uses $\hat\y_\ik= \ym_\ik \hat\cb_\ik$, with $\hat\cb_\ik = phase(\a_\ik{}^\top \xhat_k)$, to replace $\y_\ik$. See line 10 of Algorithm \ref{AltGD-Min_lrpr}.

Lastly, because of the magnitude-only measurements, the update step for updating $\bhat_k$s is no longer an LS problem. We now need to solve an $r$-dimensional standard PR problem: $\min_{\bhat} \|\ym_k - |\A_k \U \b| \|^2$.
This can be solved using any of the order-optimal algorithms for standard PR, e.g., Truncated Wirtinger Flow (TWF) \cite{twf} or Reshaped WF (RWF) \cite{rwf}. For concreteness, we assume that RWF is used. We should point out here that we only need to run $T_{RWF,t}$ iterations of RWF at outer loop iteration $t$, with $T_{RWF,t}$ set below in our theorem (we set this to ensure that the error level of this step is of order $\delta_t$).
The entire algorithm, AltGD-Min-LRPR, is summarized in Algorithm \ref{AltGD-Min_lrpr}.

\subsection{Main Result}
We can prove the following result with simple changes to the proof of Theorem \ref{gdmin_thm}.

\begin{theorem}
Consider Algorithm \ref{AltGD-Min_lrpr}. Set $\eta = c / \sigmax^2$, $ \tC = 9 \kappa^2 \mu^2$, $T = C \kappa^2 \log(1/\epsilon)$, and $T_{RWF,t} = C(t + c \log r)$.
Assume that Assumption \ref{right_incoh} holds.
If 
\[
mq \ge C \kappa^6 \mu^2  (n + q) r^2 (r + \log(1/\epsilon) \log \kappa )
\]
and  $m \ge C \max(\log q, \log n) \log(1/\epsilon)$,
then,
w.p. $1-  n^{-10}$,
$	\SEF(\Ustar,\U_T) \le \epsilon$, $\dist(  (\xhat_k)_T,\xstar_k) \le \epsilon \|\xstar_k\|$ for all $k \in [q]$, and  $\sum_k \dist^2(  (\xhat_k)_T,\xstar_k) \le \epsilon^2 \sigmax^2.$%
\label{gdmin_lrpr_thm}
\end{theorem}

We prove this result in Sec. \ref{lrpr_proof}.
Notice the $\log(1/\epsilon)$ in the sample complexity of Theorem \ref{gdmin_thm} is now replaced by $(r+\log(1/\epsilon))$. The reason is because of the different initialization approach which needs $nr^3$ samples instead of $nr^2$.
This is needed because PR is a more difficult problem: we cannot define a matrix $\Xhat_0$ for it for which $\E[\Xhat_0]$ is close to $\Xstar$. 

Observe that AltGD-Min-LRPR has the same sample complexity as that for the AltMin solution from \cite{lrpr_best}. But its time complexity is better by a factor of $\log(1/\eps)$ making it the fastest solution for LRPR. Also, we should mention here that, for solutions to the two related problems -- sparse PR (phaseless but global measurements) and LRMC (linear but non-global measurements) -- that have been extensively studied for nearly a decade,  the best sample complexity guarantees for iterative (and hence fast) algorithms are sub-optimal. The best sparse PR guarantee \cite{cai} requires $m$ to be of order $s^2$ for the initialization step. Here $s$ is the sparsity level. LRPR has both phaseless and non-global measurements. This is why its initialization step needs two extra factors of $r$ compared to the optimal. Once initialized close enough to the true solution, it is well known that a PR problem behaves like a linear one. This is true for AltGD-Min-LRPR too. 

\cred
Consider a comparison with use of a standard PR approach to recover each column of $\Xstar$ individually. If  TWF \cite{twf} or RWF \cite{rwf} were used for this, this would require $m \gtrsim n$. In comparison, ignoring log factors, our solution for LRPR needs $m \gtrsim (n/q) r^3$.  Thus, the use of altGD-min is a better idea when the rank, $r$, of the matrix $\Xstar$ is small enough so that $q \gtrsim r^3$.
\cbl 



\subsection{Proof of Theorem \ref{gdmin_lrpr_thm}} \label{lrpr_proof}

For the initialization, we use the bound from \cite{lrpr_it}.
\begin{lemma}[\cite{lrpr_it}]
	\label{bounding_U_lrpr}
Let $\SE_2(\U_0,\Ustar) = \|(\I - \Ustar \Ustar^\top) \U_0\|$. Pick a $\deltinit < 0.1$. Then, w.p. at least 
$	1 - 2\exp\left( n (\log 17)   -c \frac{\deltinit^2 mq}{ \kappa^4 r^2} \right)  - 2 \exp\left(-c \frac{\deltinit^2 mq}{ \kappa^4  \mu^2 r^2} \right)  $,
\[
	\SE_2(\U_0,\Ustar) \leq \deltinit \text{ and so }  \SEF(\U_0,\Ustar) \le \sqrt{r} \deltinit.
\]
\end{lemma}

For the iterations, without loss of generality, as also done in past works on PR, e.g., \cite{pr_altmin,twf,rwf,lrpr_best}, to make things simpler, we assume that, for each $k$, $\xstar_k$ is replaced by $\bar{z} \xstar_k$ where $z = \mathrm{phase}(\langle \xstar_k, \xhat_k \rangle)$. 
With this,  $\dist(\xstar_k, \xhat_k) = \|\xstar_k - \xhat_k\|$.

We modify Lemma \ref{algebra} using the following idea. Let $\U  = \U_t$ and $\B = \B_t$. For LRPR, the GD step uses an approximate gradient w.r.t. the old cost function $f(\U,\B)$. Let
\begin{align*}
\mathrm{Err}: = \widehat{\mathrm{GradU}} - \mathrm{GradU}.  
\end{align*}
Here $\widehat{\mathrm{GradU}} = \sum_\ik (\hat\y_\ik - \a_\ik{}^\top \xhat_k) ) \a_\ik \bhat_k{}^\top$ and  $\mathrm{GradU} = \nabla_\U f(\U,\B) = \sum_\ik (\y_\ik - \a_\ik{}^\top \xhat_k) ) \a_\ik \bhat_k{}^\top$ is the same as earlier. Thus,%
\begin{align*}
\mathrm{Err} 
& =  \sum_\ik (\hat\y_\ik - \y_\ik)  \a_\ik \bhat_k{}^\top \\
& = \sum_\ik (\hat\cb_\ik - \cb_\ik ) |\a_\ik^\top \xstar_k| \a_\ik \bhat_k{}^\top \\
& = \sum_\ik (\hat\cb_\ik \bar{\cb}_\ik - 1 ) (\a_\ik^\top \xstar_k) \a_\ik \bhat_k{}^\top
\end{align*}

Proceeding as in the proof of Lemma \ref{algebra}, and using  $\| (\I - \Ustar \Ustar^\top) \mathrm{Err}\|_F \le \|\mathrm{Err}\|_F$ and $\|\mathrm{Err}\| \le \|\mathrm{Err}\|_F$,  we can conclude the following
{\footnotesize
	\begin{align*}
& \SE(\Ustar, \U^+) \le  \\
&  \frac{ \|(I - (\eta/m) \mathrm{Hess}\| \cdot \SE(\Ustar,\U) + (\eta/m) \|\mathrm{Term2}\|_F + (\eta/m)\|\mathrm{Err}\|_F}{1 - (\eta/m) \|\mathrm{GradU}\| - (\eta/m) \|\mathrm{Err}\|_F}
\end{align*}
}
where the expressions for $\mathrm{GradU}, \mathrm{Term2}, \mathrm{Hess}$ are the same as before with one change: $\bhat_k$ is now obtained by solving a noisy $r$-dimensional PR problem (instead of a LS problem) using RWF \cite{rwf}. Thus, to complete the proof, (i) we need to bound
\[
\|\mathrm{Err}\|_F = \max_{\W \in \S_{nr} } \sum_\ik (\hat\cb_\ik \bar{\cb}_\ik - 1 ) (\a_\ik^\top \xstar_k)  (\a_\ik^\top \W  \bhat_k)
\]
and (ii) we need  bounds on the three other terms that were also bounded earlier for the linear case.

The  term $\|\mathrm{Err}\|_F$, is bounded in  Lemma 4 of \cite{lrpr_best} . We repeat the lemma below.
\newcommand{\deltatfrob}{\delta_t}
\begin{lemma}
Assume that $\SEF(\U_t, \Ustar) \le \deltatfrob$ with $\deltatfrob  < c/\kappa^2$.
Then,
 w.p. at least  $1-  2 \exp\left( nr \log(17)   - c \frac{m q\epsilon_2^2}{ \mu^2 \kappa  r} \right) -   \exp(\log q + r - c m )$,
\[
\|\mathrm{Err}\|_F  \le   C m ( \epsilon_2  +   \sqrt{\deltatfrob} )  \deltatfrob \sigmax^2
\]
\label{Err_bnd}
\end{lemma}
%
Consider the other three terms: $\mathrm{GradU}, \mathrm{Term2}, \mathrm{Hess}$. These  were bounded in Lemma \ref{terms_bnds} for the linear case. The statement and proof of this lemma remain the same as earlier because of the following reason. Its proof uses the bounds on $\bhat_k$, $\xhat_k$ from Lemma \ref{B_lemma}. The statement of this lemma also remains the same with one change: we replace $\|\xstar - \x\|$ by  $\dist(\xstar, \x)$ and $\|\X^* - \X\|_F^2$ by $\sum_{k=1}^q \dist^2(\xstar_k,\x_k)$, and the same for $\tb_k, \g_k$.  The first part of Lemma \ref{B_lemma} now follows by the first part of \cite[Lemma 3.3]{lrpr_best}.
All the subparts of the second part of Lemma \ref{B_lemma} follow exactly as given in its proof in Sec. \ref{B_lemma_proof}.

\section{Limitations of our results} \label{limitations}

Our results have three limitations: (i) the algorithm that is analyzed needs sample-splitting, even though, in numerical experiments this is not needed; (ii) our bound holds w.h.p. for a single matrix $\Xstar$ satisfying Assumption \ref{right_incoh} (and not for all such matrices); and (iii) for obtaining exactly zero error, we need an infinite number of samples. We explain here the reasons why we are unable to address these issues.
We should mention here that, since all computers are finite precision, (iii) is entirely a theoretical curiosity.  Also, many other results in the LR recovery literature, e.g., \cite{lowrank_altmin,fastmc,rmc_gd}, also have all these limitations.

\subsection{Need for sample-splitting} \label{samplesplit}
In Algorithm \ref{gdmin}, sample-splitting (line 3) helps ensure that the measurement matrices in each iteration for updating each of $\U$ and $\B$ are independent of all previous iterates: we split our sample set into $2T+1$ subsets, we use one subset for initialization of $\U$ and one subset each for $T$ iterations of updating $\B$ and updating $\U$.
This helps prove the desired error decay bound by applying the sub-exponential Bernstein inequality \cite{versh_book} which requires the summands to be mutually independent. This becomes true in our case because, conditioned on past measurement matrices, the current set of $\a_\ik$'s are independent of the last updated values of $\U,\B$; and the $\a_\ik$s for different $(i,k)$  are mutually independent by definition. Thus, under the conditioning, the summands are mutually independent. 
Since we prove convergence in order $\log(1/\epsilon)$ iterations, this only adds a multiplicative factor of $\log(1/\epsilon)$ in the sample complexity.
Sample-splitting and the above overall idea is a standard approach used in many older works; in fact it is assumed for most of the LRMC guarantees for solutions that do not solve a convex relaxation (are iterative algorithms) \cite{lowrank_altmin, fastmc, rmc_gd}. An exception is \cite{rpca_gd}. 

There are a few commonly used approaches to avoid sample splitting. (1) One is using the leave-one-out strategy as done in \cite{pr_mc_reuse_meas}. But this means that the sample complexity dependence on $r$ worsens: the LRMC sample complexity with this approach is $(n+q)r^3$ times log factors. Also, it is not clear how to develop this approach for alternating $\U, \B$ updates.
(2) The second is to try to prove error decay for all matrices that are close enough to the true $\Xstar$ and that satisfy the other assumptions of the guarantee.  There are at least two different approaches to doing this. (2a) The first, which was used in \cite{rpca_gd}, works for LRMC since its measurements are bounded and symmetric: the authors are able to utilize i.i.d. Bernoulli sampling and left and right singular vectors' incoherence to prove key probabilistic bounds for all matrices of the form $\U \V$ with $\U,\V$ both being incoherent. This does not work in our case because our measurements are asymmetric and unbounded (which means for example that $\y_\ik$ times its estimate is heavier-tailed than $\y_\ik$).

(2b) An alternative approach is the following overall idea, which has been successfully used for analyzing standard PR algorithms, e.g., see \cite{twf,rwf}, but does not always work for other problems. In our setting, this means the following: At iteration $t+1$, suppose that the previous estimate  $\U_t$ satisfies $\SE(\U_t,\Ustar) \le \delta_t$. 
We need to try to show that, for all $\U$ that are a subspace distance $\delta_t$ away from the true subspace, the next iterate (which is a function of $\U$ and of the current $\A_k,\y_k$ for all $k$) is a distance $c \delta_t$ away with a $c<1$. To be precise, for all $\U \in \T:=\{\U: \U^\top \U = \I \text{ and }  \SE(\U,\Ustar) \le \delta_t\}$, we need $\U^+(\U) = orth(\U - \eta \nabla_U f(\U,\B))$ to satisfy $\SE(\U^+, \Ustar) \le c \delta_t$ for a $c<1$. Here  $orth(\M)$ is a matrix with orthonormal columns spanning the same subspace as those of $\M$. Also recall that the columns of $\B$ are $\b_k:= (\A_k \U)^\dag \y_k$ for all $k \in [q]$.
One can show this for all $\U \in \T$ by covering $\T$ by a net containing a finite number of points that are such that any point in $\T$ is with a subspace distance $0.25\delta_t$ of some point in the net, and first proving that this bound holds for all $\U$ in the net. The first step for proving such a bound is to bound the error in the estimates $\b_k$ for all $\U$ in this net. Because of the decoupled column-wise recovery of the $\b_k$'s,  for {\em one} $\U$ in this net, the bound on $\|\b_k(\U) - \U^\top \xstar_k\|$ holds w.p. $\ge 1 - q\exp(r - c m)$. This is proved in Lemma \ref{B_lemma}. If we want this bound to hold for all $\U$'s in the  net covering $\T$, we will need a union bound over all points in the net. The smallest sized net to cover $\T$ with accuracy $\eps_{net}=0.25 \delta_t$ has size upper bounded by $C^{nr}$ \cite{versh_book}.
With using this, the probability lower bound becomes $1- \exp(nr + \log q  + r - c m)$. For this to even just be non-negative, we need  $m > C nr$ which is too large and makes our guarantee useless.

\subsection{Why we cannot prove our result for all $X^*$}
The inability to obtain a useful union bound over a net of size $C^{nr}$ explained above is also why we cannot do this. 

\subsection{Why sample complexity depends on the desired final accuracy $\eps$}
Observe from our result that the number of samples required to achieve a certain accuracy $\eps$ grows as $\log(1/\eps)$. This means that, for the algorithm to achieve zero error, we need an infinite number of samples. We should mention that this problem is not unique to our result. It is often seen for results that use sample-splitting, e.g., \cite{lowrank_altmin, rmc_gd}. An exception is \cite{fastmc} for LRMC, where the following basic idea is used: one tries to show that after enough iterations, e.g., when the recovery error is $\eps_0 = 1/n$ or smaller, one can start reusing the same samples and still prove error decay. This is also the idea used in \cite{pr_mc_reuse_meas}. 
Briefly, the reason we are unable to circumvent this problem using a similar idea to that of \cite{fastmc} is that our algorithm is not a regular GD or projected GD method.

To use a similar idea in our setting, we would need to proceed as follows. We use independent samples until the error is below an $\eps_0$ that is small enough. Pick $\eps_0 = 1/(\kappa^2n^2)$.
This happens after $T(\eps_0) = C\kappa^2 \log(n)\log(\kappa)$ iterations. Consider $t=T+1$. At this time, $\delta_{t} = \eps_0 = 1/(\kappa^2n^2)$. Thus, by Lemma \ref{B_lemma}, $\|\bhat_k - \U^\top \xstar_k\| \lesssim   (1/(\kappa^2n^2) ) \|\xstar_k\|$ and all the other bounds also hold with $\delta_t$ replaced by $\eps_0$.
We try to show error decay by applying Lemma \ref{algebra}. For this to work, we need to be able to show all of the following without using independence between $\U,\B$ and the $\A_k$s: (i) upper and lower bound the eigenvalues of $\mathrm{Hess} = \sum_\ik (\a_\ik \kron \b_k) (.)^\top$ as those proved earlier, (ii) bound $\|\nabla_\U f(\U,\B)\|/m$ by $c_0 \sigmin^2$ for a small constant $c_0< 1$ (in fact even in our main proof, such a bound is sufficient since this term only appears in the denominator), and (iii) bound $\|\mathrm{Term2}\|_F/m$ by $(c_2/\kappa^2) \delta_t \sigmax^2$ with a $c_2$ sufficiently less than one.

As we explain next, (i) and (ii) can be obtained easily, but (iii) cannot. We can obtain (i) by showing that $\mathrm{Hess}$ is close to $\mathrm{Hess}^* = \sum_\ik (\a_\ik \kron (\U^\top \Ustar \tb_k)) (.)^\top$; and $\mathrm{Hess}^*$ can be bounded almost exactly as done in our proof earlier since $\A_k$s are independent of $\xstar_k$s. The $\U$ in the expression for $\mathrm{Hess}^*$ does not matter because $\U^\top \Ustar$ is an $r \times r$ rotation matrix and one can take a maximum over all rotation matrices.  Using the loose bounds $\|\a_\ik\|\le 5 \sqrt{n}$ w.h.p., one can show that $\|\mathrm{Hess}^* - \mathrm{Hess}\| \le mq \max_\ik  [  \max_{\W \in \S_{nr}} |\a_\ik^\top \W \g_k| \cdot \max_{\W \in \S_{nr}}  |\a_\ik^\top \W (\g_k - \b_k)| ] \lesssim mq \sqrt{n} \mu \sqrt{r/q} \sigmax \cdot \sqrt{n} \eps_0    \mu \sqrt{r/q} \sigmax \le m \mu^2 (r/n) \sigmin^2 $.  Similarly, for (ii), $\sum_\ik \|\a_\ik \a_\ik^\top (\xstar_k - \x_k) \b_k^\top\| \lesssim mq \cdot  \sqrt{n} \cdot \sqrt{n} \cdot \eps_0   \cdot (\mu^2 r/q)\sigmax^2 = m (\mu^2 r/n) \sigmin^2$. Using $(\mu^2 r/n) \ll 1$, claims (i) and (ii) follow.
However, proving (iii) seems to be impossible without using the fact that $\E[\mathrm{Term2}] = 0$. But this expected value is zero only when $\A_k$s are independent of $\U,\B$. 

\bfpara{Possible ways to prove (iii)}
For bounding $\mathrm{Term2}$ for times $t > T(\eps_0)$, we can try one of the following ideas. (1) Try to use Cauchy-Schwarz in a way that the projection orthogonal to $\Ustar$ is used. There does not seem to be a way to make this work. (2) Try to use the leave-one-out strategy of \cite{pr_mc_reuse_meas} only for $t > T(\eps_0)$.



\begin{figure}[t]
%
\begin{center}
%
{
\begin{subfigure}{0.32\textwidth}
		{\includegraphics[width = 0.99\textwidth]{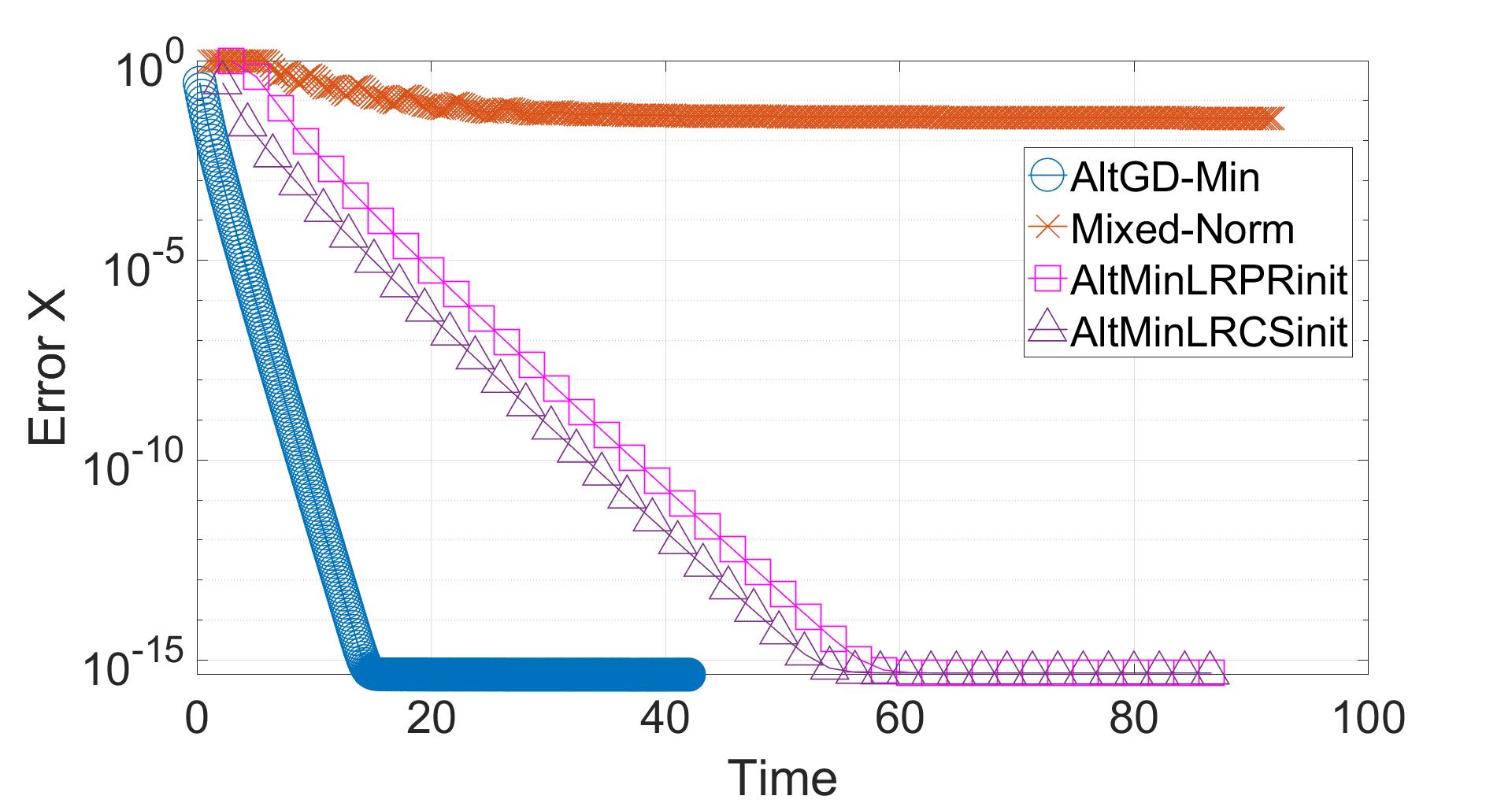}}
\caption{\small{$m=80$, $n=q=600,r=4$}}
\label{m80_gdmin}
\end{subfigure}
\begin{subfigure}{0.32\textwidth}
		{\includegraphics[width = 0.99\textwidth]{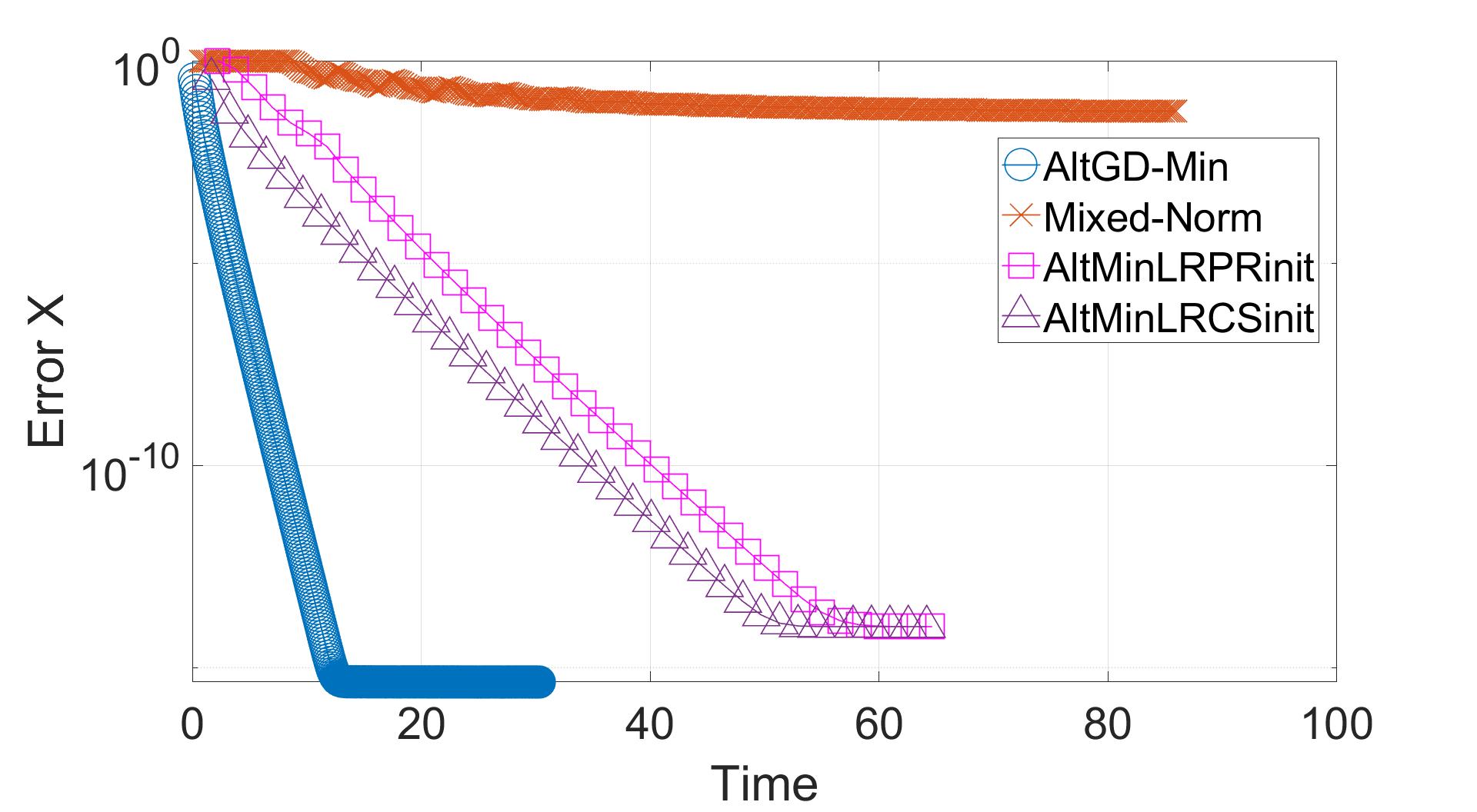}}
\caption{\small{$m=50$, $n=q=600,r=4$}}
\label{m50_gdmin}
\end{subfigure}
\begin{subfigure}{0.32\textwidth}
		{\includegraphics[width = 0.99\textwidth]{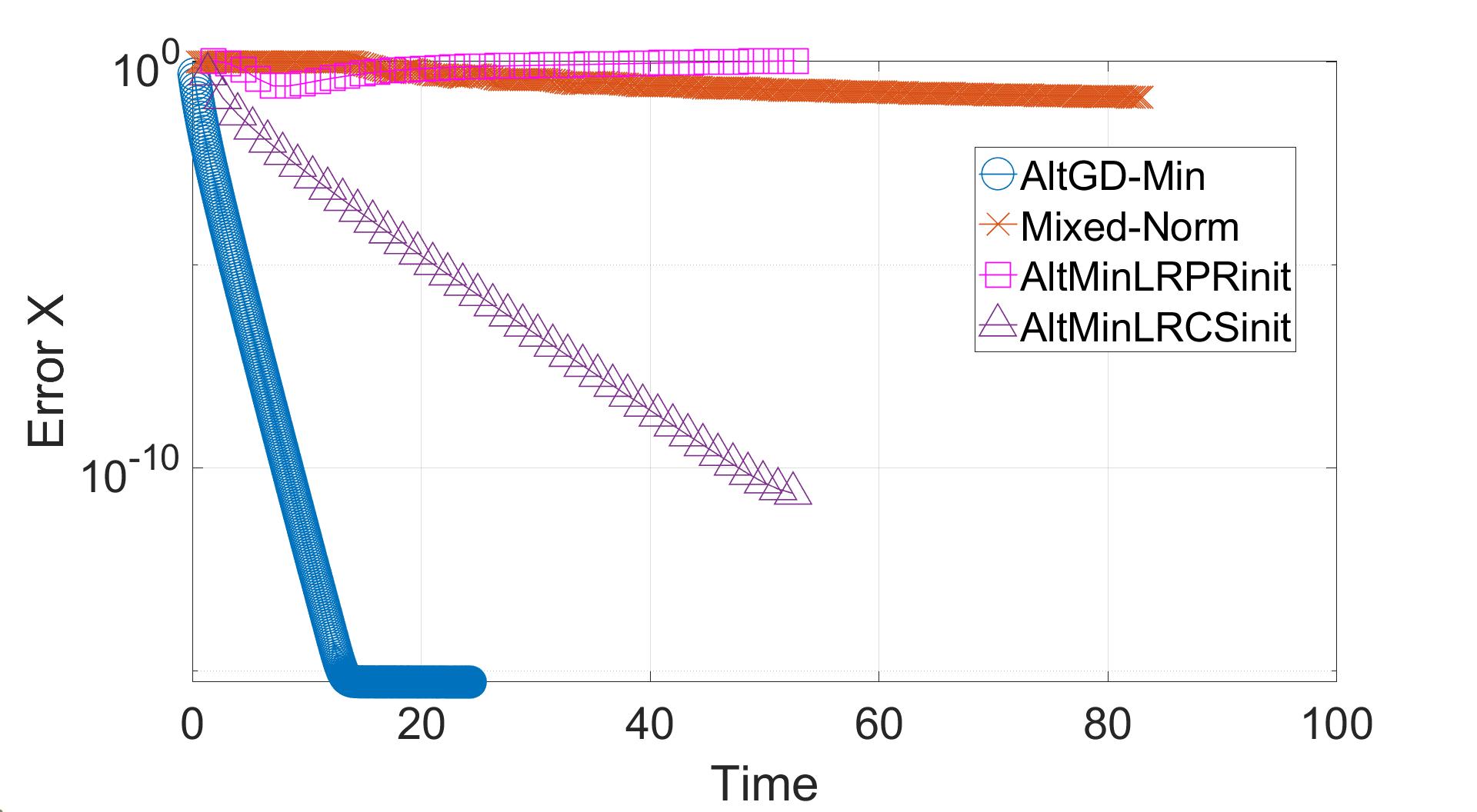}}
\caption{\small{$m=30$, $n=q=600,r=4$}}
\label{m30_gdmin}
\end{subfigure}%
}
{
\begin{subfigure}{0.32\textwidth}
		{\includegraphics[width = 0.99\textwidth]{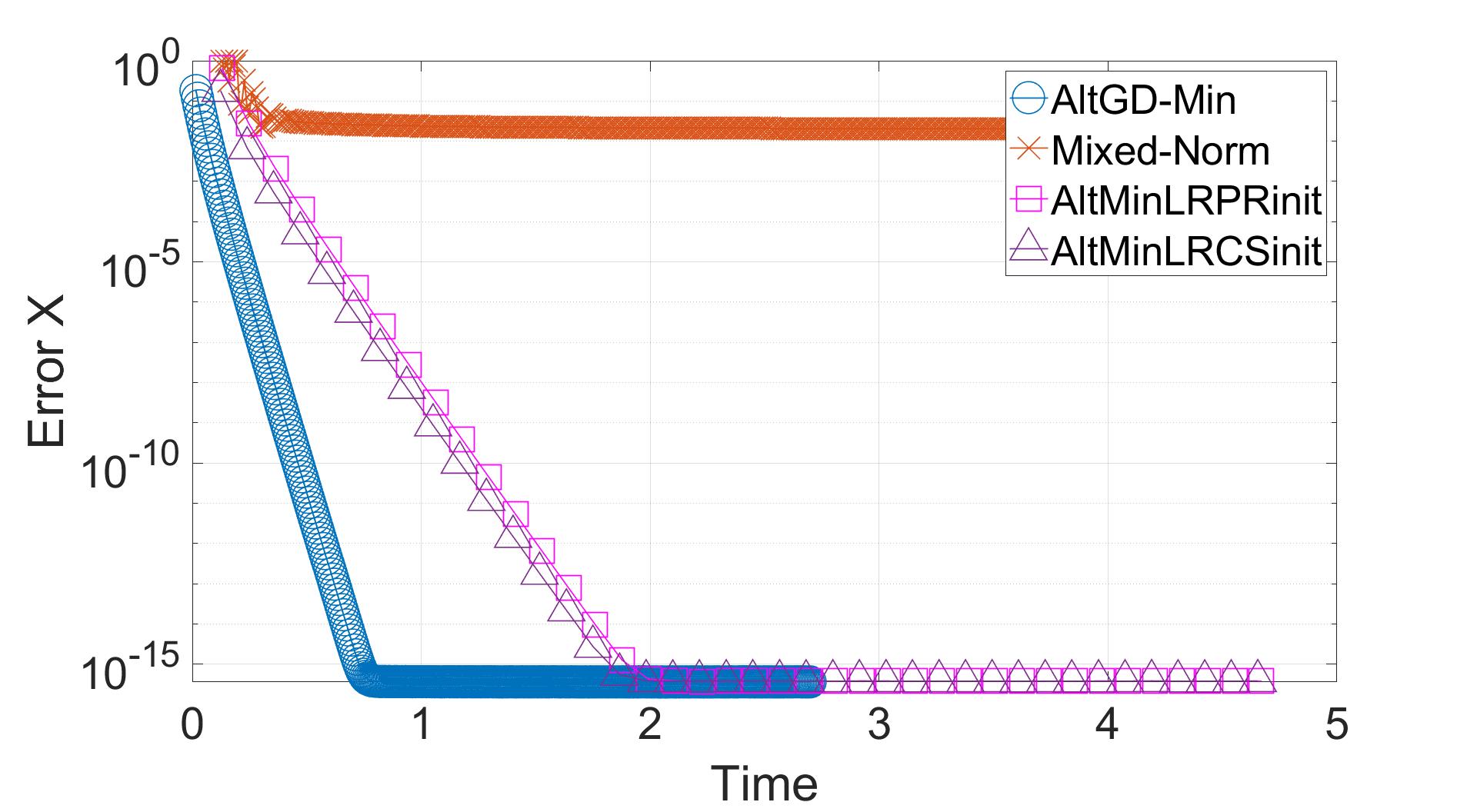}}
\caption{\small{$m=90$, $n=100,q=120,r=2$}}
\label{m90_n100_gdmin}
\end{subfigure}%
}
\end{center}
\caption{\sl\small{Comparing the proposed algorithm with existing approaches for solving LRcCS. 
}}
\label{gdmin_compare}
\vspace{-0.1in}
\end{figure}

\begin{figure}[t!]
\begin{center}
	{\includegraphics[width = 0.32\textwidth]{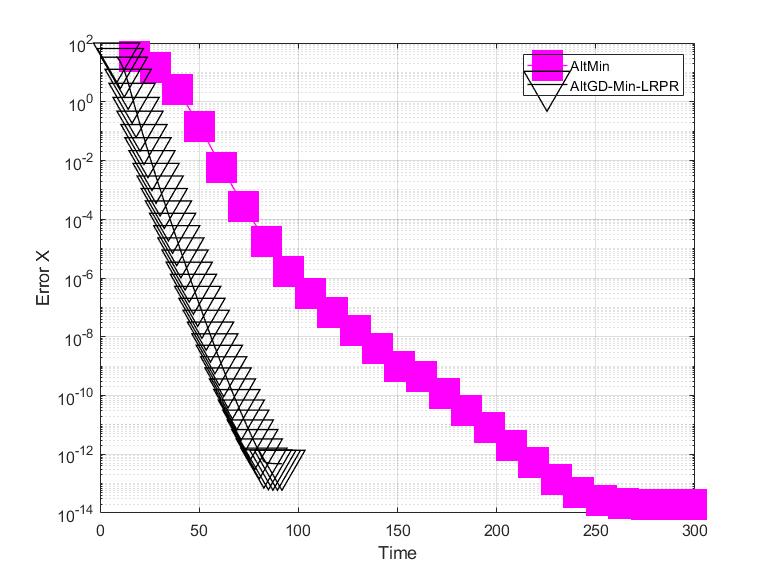}}
\end{center}
\caption{\sl\small{Comparing the proposed algorithm with existing approach for solving LRPR. We used $n=600, q=1000, r=4$ and $m=250$.}}
\label{m250_lrpr}
\vspace{-0.1in}
\end{figure}

\section{Numerical Experiments}\label{sims}

Our first experiment compares AltGD-Min with the mixed norm minimization solution from  \cite{lee2019neurips} (mixed-norm-min) and with
the AltMin algorithm \cite{lrpr_icml,lrpr_it,lrpr_best} modified for the linear LRcCS problem (replace the PR step for updating $\b_k$'s by a simple LS step). We implement this with using two possible initializations: the initialization developed in \cite{lrpr_icml,lrpr_it,lrpr_best} for LRPR (AltMinLin-LRPRinit), and with the initialization approach developed in this work (AltMinLin-LRCSinit). 
For mixed norm min, we used the code downloaded from \url{https://www.dropbox.com/sh/lywtzc0y9awpvgz/AABbjuiuLWPy_8y7C3GQKo8pa?dl=0}, which is provided by the authors. For AltMin, we used the code from \url{https://github.com/praneethmurthy/}. We implemented AltGD-Min with $\eta = 0.4/\|\Xhat_0\|^2$ and $\tC = 9$. Also, we used {\em one} set of measurements for all its iterations. 


For chosen values of $n,q,r$ and $m$, we simulated the data as follows. We simulated $\Ustar$ by orthogonalizing an $n \times r$ standard Gaussian matrix; and $\tb_k$s were generated i.i.d. from  $\n(0,\I_r)$. These were generated once. For each of 100 Monte Carlo runs, the measurement matrices $\A_k$ contained i.i.d. standard Gaussian entries. We obtained $\y_k = \A_k \Ustar \tb_k$, $k \in [q]$. For the LRPR experiment, we used $\ym_k = |\y_k|$ as the measurements.
We plot the empirical average of $\|\Xhat - \Xstar\|_F/\|\Xstar\|_F$ at each iteration $t$ on the y-axis (labeled ``Error-X'' in the plots) and the time taken by the algorithm until iteration $t$ on the x-axis. 


For our first experiment, shown in Fig. \ref{m80_gdmin}, we used $n=600,q=600,r=4$ and $m=80$. In this case, mixed-norm-min error decays to about 2-5\% but does not reduce any further. But, for our algorithm, AltGD-Min, and for both versions of AltMin, the error decays to $10^{-15}$. Notice also that AltGD-Min is much faster than all the other approaches.
Fig. \ref{m50_gdmin} reduced $m$ to $m=50$. Here a similar trend is observed, except that the error decays to only around $10^{-13}$ for AltGD-Min and $10^{-11}$ for the two AltMin approaches. Finally, for Fig. \ref{m30_gdmin}, we reduced $m$ to $m=30$. In this case, only AltGDmin and AltMin-LRCSinit work, while mixed-norm-min and AltMin-LRPRinit errors do not decrease at all. The reason is both these need a higher sample complexity (see Table \ref{compare_lrccs}).
%
\cred
Finally, we also tried an experiment with very large $m$: $n=100,q=120,r=2$ and $m=0.9n = 90$, see Fig. \ref{m90_n100_gdmin}. Even for such a large value of $m$ (compared to $n$), observe that the mixed-norm-min error saturates at around 1-2\%. The likely reason for this that, in the guarantee for mixed-norm-min \cite{lee2019neurips} (summarized for the noiseless case in Proposition \ref{convex_sol} given earlier), even for $m = n$, the error is bounded by a multiplier (more than 1) times $\sqrt{r/q}$.%
\cbl

For the comparisons for the LRPR problem shown in Fig. \ref{m250_lrpr}, we need  a much larger $q$ and $m$ since LRPR requires $mq$ to scale as $nr^3$ both for initialization and for the GDmin iterations and the multiplying constants are also much larger for LRPR. We used $n=600,q=1000,r=4$ and $m=250$. 
Notice that altGD-Min-LRPR is faster than AltMin-LRPR.
We implemented altGD-Min-LRPR with $\eta = 0.9/\|\Xhat_0\|^2$, $\tC = 9$, and $T_{RWF,t}= \max(5+t,40)$ in the RWF code (code for \cite{rwf}, downloaded from the specified site). Also, here again, we used {\em one} set of measurements for all its iterations.

\section{Conclusions}\label{conclude}
This work developed a sample-efficient and fast gradient descent (GD) solution, called AltGD-Min, for provably recovering a low-rank (LR) matrix from mutually independent column-wise linear projections. This problem, which we refer to as ``Low Rank column-wise Compressive Sensing (LRcCS)'', frequently occurs in LR-based accelerated low rank dynamic MRI and in federated sketching. If used in a federated setting, AltGD-Min is also communication-efficient. The LRcCS problem has not received little attention in the theoretical literature unlike the other well-studied LR recovery problems (matrix completion, sensing, or multivariate regression).

\appendices \renewcommand\thetheorem{\Alph{section}.\arabic{theorem}}

\section{Understanding why LRMC-style GD approaches cannot be easily analyzed for LRcCS} \label{algo_understand} 

\begin{table*}[t]
\caption{\small{
Understanding why LRMC style projected-GD on $\X$ does not work in our case.
	}}
\label{LRMC_diff}
\vspace{-0.1in}	
\begin{center}
\renewcommand*{\arraystretch}{1.01}
{
\begin{tabular}{l l l} \toprule
  & LRMC &     Our Problem, LRcCS   \\  \midrule
$\tilde{f}(\Xhat)$  & $\ds \sum_{k=1}^q \sum_{j=1}^n (\y_{jk} - \delta_{jk} \Xhat_{jk})^2 $ & $\ds  \sum_{k=1}^q \sum_{i=1}^m (\y_\ik - \a_\ik^\top \xhat_k)^2 $ \\
& $\delta_{jk} \iidsim Bernoulli(p)$ & $\a_\ik \iidsim \n(0, \I_n)$ \\
\hline
$\nabla_X \tilde{f}(\Xhat)$  & $\ds  \sum_{k=1}^q \sum_{j=1}^n \delta_{jk}(\y_{jk} - \delta_{jk} \Xhat_{jk}) \e_j \e_k^\top$ & $\ds  \sum_{k=1}^q \sum_{i=1}^m (\y_\ik - \a_\ik^\top \xhat_k) \a_\ik \e_k^\top  $ \\
   & $\ds = \sum_{k=1}^q \sum_{j=1}^n \delta_{jk}( \Xstar_{jk} - \Xhat_{jk}) \e_j \e_k^\top$ & $\ds = \sum_{k=1}^q \sum_{i=1}^m \a_\ik^\top (\xstar_k - \xhat_k) \a_\ik \e_k^\top  $ \\
%
$\tilde\H:= \H - \eta \nabla f(\Xhat)$   & $\ds \sum_{k=1}^q \sum_{j=1}^n (1 - \frac{\delta_{jk}}{p}) \H_{jk} \e_j \e_k{}^\top$  & $\ds \frac{1}{m}  \sum_{k=1}^q \sum_{i=1}^m  (\I -    \a_\ik \a_\ik{}^\top ) \h_k   \e_k{}^\top$    \\
\bottomrule
			\end{tabular}
		}	
	\end{center}
\vspace{-0.15in}	
\end{table*}

\subsection{Gradient Descent} \label{lrmc_compare_iters}
%


\cred
The iterates of a gradient descent (GD) algorithm converge when the gradient approaches zero.
Thus, in order to show its convergence, one needs to be able to bound the norm of the gradient and show that it goes to zero with iterations. In order to show fast enough convergence (reach $\eps$ error in order $\log(1/\eps)$ iterations), one further needs to show that this bound on the gradient norm decreases sufficiently with each iteration.
\cbl
Consider projGD-X which was studied in \cite{rmc_gd} for solving LRMC. ProjGD-X iterations involve computing $\Xhat^+ \leftarrow \proj_r( \Xhat - \nabla_\X \tilde{f}(\Xhat) )$, here $\proj_r(\M)$ projects its argument onto the space of rank-$r$ matrices. 
To bound $\|\nabla_\X \tilde{f}(\X)\|$, we need to bound $|\w^\top \nabla_\X \tilde{f}(\X) \z|$ for any unit norm vectors $\w,\z$.
We show the cost function $\tilde{f}(\Xhat)$ and its gradient for both LRMC and LRcCS in Table \ref{LRMC_diff}. Observe that, for LRcCS, $\w^\top \nabla_\X \tilde{f}(\X) \z$ is a sum of sub-exponential r.v.s with sub-exponential norms bounded by $K_e = \max_k \|\w\| \cdot \|\xstar_k - \xhat_k\| \cdot |\z_k|  \le  \max_k \|\xstar_k - \xhat_k\|$. Thus, in order to get a small enough bound on $|\w^\top \nabla_\X \tilde{f}(\X) \z|$ by applying the sub-exponential Bernstein inequality \cite{versh_book}, we need a small enough bound on $\max_k \|\xstar_k - \xhat_k\|$ (column-wise error bound). 
%
It is not clear how to get this because the projection step introduces coupling between the different columns of the estimated matrix $\X$ \footnote{
Let $\H := \X - \Xstar$, $\tilde\H:= (\Xhat - \eta \nabla f(\Xhat)) - \Xstar = \H - \eta \nabla f(\Xhat)$, and $\H^+ = \Xhat^+ - \Xstar = \proj_r(\Xhat - \nabla f(\Xhat)) - \Xstar = \proj_r(\Xstar + \tilde\H) - \Xstar $.
To bound the LRMC projGD-X errors, one needs an entry-wise bound of the form $\|\H^+\|_{\max} \le \delta_t \|\Xstar\|_{\max}$ with $\delta_t$ decaying exponentially. 
We show the expressions for $\tilde\H$ in the table. For LRMC, notice that different summands of $\tilde\H$ are mutually independent and each depends on only one entry of $\H$. This fact is carefully exploited in \cite[Lemma 1]{rmc_gd} and \cite[Lemma 1]{fastmc}.  By borrowing ideas from the literature on spectral statistics of Erdos-Renyi graphs \cite{ekyy}, the authors are able to obtain expressions for higher powers of $(\tilde\H \tilde\H^\top)$. These expressions help them get the desired bound under the desired sample complexity. 
For LRcCS, using the gradient expression, we need a  bound on $\max_k \|\h^+_k\|$ in terms of $\|\h_k\|$ in order to show its exponential decay. Since the different entries of $\tilde\H$ are not mutually independent and not bounded, the LRMC proof approach cannot be borrowed.
}. 
Moreover, even if we could somehow get such a bound, in the best case, it would be proportional to $\delta_t \max_k \|\xstar_k\|$ with $\delta_t<1$ and decaying exponentially with $t$. Using Assumption \ref{right_incoh},
this would then imply that  $K_e \le \delta_t \max_k \|\xstar_k\| \le \delta_t \mu  \sqrt{r/q} \sigmax$. But, this is not small enough. We need it to be proportional to $\delta_t (r/q)$ in order to be able to bound the gradient norm under the desired sample complexity.

Consider altGDnormbal studied in \cite{lafferty_lrmc,rpca_gd} for LRMC. In this case again, the desired column-wise error bound cannot be obtained because the update step for $\B$ involves GD w.r.t. $f(\U,\B) + f_2(\U,\B)$. The gradient w.r.t $f_2$ (norm-balancing term) introduces coupling  between the different columns of $\B$, and hence, also between columns of $\X = \U \B$.
Thus, once again, it is not clear how to get a tight bound on $\max_k \|\xstar_k - \xhat_k\|$. 

For AltGD-Min, because the min step for updating $\B$ is a decoupled LS problem, it is possible to get the desired column-wise error bound. Secondly, because we use GD w.r.t $\U$, there is an extra $\b_k^\top$ term in the gradient summands. This makes the gradient (and its deviation from its expected value), a sum of {\em nice-enough} sub-exponential r.v.s as explained in Sec. \ref{outline_iters}.%

\begin{table*}[t]
\caption{\small{
Why the LRMC initialization approach cannot be directly borrowed?
}}
\label{LRMC_init_diff}
\vspace{-0.1in}	
\begin{center}
\renewcommand*{\arraystretch}{1.01}
{
\begin{tabular}{lll} \toprule
  & LRMC &     Our Problem, LRcCS   \\  \midrule
$\Xhat_{0,full}=$ & $\ds \sum_k \sum_j \frac{\delta_{jk}}{p} \y_{jk} \e_j \e_k{}^\top$    & $\ds \frac{1}{m}  \sum_k  \sum_i   \a_\ik \y_\ik \e_k{}^\top$  \\
&  $\delta_{jk} \iidsim Bernoulli(p)$ & $\a_\ik \iidsim \n(0, \I_n)$  \\
$\H_0 =\Xhat_{0,full} - \Xstar$   & $\ds \sum_{k=1}^q \sum_{j=1}^n (1 - \frac{\delta_{jk}}{p}) \Xstar_{jk} \e_j \e_k{}^\top$  & $\ds \frac{1}{m}  \sum_{k=1}^q \sum_{i=1}^m  (\I -    \a_\ik \a_\ik{}^\top ) \xstar_k   \e_k{}^\top$    \\
%
 Each summand is  & nicely bounded by  &   unbounded \& sub-expo. norm$^{**}$ is \\
 & $\mu^2 \sigmax (r/\sqrt{nq}) $  &  $\mu \sigmax \sqrt{r/q}$ (too large, need $r/q$)   \\
Concen. ineq. & Matrix Bernstein \cite{tail_bound}   & Sub-expo Bernstein \cite{versh_book}   \\
& gives desired sample comp.  & does not give desired sample comp. \\
\bottomrule
\end{tabular}
		}	
	\end{center}
{\footnotesize ${**}$: ``max sub-expo. norm": max sub-exponential norm of $(\a_\ik{}^\top\w) (\a_\ik{}^\top \xstar_k) (\e_k^\top \z)$ for any unit vectors  $\w,\z$.}
\vspace{-0.15in}	
\end{table*}

\subsection{Initialization} \label{lrmc_compare_init}

The standard approach used for initializing iterative algorithms for LRMC (as well as other linear LRR problems) is to compute the top $r$ left singular vectors of the matrix $\X_{0,full}$ that satisfies $(\X_{0,full})_{vec}= \mathcal{A}^\top (\y_{all})$, where $\y_{all}$ is the $mq$-length vector of all measurements and $\mathcal{A}$ denotes the linear mapping from $(\Xstar)_{vec}$ to $\y_{all}$. In case of LRMC and LRcCS, this is computed is as given in Table \ref{LRMC_init_diff}. It is not hard to see that, in both cases, $\E[\Xhat_{0,full}]= \Xstar$.
To show that this approach works, one typically uses a $\sin \Theta$ theorem, e.g., Davis-Kahan or Wedin, to bound $\SE(\Ustar, \U_0)$ as a function of terms that depend on  $\H_0:= \Xhat_{0,full} - \Xstar$. Thus a first requirement is to bound $\|\H_0\|$. For LRMC, this can be done easily since $\H_0$ is a sum of the independent one-sparse random matrices shown in the table with each matrix containing an i.i.d. Bernoulli r.v. times $\Xstar_{jk}$ ($jk$-th entry of $\Xstar$) as its nonzero entry. Using the left and right singular vectors' incoherence (assumed in all LRMC guarantees), and $\Xstar_{jk} = \e_j^\top \Xstar \e_k$, one can argue that, for unit vectors $\w,\z$, each summand of $|\w^\top \H_0 \z|$ is of order at most $(1/p) \sigmax r/\sqrt{nq}$. This bound, along with a bound on the ``variance parameter" needed for applying matrix Bernstein \cite{tail_bound},\cite[Chap 5]{versh_book} helps show that  $\|\H_0\| \le c \sigmax$ w.h.p., under the desired sample complexity bound. 
%
For LRcCS,  the summands of $\Xhat_{0,full}$, and hence of $\H_0$, are sub-exponential r.v.s. These can be bounded using the sub-exponential Bernstein inequality \cite[Chap 2]{versh_book}. This requires a bound on the maximum sub-exponential norm of any summand. Denote this bound by $K_e$. In order to show that $\|\H_0\| \le c \sigmax$ w.h.p, under the desired sample complexity, we need $K_e$ to be of order $(r/q)$ or smaller. However, for our summands, we can only guarantee $K_e \le (1/m) \max_k \|\xstar_k\| \le (1/m) \mu \sqrt{r/q} \sigmax$. This is not small enough, i.e., the summands are not {\em nice-enough} subexponentials. It will require $mq \gtrsim (n+q) r \cdot \sqrt{q}$ which is too large.%

\section{Proof of Initialization Theorem \ref{init_thm} without sample-splitting}
\label{init_reuse_proof}
Consider the initialization using $\Xhat_0$ defined in \eqref{newinit}. We we want to bound the initialization error without sample-splitting. This means that the threshold $\alpha$ is not independent of the $\a_\ik, \y_\ik$ used in the expression for $\Xhat_0$ and thus, it is not clear how to compute its expected value even if we condition on $\alpha$. However, the following slightly more complicated approach can be used.
Using Fact \ref{sumyik_bnd} and Assumption \ref{right_incoh}, it is possible to show that $\Xhat_0$ is close to a matrix, $\X_+(\epsilon_1)$ given next for which $\E[\X_+]$ is easily computed:
Let
\[
\alpha_+:= \tC (1+ \epsilon_1) \frac{\|\Xstar\|_F^2}{q}
\]
and define
\begin{align}
 \X_+ (\epsilon_1) & := \frac{1}{m} \sum_\ik  \a_\ik \y_\ik \e_k{}^\top \indic_{ \small \{ \y_\ik^2 \le \alpha_+ \} }. \text{ Then, } \nonumber \\      
 \E[\X_+] & = \Xstar \D(\epsilon_1), \nonumber\\  \D &:= diagonal(\beta_k(\epsilon_1)), \nonumber\\ \beta_k(\epsilon_1)&:= \E\left[\zeta^2 \indic_{\small \left\{ \zeta^2 \le \frac{\alpha_+}{\|\xstar_k\|^2} \right\} } \right]
\label{Xplus}
\end{align}
with $\zeta$  being a scalar standard Gaussian.
Thus $\Xhat_+$ is $\Xhat_0$ with the threshold $\alpha$ replaced by $\alpha_+$ which is deterministic. Consequently $ \E[\X_+]$ has a similar form too and is obtained as explained in the proof of Lemma \ref{Wedinlemma} given in Sec. \ref{Wedinlemma_proof}.

Next, recall that $\Xstar \svdeq \Ustar \bSigma \Bstar$ and   $\tC = 9 \kappa^2 \mu^2$. Let $\tc = c/\tC$ for a $c<1$. Clearly, the span of the top $r$ singular vectors of $\E[\X_+] = \Xstar \D$ equals $\Span(\Ustar)$ and it is rank $r$ matrix. Let,
\[
\E[\X_+]= \Xstar \D  \svdeq \Ustar \check\bSigma \Bcheck
\]
be its $r$-SVD (here $\Bcheck$ is an $r \times q$ matrix with its rows containing the $r$ right singular vectors). We thus have
\begin{align*}
\sigma_r(\E[\X_+]) &=  \sigma_{\min}(\check\bSigma) =\sigma_{\min}(\bSigma \Bstar \D \Bcheck{}^\top) \\&\ge \sigma_{\min}(\bSigma)\sigma_{\min}(\Bstar)\sigma_{\min}(\D)\sigma_{\min}(\Bcheck{}^\top) \\&= \sigmin \cdot 1 \cdot (\min_k \beta_k) \cdot 1
\end{align*}
Fact \ref{betak_bnd} given earlier shows that $(\min_k \beta_k) \ge 0.9$ and thus,
\[
\sigma_r(\E[\X_+]) \ge  0.9 \sigmin
\]
Also, $\sigma_{r+1}(\E[\X_+]) = 0$ since it is a rank $r$ matrix. Thus, using Wedin's $\sin \Theta$ theorem for $\SE$ (summarized in Theorem \ref{Wedin_sintheta}) applied with $\M \equiv \Xhat_0$, $\M^* \equiv \E[\X_+]$ gives
\begin{align}\label{Wedin}
&\SEF(\U_0,\Ustar)\nonumber \\&  \le  \dfrac{\sqrt{2} \max\left( \| (\Xhat_0 - \E[\X_+])^\top \Ustar \|_F , \| (\Xhat_0 - \E[\X_+])  \Bcheck{}^\top \|_F \right)}{0.9 \sigmin - \|\Xhat_0 - \E[\X_+]\|}
\end{align}

In the next three subsections, we prove a set of six lemmas that help bound the three terms in the expression above. {\em The main new ideas over the proof given earlier in Sec \ref{init_proof}, are in the proof of the first lemma, Lemma \ref{Xhat0_2} given below, and in the proof of Claim \ref{claim:expect_init} that is used in this proof.}

\begin{claim}
		\label{claim:expect_init}
\label{EXhat0_Xplus}
Let $\xstar \in \Re^n$, $\z \in \Re^n$ be two deterministic vectors and let $\alpha$ be a deterministic scalar. Let $\a \sim \n(0,\I_n)$ be a standard Gaussian vector and define $\y := \a^\top\xstar$.
For an $0 < \epsilon < 1$,
\[
\E\left[ |\y (\a{}^\top\z) | \indic_{\{ \y^2 \in [1\pm\epsilon ]\alpha  \}}   \right] \leq C \epsilon \|\z\| \sqrt{\alpha}.
\]
\end{claim}

Combining Lemmas \ref{Xhat0_1} and \ref{Xhat0_2} and using Fact \ref{sumyik_bnd}, and setting $\epsilon_1 = c\delta_0 / \sqrt{r} \kappa $, we conclude that, w.p. at least
 \\ $1-2\exp( (n+q)-\tc\epsilon_1^2 mq   ) -  \exp(-\tc mq \epsilon_1^2 ) \ge 1-2\exp( (n+q)- \tc mq \delta_0^2 / r \kappa^2 ) -  \exp(-\tc mq \delta_0^2 /r \kappa^2)$,
\[
\|\Xhat_0 - \E[\X_+]\| \lesssim \epsilon_1  \|\Xstar\|_F  \lesssim   c \delta_0 \sigmin
\]
By combining Lemmas  \ref{Xhat0_Bstar_2},  \ref{Xhat0_Bstar_1}, \ref{Xhat0_Ustar_2}, and \ref{Xhat0_Ustar_1} and using Fact \ref{sumyik_bnd},  and setting $\epsilon_1 = c \delta_0 / \sqrt{r} \kappa $, we conclude that, w.p. at least
\\ $1-2\exp( n r - \tc mq \delta_0^2 / r \kappa^2 ) - 2\exp( q r - \tc mq \delta_0^2 / r \kappa^2 ) -  \exp(- \tc mq \delta_0^2 /r \kappa^2 )$,
\[
\max\left( \| (\Xhat_0 - \E[\X_+])^\top \Ustar \|_F , \| (\Xhat_0 - \E[\X_+])  \Bcheck^\top \|_F \right) \lesssim c \delta_0 \sigmin
\]
Plugging these into \eqref{Wedin} proves Theorem \ref{init_thm}

\subsection{Bounding the denominator term}\label{denom_bnd}
By triangle inequality, $\|\Xhat_0 - \E[\X_+]\| \le \|\X_+ - \E[\X_+]\| + \|\Xhat_0 - \X_+\|.$ The next two lemmas bound these two terms. The lemmas assume the claim of Fact \ref{sumyik_bnd} holds, i.e., that $\frac{1}{mq}\sum_\ik \y_\ik^2 \in [1\pm \epsilon_1] \tC \|\Xstar\|_F^2/q$ where $\tC = 9\mu^2\kappa^2$.
	\begin{lemma}
		\label{lem:init_denom_term}
\label{Xhat0_2}
Assume that $\frac{1}{mq}\sum_\ik \y_\ik^2 \in [1\pm \epsilon_1] \tC \|\Xstar\|_F^2/q$ (claim of Fact \ref{sumyik_bnd} holds). Then, w.p. $1-\exp(C(n+q)-\epsilon_1^2mq/\mu^2\kappa^2)$,
		\[
		\|\Xhat_0 - \X_+\|  \leq C\epsilon_1 \mu\kappa\|\Xstar\|_F.
		\]
	\end{lemma}
	\begin{proof}[Proof of Lemma \ref{lem:init_denom_term}]
		We have
\begin{align*}
	\|\X_+ - \Xhat_0 \| &= \max_{\z\in \S^n,~\w\in \S^q} \z{}^\top\left(\X_+ - \Xhat_0 \right)\w \\&=   \max_{\z\in \S^n,~\w\in \S^q}  \frac{1}{m} \sum_\ik  \w(k)\y_\ik(\a_\ik{}^\top\z)\\&\qquad\times\indic_{\left\{  \frac{\tC}{mq}\sum_\ik \y_\ik^2\leq \y_\ik^2 \leq \frac{\tC(1+\epsilon_1)}{q}\|\Xstar\|_F^2 \right\} }.
\end{align*}
For the last expression above, we have used the assumption $\sum_\ik \y_\ik^2/m \le \tC(1+\epsilon_1)\|\Xstar\|_F^2$.
Consider the RHS for a fixed unit norm $\z$ and $\w$. The lower threshold of the indicator function is itself a r.v.. To convert it into a deterministic bound, we need the following sequence of bounding steps: To use our assumption that $\sum_\ik \y_\ik^2 / m \geq (1-\epsilon_1)\tC  \|\Xstar\|_F^2$, we first need to bound the summands by their absolute values. This is done as follows:
		\begin{align*}
		|\z{}^\top\left(\X_+ - \Xhat_0 \right)\w |
		&\leq \frac{1}{m} \sum_\ik  \big|\w(k)\y_\ik(\a_\ik{}^\top\z)\big|\\&\qquad\times\indic_{\left\{  \frac{\tC}{mq}\sum_\ik \y_\ik^2\leq |\y_\ik|^2 \leq \frac{\tC(1+\epsilon_1)}{q}\|\Xstar\|_F^2 \right\} },\\
		&\leq \frac{1}{m} \sum_\ik  \big|\w(k)\y_\ik(\a_\ik{}^\top\z)\big|\\&\qquad\times\indic_{\left\{   |\y_\ik|^2 \in [1\pm\epsilon_1]\frac{\tC}{q}\|\Xstar\|_F^2 \right\} },
		\end{align*}
		where in the last line we used our assumption that $\sum_\ik \y_\ik^2 / m \geq (1-\epsilon_1)\tC  \|\Xstar\|_F^2$.
This final expression is a sum of mutually independent sub-Gaussian r.v.s with subGaussian norm  $K_\ik \leq C |\w(k)| \sqrt{\tC(1+\epsilon_1)}\|\Xstar\|_F/\sqrt{q}  \le \sqrt{\tC}|\w(k)|  \|\Xstar\|_F/\sqrt{q}$. Thus, by applying the sub-Gaussian Hoeffding inequality, Theorem 2.6.2 of \cite{versh_book},
		\begin{align*}
		&\Pr\left\{ \Big| \sum_\ik  \big|\w(k)\y_\ik(\a_\ik{}^\top\z)\big|\indic_{\left\{   |\y_\ik|^2 \in [1\pm\epsilon_1]\frac{\tC}{q}\|\Xstar\|_F^2 \right\} } \right.\\&- \left.\E\left[\sum_\ik  \big|\w(k)\y_\ik(\a_\ik{}^\top\z)\big|\indic_{\left\{   |\y_\ik|^2 \in [1\pm\epsilon_1]\frac{\tC}{q}\|\Xstar\|_F^2 \right\} }\right]   \Big|\geq t  \right\}\\
		& \qquad\leq 2\exp\left[-c\frac{t^2}{\sum_\ik K_\ik^2}\right].
		\end{align*}
By setting $t = \epsilon_1 m \|\Xstar\|_F$,
		\[
		\frac{t^2}{\sum_\ik K_\ik^2} \geq \frac{m^2 q \epsilon_1^2 \|\Xstar\|_F^2}{\sum_\ik \tC \|\Xstar\|_F^2 |\w(k)|^2} = \frac{\epsilon_1^2mq}{\tC}.
		\]
Since $\tC = 9\mu^2\kappa^2$, thus, w.p. $1-\exp(-c\epsilon_1^2mq/\mu^2\kappa^2)$, for a fixed $\z$ and $\w$,
		\[
		\z{}^\top\left(\Xhat_0 - \X_+\right)\w \leq \epsilon_1 \|\Xstar\|_F + \E\left[\frac{1}{m}\sum_\ik  \big|\w(k)\y_\ik(\a_\ik{}^\top\z)\big|\indic_{\left\{   |\y_\ik|^2 \in [1\pm\epsilon_1]\frac{\tC}{q}\|\Xstar\|_F^2 \right\} }\right] .
		\]
By using Claim \ref{claim:expect_init} and $|\w(k)| \|\z\| = |\w(k)|$ we have
		\begin{align*}
		&\E\left[\frac{1}{m}\sum_\ik  \big|\y_\ik(\a_\ik{}^\top\z) \w(k) \big|\indic_{\left\{   |\y_\ik|^2 \in [1\pm\epsilon_1]\frac{\tC}{q}\|\Xstar\|_F^2 \right\} }\right] \\&\leq \sqrt{\tC(1+\epsilon_1)}\epsilon_1\|\Xstar\|_F \sum_{k} \big|\w(k)\big|/\sqrt{q} \leq C\epsilon_1\mu\kappa \|\Xstar\|_F,
		\end{align*}
		where in the last inequality we used Cauchy-Schwarz to show that $\sum_{k} \big|\w(k)\big|/\sqrt{q} \le \sqrt{ \sum_k \big|\w(k)\big|^2 \sum_k (1/q) } = 1$. Or this also follows by $\|\w\|_1/\sqrt{q} \le \|\w\| = 1$.  Also, we used $\sqrt{\tC} = C \kappa \mu$.

Thus, w.p. $1-\exp(-c\epsilon_1^2mq/\mu^2\kappa^2)$, for a fixed $\z$ and $\w$, $\z{}^\top\left(\Xhat_0 - \X_+\right)\w \leq C\epsilon_1\mu\kappa \|\Xstar\|_F$.

By Proposition \ref{epsnet_MW}, $\max_{\z\in \S^n,~\w\in \S^q}  \z{}^\top\left(\Xhat_0 - \X_+\right)\w  \le 1.4 C \epsilon_1 \mu \kappa \|\Xstar\|_F$ w.p. at least $1-\exp( (n+q)\log(17)-c\epsilon_1^2mq/\mu^2\kappa^2)$.
\end{proof}

	\begin{lemma}
		\label{lem:init_denom_term2}
\label{Xhat0_1}
		Consider $\X_+$. Fix $1 < \epsilon_1 < 1$. Then, w.p. $1-\exp\left[C(n+q)-c\epsilon_1^2mq/\mu^2\kappa^2\right]$
		\[
		\|\X_+ -\E[\X_+]\| \leq C\epsilon_1 \|\Xstar\|_F.
		\]
	\end{lemma}
	\begin{proof}[Proof of Lemma \ref{lem:init_denom_term2}]
		The proof involves an application of the sub-Gaussian Hoeffding inequality followed by an epsilon-net argument, both almost the same as those used in the proof of Lemma \ref{lem:init_denom_term} given above. We have,
		\[
		\|\X_+ -\E[\X_+]\| = \max_{\z\in\mathcal{S}_n, \w\in\mathcal{S}_q} \langle  \X_+ -\E[\X_+], ~\z\w{}^\top\rangle.
		\]
For a fixed $\z\in\mathcal{S}_n, \w\in\mathcal{S}_q$, we have
		\begin{align*}
		&\langle  \X_+ -\E[\X_+], ~\z\w{}^\top\rangle \\&= \frac{1}{m} \sum_\ik \left(\w(k)\y_\ik(\a_\ik{}^\top\z)\indic_{\left\{|\y_\ik|^2 \leq \frac{\tC(1+\epsilon_1)}{q}\|\Xstar\|_F^2 \right\} }  \right.\\&\qquad-\left. \E\left[\w(k)\y_\ik(\a_\ik{}^\top\z)\indic_{\left\{|\y_\ik|^2 \leq \frac{\tC(1+\epsilon_1)}{q}\|\Xstar\|_F^2 \right\} } \right]\right) .
		\end{align*}
The summands are mutually independent, zero mean sub-Gaussian r.v.s with norm $K_\ik\leq C |\w(k)| \sqrt{\tC(1+\epsilon_1)}\|\Xstar\|_F/\sqrt{q}$. We will again apply the sub-Gaussian Hoeffding inequality Theorem 2.6.2 of \cite{versh_book}.
Let $t=\epsilon_1m\|\Xstar\|_F$. Then
		\[
		\frac{t^2}{\sum_\ik K_\ik^2} \geq \frac{\epsilon_1^2m^2\|\Xstar\|_F^2}{\sum_\ik \tC(1+\epsilon_1)\|\Xstar\|_F^2/q} \geq \frac{\epsilon_1^2mq}{C\mu^2\kappa^2}
		\]
		Thus, for a fixed $\z\in\mathcal{S}_n, \w\in\mathcal{S}_q$, by sub-Gaussian Hoeffding, we conclude that, w.p. at least $1-\exp\left[-c\epsilon_1^2mq/\mu^2\kappa^2\right]$,
		\[
		\langle  \X_+ -\E[\X_+], ~\z\w{}^\top\rangle \leq C \epsilon_1 \|\Xstar\|_F.
		\]
By Proposition \ref{epsnet_Mwz}, the above bound holds w.p. at least $1-\exp\left[ (n+q) -c\epsilon_1^2mq/\mu^2\kappa^2\right]$.
	\end{proof}

\subsection{Bounding the $\Bcheck$ numerator term}
We bound $ \| (\Xhat_0 - \E[\X_+])  \Bcheck^\top \|_F $ in this section. By triangle inequality. it is bounded by $\|\left(\Xhat_0 - \X_+\right)\Bcheck{}^\top\|_F + \|\left(\Xhat_+ -\E[ \X_+]\right)\Bcheck{}^\top\|_F$.


	\begin{lemma}
		\label{lem:init_nom_B_term1}
\label{Xhat0_Bstar_2}
Assume that $\frac{1}{m}\sum_\ik \y_\ik^2\in[1\pm \epsilon_1]\|\Xstar\|_F^2$. Then, w.p. $1-\exp\left[ nr-c\epsilon_1^2mq/\mu^2\kappa^2\right]$,
		\[
		\|\left(\Xhat_0 - \X_+\right)\Bcheck{}^\top\|_F \leq C\epsilon_1\mu\kappa\|\Xstar\|_F.
		\]
	\end{lemma}
	\begin{proof}[Proof of Lemma \ref{lem:init_nom_B_term1}]
The initial part of the proof is very similar to the that of the proof of Lemma \ref{Xhat0_2}.
We have,
		$
		\|\left(\Xhat_0 - \X_+\right)\Bcheck{}^\top\|_F = \max_{\W\in \S_{nr} } \langle \W,~ \left(\Xhat - \X_+\right)\Bcheck{}^\top\rangle.
		$
		For a fixed $\W\in \S_{nr}$,
		\begin{align*}
		&\langle \W,~ \left(\Xhat_0 - \X_+\right)\Bcheck{}^\top\rangle \\&= \frac{1}{m} \sum_\ik  \y_\ik (\a_\ik{}^\top\W\bcheck_{k}) \indic_{\left\{  \frac{\tC}{mq}\sum_\ik \y_\ik^2\leq |\y_\ik|^2 \leq \frac{\tC(1+\epsilon_1)}{q}\|\Xstar\|_F^2 \right\} }
		\end{align*}
Proceeding as in the proof of Lemma \ref{Xhat0_2},
		\begin{align*}
		&\frac{1}{m} \sum_\ik  \y_\ik (\a_\ik{}^\top\W\bcheck_{k})\indic_{\left\{  \frac{\tC}{mq}\sum_\ik \y_\ik^2\leq |\y_\ik|^2 \leq \frac{\tC(1+\epsilon_1)}{q}\|\Xstar\|_F^2 \right\} }\\
		&\leq  \frac{1}{m} \sum_\ik  \big|\y_\ik (\a_\ik{}^\top\W\bcheck_{k})\big|\indic_{\left\{  \frac{\tC}{mq}\sum_\ik \y_\ik^2\leq |\y_\ik|^2 \leq \frac{\tC(1+\epsilon_1)}{q}\|\Xstar\|_F^2 \right\} } ,\\
		&\leq \frac{1}{m} \sum_\ik  |\y_\ik| |(\a_\ik{}^\top\W\bcheck_{k})| \indic_{\left\{   |\y_\ik|^2 \in  [1\pm \epsilon_1] \frac{\tC}{q}\|\Xstar\|_F^2 \right\} } .
		\end{align*}
The summands are mutually independent sub-Gaussian r.v.s with norm $K_\ik \leq C \sqrt{\tC(1+\epsilon_1)} \|\W\bcheck_{k}\| \|\Xstar\|_F/\sqrt{q}$.
		Thus, we can apply the sub-Gaussian Hoeffding inequality Theorem 2.6.2 of \cite{versh_book}.
Set $t=\epsilon_1m\|\Xstar\|_F$. Then we have
		\[
		\frac{t^2}{\sum_\ik K^2_\ik} \geq \frac{\epsilon_1^2m^2\|\Xstar\|_F^2}{ (\sum_\ik \|\W\bcheck_{k}\|^2 ) \tC(1+\epsilon_1)\|\Xstar\|_F^2/q} \geq \frac{\epsilon_1^2mq}{C\mu^2\kappa^2},
		\]
		where we used the fact that  $\Bcheck\Bcheck{}{}^\top=\I$ ($\Bcheck^\top$ contains right singular vectors of a matrix) and thus $\|\W\Bcheck\|_F = 1$. Applying sub-Gaussian Hoeffding, we can conclude that, w.p., $1-\exp\left[-c\epsilon_1^2mq/\mu^2\kappa^2\right]$
		\begin{align*}
		&\frac{1}{m} \sum_\ik \big|\y_\ik (\a_\ik{}^\top\W\bcheck_{k})\big|\indic_{\left\{   |\y_\ik|^2 \in  [1\pm \epsilon_1] \frac{\tC}{q}\|\Xstar\|_F^2 \right\} } \\&\leq \epsilon_1 \|\Xstar\|_F \\&\qquad+ \frac{1}{m}\sum_\ik  \E\left[\big|\y_\ik (\a_\ik{}^\top\W\bcheck_{k})\big|\indic_{\left\{   |\y_\ik|^2 \in  [1\pm \epsilon_1] \frac{\tC}{q}\|\Xstar\|_F^2 \right\} }\right] .
		\end{align*}
We use Claim \ref{claim:expect_init} to bound the expectation term. Using this lemma with $\alpha^2 \equiv \tC (1 +\epsilon_1) \|\Xstar\|_F^2/q$,  $\z \equiv \W \bcheck_k$
		\begin{align*}
		&\frac{1}{m}\sum_\ik  \E\left[\big|\y_\ik (\a_\ik{}^\top\W\bcheck_{k})\big|\indic_{\left\{   |\y_\ik|^2 \in  [1\pm \epsilon_1] \frac{\tC}{q}\|\Xstar\|_F^2 \right\} }\right] \\&\leq \frac{1}{m}\sum_\ik  \sqrt{\tC(1+\epsilon_1)}\epsilon_1 \|\Xstar\|_F \|\W\bcheck_{k}\|/\sqrt{q} \leq C\epsilon_1\mu\kappa\|\Xstar\|_F.
		\end{align*}
where the last inequality used Cauchy-Schwarz on $\sum_k \|\W \bcheck_k\|/\sqrt{q}$ to conclude that $\sum_k \|\W \bcheck_k\| (1/\sqrt{q}) \le \sqrt{ \sum_k \|\W \bcheck_k\|^2  \sum_k (1/q) } =\sqrt{ \|\W \Bcheck\|_F^2 \cdot 1}  = 1$ since $ \|\W \Bcheck\|_F=1$.

By Proposition \ref{epsnet_MW}, the above bound holds for all $\W \in \S_{nr}$, w.p. at least $1-\exp\left[nr\log(1+2/\epsilon_{net})-c\epsilon_1^2mq/\mu^2\kappa^2\right]$.
	\end{proof}
	
	\begin{lemma}
		\label{lem:init_nom_B_term2}
\label{Xhat0_Bstar_1}
Consider $0 < \epsilon_1 < 1$. Then, w.p. $1-\exp\left[nr -\epsilon_1^2mq/\mu^2\kappa^2\right]$
		\[
		\|\left(\X_+ - \E[\X_+]\right)\Bcheck{}^\top\|_F \leq C\epsilon_1\|\Xstar\|_F.
		\]
	\end{lemma}
	
	\begin{proof}[Proof of Lemma \ref{lem:init_nom_B_term2}]
The proof is quite similar to the previous one.
For a fixed $\W\in\S_{nr}$ we have,
		\begin{align*}
		&\langle \left(\X_+ - \E[\X_+]\right)\Bcheck{}^\top,~\W\rangle \\& = \frac{1}{m}\sum_\ik \left( \y_\ik (\a_\ik{}^\top\W\bcheck_{k})\indic_{\left\{  |\y_\ik|^2 \leq \frac{\tC(1+\epsilon_1)}{q}\|\Xstar\|_F^2 \right\} } - \E[.]\right)
		\end{align*}
where $\E[.]$ is the expected value of the first term. 
The summands are independent, zero mean, sub-Gaussian r.v.s with subGaussian norm less than $K_\ik \leq C \sqrt{\tC(1+\epsilon_1)}\|\Xstar\|_F\|\W\b_k\|/\sqrt{q}$. Thus, by applying the  sub-Gaussian Hoeffding inequality Theorem 2.6.2 of \cite{versh_book},
with $t=\epsilon_1 m \|\Xstar\|_F$, and using $\|\W\Bcheck\|_F = 1$, we can conclude that,
w.p. $1-\exp\left[-\epsilon_1^2mq/(C\mu^2\kappa^2)\right]$,
		$$
		\langle \left(\X_+ - \E[\X_+]\right)\Bcheck{}^\top,~\W\rangle\leq C \epsilon_1 \|\Xstar\|_F.
		$$
By Proposition \ref{epsnet_MW}, the above bound holds for all $\W \in \S_{nr}$ w.p. $1-\exp\left[nr -\epsilon_1^2mq/(C\mu^2\kappa^2)\right]$.
%
	\end{proof}

\subsection{Bounding the U* numerator term}
We bound  $ \| (\Xhat_0 - \E[\X_+])^\top \Ustar \|_F $ here. By triangle inequality, it is bounded by $\|\left(\Xhat_0 - \X_+ \right){}^\top\Ustar \|_F + \|\left(\Xhat_+ - \E[\X_+] \right){}^\top\Ustar \|_F$.

	\begin{lemma}
		\label{lem:init_term1}
\label{Xhat0_Ustar_2}
Assume that $\frac{1}{mq}\sum_\ik \y_\ik^2 \in [1\pm \epsilon_1]\|\Xstar\|_F^2/q$. Then, w.p. $1-\exp\left[qr-c\epsilon_1^2mq/\mu^2\kappa^2\right]$
		\[
		\|\left(\Xhat_0 - \X_+ \right){}^\top\Ustar \|_F \leq C\epsilon_1 \mu\kappa\|\Xstar\|_F.
		\]
	\end{lemma}

	\begin{proof}[Proof of Lemma \ref{lem:init_term1}]
The proof is similar to that of Lemmas \ref{Xhat0_2} and \ref{Xhat0_Bstar_2}.
		We have,
		$
		\|\left(\Xhat_0 - \X_+ \right){}^\top\Ustar\|_F = \max_{\W\in\S_{qr}} \langle \W,~ \left(\Xhat - \X_+ \right){}^\top\Ustar \rangle.
		$
		For a fixed $\W \in \S_{qr}$, using the same approach as in Lemma \ref{Xhat0_2}, and letting $\w_k$ be the $k$-th column of the $r \times q$ matrix $\W$,
		\begin{align*}
		&\langle \W,~ \left(\Xhat_0 - \X_+ \right){}^\top\Ustar \rangle \\&\qquad\leq \frac{1}{m}\sum_\ik \big|\y_\ik(\a_\ik{}^\top\Ustar\w_k) \big|\indic_{  \left\{ \frac{\tC}{mq}\sum_\ik  |\y_\ik|^2 \leq|\y_\ik|^2 \leq \frac{\tC(1+\epsilon_1)}{q}\|\Xstar\|_F^2 \right\} },\\
		&\qquad\leq \frac{1}{m}\sum_\ik \big|\y_\ik(\a_\ik{}^\top\Ustar\w_k) \big|\indic_{  \left\{ |\y_\ik|^2 \in [1\pm\epsilon_1]\frac{\tC}{q}\|\Xstar\|_F^2 \right\} }.
		\end{align*}
The summands are now mutually independent sub-Gaussian r.v.s with norm $K_\ik \leq \sqrt{\tC(1+\epsilon_1)}\|\w_k\| \|\Xstar\|_F/\sqrt{q}$. Thus, we can apply the  sub-Gaussian Hoeffding inequality Theorem 2.6.2 of \cite{versh_book}, to conclude that,
for a fixed $\W \in \S_{qr}$, w.p. $1-\exp\left[-c\epsilon_1^2mq/\mu^2\kappa^2\right]$,
		\begin{align*}
		&\frac{1}{m}\sum_\ik \big|\y_\ik(\a_\ik{}^\top\Ustar\w_k) \big|\indic_{  \left\{ |\y_\ik|^2 \in [1\pm\epsilon_1]\frac{\tC}{q}\|\Xstar\|_F^2 \right\} } \\&\leq \epsilon_1 \|\Xstar\|_F + \frac{1}{m}\sum_{k}\E\left[\big|\y_\ik(\a_\ik{}^\top\Ustar\w_k) \big|\indic_{  \left\{ |\y_\ik|^2 \in [1\pm\epsilon_1]\frac{\tC}{q}\|\Xstar\|_F^2 \right\} }\right]
		\end{align*}
By Claim \ref{claim:expect_init}, and using  $\sum_{k} \|\w_k \|/\sqrt{q} \leq \sqrt{\sum_{k}\|\w_k \|^2}~\sqrt{\sum_k 1/q} = 1$,
		\begin{align*}
		&\frac{1}{m}\sum_{k}\E\left[\big|\y_\ik(\a_\ik{}^\top\Ustar\w_k) \big|\indic_{  \left\{ |\y_\ik|^2 \in [1\pm\epsilon_1]\frac{\tC}{q}\|\Xstar\|_F^2 \right\} }\right] \\&\leq \frac{1}{m}\sum_\ik \epsilon_1 \|\w_k\| \sqrt{\tC(1+\epsilon_1)/q}\|\Xstar\|_F,\\
		&\leq C\epsilon_1\mu\kappa \|\Xstar\|_F,
		\end{align*}
By Proposition \ref{epsnet_MW} (epsilon net argument), the bound holds for all unit norm $\W$ w.p. $1-\exp\left[ qr-c\epsilon_1^2mq/\mu^2\kappa^2\right]$.
	\end{proof}
	
	\begin{lemma}
		\label{lem:init_term1_2}
\label{Xhat0_Ustar_1}
	Consider $0 < \epsilon_1 < 1$. Then, w.p. $1-\exp\left[ qr -\epsilon_1^2/mq\mu^2\kappa^2\right]$
		\[
		\|\left(\X_+ - \E[\X_+]\right){}^\top\Ustar\|_F \leq C\epsilon_1\|\Xstar\|_F.
		\]
	\end{lemma}
	
	\begin{proof}[Proof of Lemma \ref{lem:init_term1_2}]
For fixed $\W \in \mathcal{S}_{qr}$,
		\begin{align*}
		&\trace\left(\W{}^\top \left(\X_+ - \E[\X_+]\right){}^\top\Ustar\right) \\
		&\qquad= \frac{1}{m}\sum_\ik \left( \y_\ik(\a_\ik{}^\top\Ustar\w_k)\indic_{\left\{|\y_\ik|^2 \leq \frac{\tC(1+\epsilon_1)}{q}\|\Xstar\|_F^2 \right\} } \right.\\&\qquad-\left. \E\left[\y_\ik(\a_\ik{}^\top\Ustar\w_k)\indic_{\left\{|\y_\ik|^2 \leq \frac{\tC(1+\epsilon_1)}{q}\|\Xstar\|_F^2 \right\} }\right] \right)
		\end{align*}
The summands are independent zero mean sub-Gaussian r.v.s with norm less than $K_\ik \leq \sqrt{\tC(1+\epsilon_1)}\|\Xstar\|_F\|\w_k\|/\sqrt{q}$. Thus, by applying the sub-Gaussian Hoeffding inequality Theorem 2.6.2 of \cite{versh_book},
with $t=\epsilon_1 m \|\Xstar\|_F$, we can conclude that,
for a fixed $\W \in \mathcal{S}_{qr}$, w.p. $1-\exp\left[-\epsilon_1^2mq/C\mu^2\kappa^2\right]$,
		\[
		\trace\left(\W{}^\top\left(\X_+ - \E[\X_+]\right){}^\top\Ustar \right) \leq \epsilon_1 \|\Xstar\|_F.
		\]
By Proposition \ref{epsnet_MW} (epsilon net argument), the bound holds for all unit norm $\W$ w.p. $1-\exp\left[qr -\epsilon_1^2mq/C\mu^2\kappa^2\right]$.
	\end{proof}


\subsection{Proof of Claim \ref{EXhat0_Xplus}}

\begin{proof}
We can write $\xstar = \|\xstar\| \Q  \e_1$ where $\Q$ is a unitary matrix with first column proportional to $\xstar_k$. We need to bound
\begin{align*}
&\E[ \|\xstar\| \cdot | (\a^\top \Q \e_1) (\a^\top\Q \Q^\top \z)  | \indic_{\{ \|\xstar\|^2 | \a^\top \Q e_1|^2 \in [1\pm \epsilon] \alpha } \}  ]
\\&= \|\xstar\|  \cdot \|\z\|  \cdot \E[  | \tilde\a(1) \tilde\a^\top \bar\z_Q  | \indic_{ \{ |\tilde\a(1)|^2 \in [1\pm \epsilon] \beta^2 } \} ]
\end{align*}
where $\bar\z_Q:= \Q^\top \z / \|\z\|$, $\tilde\a:= \Q^\top\a$ and $\beta: = \sqrt{\alpha}/\|\xstar\|$. Since $\Q$ is unitary and $\a$ Gaussian, thus $\tilde\a$ has the same distribution as $\a$.
Let $\tilde\a(1)$ be its first entry and $\tilde\a(\mathrm{rest})$ be the $(n-1)$-length vector with the rest of the $n-1$ entries and similarly for $\bar\z_Q$. Then, $ \tilde\a^\top \bar\z_Q =\tilde\a(1)\cdot\bar\z_Q(1) + \tilde\a(\mathrm{rest})^\top \bar\z_Q(\mathrm{rest})$.
Since $\tilde\a(1)$ and $\tilde\a(\mathrm{rest})$ are independent, 
\begin{align*}
&\E[  | \tilde\a(1) \tilde\a^\top \bar\z_Q  | \indic_{|\tilde\a(1)|^2 \in [1\pm \epsilon] \beta^2 } ]
\\&\le  |\bar\z_Q(1)| \E[ | \tilde\a(1)^2  | \indic_{|\tilde\a(1)|^2 \in [1\pm \epsilon] \beta^2 } ] \\&\qquad+  \E[ | \tilde\a (\mathrm{rest})^\top \bar\z_Q(\mathrm{rest})  |  ] \ \E[ | \tilde\a(1)| \indic_{|\tilde\a(1)|^2 \in [1\pm \epsilon] \beta^2 }]  \\
\le &  \E[ | \tilde\a(1)^2   | \indic_{|\tilde\a(1)|^2 \in [1\pm \epsilon] \beta^2 } ] + 2 \E[ | \tilde\a(1)| \indic_{|\tilde\a(1)|^2 \in [1\pm \epsilon] \beta^2 }] \\
\le & \epsilon \beta  + 2 \epsilon \beta  = 3\epsilon \beta = C \epsilon \frac{\sqrt{\alpha}}{  \|\xstar\|}.
\end{align*}
The second inequality used the facts that (i) $|\bar\z_Q(1)| \le \|\bar\z_Q\|=1$ by definition and (ii) $\zeta:=\tilde\a (\mathrm{rest})^\top \bar\z_Q(\mathrm{rest})$ is a scalar standard Gaussian r.v. and so $\E[|\zeta|] \le 2$.
The third one relies on the following two bounds:
\ben
\item
\begin{align*}
		&\E\left[ |\a(1)|^2\indic_{ \left\{ |\a(1) |^2 \in [1\pm\epsilon]\beta^2  \right\}}   \right] \\&= \frac{2}{\sqrt{2\pi}}\int_{\sqrt{1-\epsilon}\beta}^{\sqrt{1+\epsilon}\beta} z^2\exp(-z^2/2)dz,\\
		&\leq \frac{2e^{-1/2}}{\sqrt{2\pi}}\int_{\sqrt{1-\epsilon}\beta}^{\sqrt{1+\epsilon}\beta} dz \leq  \frac{2e^{-1/2}}{\sqrt{2\pi}} \epsilon\beta\leq \epsilon\beta/3
		\end{align*}
		where we used the facts that $z^2\exp(-z^2/2) \leq \exp(-1/2)$ for all $z \in\Re$; $\sqrt{1-\epsilon} \geq 1-\epsilon/2$ and $\sqrt{1+\epsilon} \leq 1+\epsilon/2$ for $0 < \epsilon <1$.
\item Similarly, we can show that
		\begin{align*}
		&\E\left[ |\a(1)|\indic_{\{ |\a(1)|^2 \in [1\pm\epsilon]\beta^2  \}}   \right] \\&= \frac{2}{\sqrt{2\pi}}\int_{\sqrt{1-\epsilon}\beta}^{\sqrt{1+\epsilon}\beta} z\exp(-z^2/2)dz,\\
		&\leq \frac{2e^{-1/2}}{\sqrt{2\pi}}\int_{\sqrt{1-\epsilon}\beta}^{\sqrt{1+\epsilon}\beta} dz = \frac{2e^{-1/2}}{\sqrt{2\pi}} \epsilon\beta\leq \epsilon\beta/3
		\end{align*}
\een
The claim follows by combining the two equations given above.
\end{proof}

\bibliographystyle{IEEEtran}
\bibliography{./tipnewpfmt_kfcsfullpap}

\section*{Author Biographies}

{\bf Seyedehsara Nayer (Email: sarana@iastate.edu)} recently completed her Ph.D. in ECE at Iowa State University. She has an M.S. from Sharif University in Iran.  She works as a Senior Engineer at ASML in Santa Clara, CA.
Her research interests are around various aspects of information science and focuses on Signal Processing, and Statistical Machine Learning.


{\bf Namrata Vaswani (Email: namrata@iastate.edu)} received a B.Tech from IIT-Delhi in India in 1999 and a Ph.D. from the University of Maryland, College Park in 2004, both in Electrical Engineering. Since Fall 2005, she has been with the Iowa State University where she is currently the Anderlik Professor of Electrical and Computer Engineering. Her research interests lie in a data science, with a particular focus on Statistical Machine Learning and Signal Processing. She has served two terms as an Associate Editor for the IEEE Transactions on Signal Processing; as a lead guest-editor for a 2018 Proceedings of the IEEE Special Issue (Rethinking PCA for modern datasets); and as an Area Editor for the IEEE Signal Processing Magazine (2018-2020). Vaswani is a recipient of the Iowa State Early Career Engineering Faculty Research Award (2014), the Iowa State University Mid-Career Achievement in Research Award (2019) and University of Maryland’s ECE Distinguished Alumni Award (2019). She also received the 2014 IEEE Signal Processing Society Best Paper Award for her 2010 IEEE Transactions on Signal Processing paper co-authored with her student Wei Lu on “Modified-CS: Modifying compressive sensing for problems with partially known support”. She is a Fellow of the IEEE Fellow (class of 2019).

%


\end{document}